%% file: main.tex
\documentclass[sigconf]{acmart}

\usepackage{graphicx,mathtools,kantlipsum}
\usepackage[normalem]{ulem}
\usepackage{url}
\usepackage{amsmath,amsfonts}
\usepackage{listings}
\usepackage{color}
\usepackage{courier}
\usepackage{float}
\usepackage{makecell}
\usepackage[linesnumbered,vlined,ruled]{algorithm2e}
\usepackage[noend]{algpseudocode}
\usepackage{todonotes}
\usepackage{booktabs}
\usepackage{caption}
\usepackage{subcaption}
\usepackage{hyperref}
\usepackage[capitalize]{cleveref}
\usepackage{balance}
\usepackage{enumitem}
\usepackage{bm}
\usepackage{multicol}
\usepackage{multirow}
\usepackage{xspace}
\usepackage{tabularx}

\newcommand{\SUN}[1]{{\textcolor{black}{ #1}}}
\AtBeginDocument{%
    \providecommand\BibTeX{{%
        \normalfont B\kern-0.5em{\scshape i\kern-0.25em b}\kern-0.8em\TeX}}}

\copyrightyear{2021}
\acmYear{2021}
\setcopyright{none}
\acmConference[SIGMOD '21]{Proceedings of the 2021 International Conference on Management of Data}{June 20--25, 2021}{Virtual Event, China}
\acmBooktitle{Proceedings of the 2021 International Conference on Management of Data (SIGMOD '21), June 20--25, 2021, Virtual Event, China}
\acmDOI{10.1145/3448016.3457290}
\acmISBN{978-1-4503-8343-1/21/06}

\renewcommand\footnotetextcopyrightpermission[1]{}
\begin{document}
\fancyhead{}
\title{PathEnum: Towards Real-Time Hop-Constrained s-t Path Enumeration (Complete Version)}



\author{Shixuan Sun}
\affiliation{%
  \institution{National University of Singapore}
  \country{Singapore}}
\email{sunsx@comp.nus.edu.sg}

\author{Yuhang Chen}
\affiliation{%
  \institution{National University of Singapore}
  \country{Singapore}}
\email{yuhangc@comp.nus.edu.sg}

\author{Bingsheng He}
\affiliation{%
  \institution{National University of Singapore}
  \country{Singapore}}
\email{hebs@comp.nus.edu.sg}

\author{Bryan Hooi}
\affiliation{%
  \institution{National University of Singapore}
  \country{Singapore}}
\email{dcsbhk@nus.edu.sg}


\input{0_abstract}

\maketitle

\input{1_introduction}
\input{2_backgroud}
\input{3_algorithm_overview}

\input{4_build_index}

\input{5_search_on_index}
\input{6_search_optimization}
\input{7_experiments}

\input{8_conclusion}

\bibliographystyle{ACM-Reference-Format}
\bibliography{reference}

\input{9_appendix}

\end{document}

%% file: 0_abstract.tex
\begin{abstract}
We study the hop-constrained \emph{s-t} path enumeration (\textbf{HcPE}) problem, which takes a graph $G$, two distinct vertices $s,t$ and a hop constraint $k$ as input, and outputs all paths from $s$ to $t$ whose length is at most $k$. The state-of-the-art algorithms suffer from severe performance issues caused by the costly pruning operations during enumeration for the workloads with the large search space. Consequently, these algorithms hardly meet the real-time constraints of many online applications. In this paper, we propose PathEnum, an efficient index-based algorithm towards real-time HcPE. For an input query, PathEnum first builds a light-weight index aiming to reduce the number of edges involved in the enumeration, and develops efficient index-based approaches for enumeration, one based on depth-first search and the other based on joins. We further develop a query optimizer based on a join-based cost model to optimize the search order. We conduct experiments with 15 real-world graphs. Our experiment results show that PathEnum outperforms the state-of-the-art approaches by orders of magnitude in terms of the query time, throughput and response time.
\end{abstract}




%% file: 1_introduction.tex
\section{Introduction} \label{sec:introduction}


Because paths are widely used to measure the relationship between vertices, \textbf{HcPE} serves as an important building brick in a number of emerging
real-world applications. \SUN{Furthermore, HcPE can be easily extended with variant constraints to capture complexities of these applications.}
For example:

\emph{(1) Detecting Money Laundering \cite{li2020flowscope,force2013fatf,jedrzejek2009graph}.} Money laundering is the illegal process of injecting "dirty" money 
into the legitimate financial system, typically by using a bank's services to move illegal money from source accounts into 
destination accounts through a series of transactions. We can construct a graph by representing bank accounts as vertices and transactions
as edges. \SUN{The report \cite{force2013fatf} lists a number of known "red flag indicators" which are regarded as indicative of money laundering.
For example, the use of multiple bank accounts as well as that of intermediaries without good reasons is a red flag, which is also used
in \cite{li2020flowscope,jedrzejek2009graph}. They observe many instances of money laundering along short flow paths (e.g. two-hop),
noting that longer paths can increase costs for the fraudsters. As such, this flag can be detected by enumerating hop-constrained paths between two target accounts.
Moreover, banks or regulatory bodies may designate certain factors as risky (e.g., capital from foreign companies). Since a single risk
factor may not be conclusive on its own, we want to find transactions exhibiting a certain level of total risk. In that case, we associate each
edge with a weight representing the risk factor and extend HcPE by requiring that the accumulative value of weights on edges in a path is above a threshold.}

\emph{(2) E-Commerce Merchant Fraud Detection \cite{qiu2018real}.} The activities of online shopping can be modeled as a graph in
which vertices are individual users (e.g., sellers and buyers) and edges are online transactions (e.g., online payment and shipment
of goods). In order to increase the popularity of products, some sellers create fake transactions. In brief, the entire process generates cycles in the graph. Therefore,
the cycles triggered by new edges are strong indications of potential fraud. In the applications, \cite{qiu2018real} enumerates
the cycles within a small hop constraint (e.g, $k = 6$) because a large hop constraint can result in a huge number of
results causing massive false alarms. This also suggests the usage of hop constraints. We can issue a query $q(v', v, k - 1)$ to find paths from $v'$ to $v$
to enumerate cycles triggered by the new edge $e(v, v')$. \SUN{Moreover, we may also impose constraints based on attributes of edges (e.g., monitor
fake transactions with particular types of user activities \cite{qiu2018real}). Then, we can express the constraints as predicates on edges, and
extend HcPE by requiring that each edge in a path satisfies conditions in predicates.}

\emph{(3) Knowledge Graph Completion \cite{wang2020entity}.} A variety of applications such as recommendation systems,
search and question answering depend on knowledge graphs (KGs). Because KGs are generally incomplete,
the problem of knowledge graph completion, which aims to predict missing relations in KGs, is very important.
In particular, paths between two entities indicate the relationship between them, and the knowledge graph completion methods
generally use these paths to train models to predict the relationship. Previous work has observed that entities connected by
many short paths have a higher tendency to be related, e.g., \cite{shiralkar2017finding,shi2016discriminative},
suggesting the utility of hop constraints in this setting. \SUN{Furthermore, real-world applications may require that the paths satisfy the constraints
on the sequence of actions (e.g., the sequence "write->mention"). In that case, each edge label represents an action, and
we extend HcPE by requiring that the label sequence of each path meets the constraint on the sequence of actions.}

Due to its importance, the HcPE problem has recently received significant interest.
Existing approaches \cite{peng2019towards,grossi2018efficient,rizzi2014efficiently} focus on designing the
\emph{polynomial delay} algorithms such that the time between finding two successive results is bounded by a polynomial
function of the input size in the worst case \cite{johnson1988generating}. They adopt the backtracking method to recursively
enumerate paths from the source to the target. To achieve polynomial delay, they introduce
pruning rules at each recursive call to reduce the invalid search space, for example, performing a single source
shortest path query from the target to update the distance between each vertex and the target \cite{rizzi2014efficiently}.
Benefiting from the pruning strategies, the delay per output is within $O(k \times |E(G)|)$ time where $k$ is the length
constraint and $|E(G)|$ is the number of edges in $G$.

Despite their theoretical guarantees, we find that these algorithms suffer from serious performance issues in practice.
For the workloads with the large search space, the pruning at each step is expensive,
to the extent that the pruning overhead can offset its benefits of reducing the search space. This fails to satisfy the requirement from many
applications, especially the online scenarios \cite{qiu2018real,kim2018turboflux} with rigid real-time requirement on query
time.


In this paper, we propose \textbf{PathEnum},
an efficient approach to the HcPE problem. In contrast to existing algorithms \cite{peng2019towards,grossi2018efficient,rizzi2014efficiently}
that conduct pruning operations during the enumeration, the key design principle of PathEnum is to develop a
light-weight index for the input query so that the index can be used to keep each step in the enumeration simple and efficient.
 
We first design a join-based model to abstract the HcPE problem, and analyze two key performance factors in the model,
which are the number of edges involved in the enumeration and the order of enumerating results. Next, we develop an index-based approach to evaluate the query. Specifically, given a graph $G$ and a query $q(s, t, k)$,
we first build a light-weight index $\mathcal{I}$ at runtime, which is constructed based on the \emph{distance} (i.e., the length of the
shortest path) between each vertex to $s$ and $t$. The index is used to reduce the number of edges accessed during the enumeration.
Given a vertex $v$ and an integer $b$, we can quickly retrieve the neighbors $v'$ of $v$ such that the distance
from $v'$ to $t$ (or from $s$ to $v'$) is bounded by $b$ from $\mathcal{I}$. The time complexity of
constructing $\mathcal{I}$ is $O(|E(G)| + |V(G)|)$.

We further develop a cost-based query optimizer to optimize the order of enumerating results. Specifically, we develop a depth-first
search based method and a join-based method to enumerate the results based on $\mathcal{I}$. The DFS-based method recursively
extends the partial result by one vertex at a step to enumerate all results, whereas the join-based method first cuts the query
into two sub-queries, and then evaluates them with the DFS-based method, respectively, and finally join the intermediate results
of the two sub-queries. These two approaches can generate different number of partial results during the enumeration.
The query optimizer selects the method with a lower cost to evaluate the query.

We conduct extensive experiments with 15 real-world datasets. The experiment results show that 
PathEnum provides speedups of 1.9 to 240.7 times over the state-of-the-art method \cite{peng2019towards} in terms of query time and
14.2 to 358.5 times in terms of response time. In summary, we make the following contributions in this paper.

\begin{itemize}[noitemsep,topsep=0pt]
    \item We study the hop-constrained \emph{s-t} path enumeration problem, and propose \textbf{PathEnum},
    an efficient solution towards practical and real-time enumeration in many online applications.
    \item Different from existing backtracking solutions, PathEnum is an efficient index-based approach for HcPE. For each query, we first develop an efficient light-weight indexing method to prune the vertices involved in the subsequent enumeration. During the enumeration process, we design two index-based approaches, and an effective join order optimization method to reduce the search space.
    \item We conduct extensive experiments with a variety of workloads, and demonstrate PathEnum significantly outperforms state-of-the-art algorithms.
\end{itemize}

Supplement results are presented in the appendix. Our source code is publicly available at GitHub \cite{path_enum_source_code}.

\textbf{Paper Organization.} We introduce the preliminaries and related work in Section 2. In Section 3, we formulate the HcPE problem in a join model and give an overview of PathEnum. We design a light-weight index, and develop index based enumeration approaches in Sections 4 and 5, respectively. We optimize the search order in Section 6. We present the experiment results in Section 7 and conclude in Section 8.

%% file: 2_backgroud.tex
\section{Background and Related Work} \label{sec:background}


\subsection{Preliminaries} \label{sec:preliminaries}

$G = (V, E)$ denotes a directed graph where $V$ is a set of vertices
and $E \subseteq V \times V$ is a set of edges. $e(v, v')$ denotes a directed edge from the
vertex $v$ to the vertex $v'$. $N(v) = \{v'|e(v,v')\in E\}$ represents
the outgoing neighbors of $v$, and $d(v)$ denotes the out degree of $v$, i.e., $d(v) = |N(v)|$.
By default, the neighbors of $v$ refer to the outgoing neighbors. $G ^ r$ represents the graph
obtained by reversing the direction of each edge in $G$. Given two vertices $v$ and $v'$,
the \emph{distance} from $v$ to $v'$, denoted by $S(v, v'|G)$, is the length of the shortest path
from $v$ to $v'$ in $G$. $G - \{v\}$ represents the graph that removes $v$ as well as edges connecting with $v$ from $G$.

A \emph{walk} $W$ is a sequence of vertices $(v_0, v_1, ..., v_{l})$ such that
$\forall 1 \leqslant i \leqslant l, e(v_{i - 1}, v_i) \in E$.  $|W|$ denotes the number of vertices in $W$,
while $L(W)$ represents the number of edges in $W$. Therefore, $L(W) = |W| - 1$ when $W$ is not empty.
$W[i]$ denotes the $i$th vertex in $W$ where $0 \leqslant i \leqslant |W| - 1$.
Given two distinct vertices $s$ and $t$, we define a walk from $s$ to $t$ in Definition \ref{def:walk_st}.
$\mathcal{W}(s, t, k, G)$ represents all walks $W$ from $s$ to $t$ in $G$ that satisfy $L(W) \leqslant k$.
A \emph{path} $P$ is a walk in which all vertices are distinct. Then, a path from $s$ to $t$ is a walk from $s$ to $t$ in which all vertices are distinct.
$\mathcal{P}(s, t, k, G)$ denotes all paths $P$ from $s$ to $t$ such that $L(P) \leqslant k$. Apparently, given $P \in \mathcal{P}(s, t, k, G)$, $P$ belongs to
$\mathcal{W}(s, t, k, G)$. \SUN{Table \ref{tb:notation} summarizes notations frequently
used in this paper.}

\begin{definition} \label{def:walk_st}
    A walk from $s$ to $t$ is a walk $W$ such that (1) $W[0] = s \wedge W[|W| - 1] = t$; and (2) $\forall 0 < i < |W| - 1$, $W[i] \notin \{s, t\}$.
\end{definition}

\begin{table}[t]
	\small
	\centering
	\setlength{\abovecaptionskip}{0pt}
	\setlength{\belowcaptionskip}{0pt}
	\caption{\SUN{A summary of notations frequently used.}}
	\label{tb:notation}
\begin{tabular}{l|l}
	\hline
	\textbf{Notations}        & \textbf{Descriptions}                                                                                                               \\ \hline
	$s,t$ and $k$             & source, target and length constraint                                                                                                \\ \hline
	$G$ and $q(s,t,k)$        & graph and HcPE query                                                                                                                \\ \hline
	$Q$ and $R$               & join query and relation                                                                                                             \\ \hline
	$V(G)$ and $E(G)$         & vertex and edge sets of $G$                                                                                                         \\ \hline
	$e(v, v')$                & edge between $v$ and $v'$                                                                                                           \\ \hline
	$d(v)$ and $N(v)$         & degree and neighbors of $v$                                                                                                         \\ \hline
	$P, W$ and $M$            & path, walk and partial result                                                                                                       \\ \hline
	$L(P), L(W), L(M)$   & number of edges in $P, W$ and $M$                                                                                                   \\ \hline
	$\mathcal{P}(s, t, k, G)$ & paths $P$ from $s$ to $t$ with $L(P) \leqslant k$                                                                                   \\ \hline
	$\mathcal{W}(s,t,k,G)$    & walks $W$ from $s$ to $t$ with $L(W) \leqslant k$                                                                                    \\ \hline
	$\delta_P$ and $\delta_W$ & $|\mathcal{P}(s, t, k, G)|$ and $|\mathcal{W}(s, t, k, G)|$                                                                         \\ \hline
	$S(v, v'|G)$              & distance between $v$ and $v'$ in $G$                                                                                                \\ \hline
	$\mathcal{I}$             & the light-weight index                                                                                                              \\ \hline
	$\mathcal{I}(i)$          & \begin{tabular}[c]{@{}l@{}}vertices $v$ satisfying $S(s, v|G-\{t\})\leqslant i$\\ and $S(v, t|G-\{s\})\leqslant k - i$\end{tabular} \\ \hline
	$\mathcal{I}_t(v,b)$      & neighbors $v'$ of $v$ satisfying $S(v', t|G-\{s\})\leqslant b$                                                                      \\ \hline
\end{tabular}
\end{table}

\textbf{Problem Statement.} Given $G = (V, E)$, two distinct vertices $s,t$ and a hop constraint $k$,
the \emph{hop-constrained s-t path enumeration} (HcPE) problem aims to find all paths in $\mathcal{P}(s, t, k, G)$.
The query is denoted by $q(s, t, k)$. We assume that $k \geqslant 2$ in this paper.

\subsection{State-of-the-art Approaches}

Algorithm \ref{algo:basic_method} illustrates a generic depth-first search based framework
to find $\mathcal{P}(s, t, k, G)$. It adopts the backtracking strategy. $M$ stores a sequence of vertices, which initially
contains $s$ (Line 1). Line 5 emits $M$ when the last vertex of $M$ is $t$. Otherwise, Lines 6-8 loop over $N(v)$ to extend $M$.
Particularly, $B(v')$ stores the distance from $v'$ to $t$. Before the enumeration, we can initialize it
by performing a breadth-first search from $t$ along $G ^ r$. Line 7 checks (1) whether $v'$ belongs to $M$; and (2) whether
we can extend $M$ by adding $v'$ to generate a path satisfying the hop constraint. If $v'$ passes the check, then we add
$v'$ to $M$ and continue the search. Otherwise, we skip $v'$. Therefore, the \emph{Search} procedure can be viewed as performing
a depth-first search in a search tree where each node is a partial result $M$ and each edge is the action of adding a vertex to $M$. 

Existing approaches \cite{peng2019towards,grossi2018efficient,rizzi2014efficiently} adopt the same backtracking strategy
as Algorithm \ref{algo:basic_method}, but introduce different pruning techniques to achieve polynomial
delay. They update $B(v)$ for a vertex during the enumeration because a path contains no duplicate vertices and
the update of $M$ can break the shortest path from $v$ to $t$. Peng et al. \cite{peng2019towards}
designed a \emph{barrier}-based method, which dynamically maintains the distance from each vertex to $t$. Initially,
they set the barrier for each $v \in V(G)$ as $S(v, t|G)$. During the enumeration, if they find that a sub-tree rooted
at a node in the search tree contains no result, then they will increase the barrier to
avoid falling into the same sub-tree again. T-DFS \cite{rizzi2014efficiently} and T-DFS2 \cite{grossi2018efficient}
are two theoretical works. They achieve polynomial delay by ensuring that each search branch in the search tree leads to a result.
For example, before extending $M$ by adding $v'$ in Algorithm \ref{algo:basic_method}, T-DFS checks whether there is a shortest
path from $v'$ to $t$ without vertices in $M$ whose length is bounded by $k - L(M) - 1$. Although all the three
algorithms achieve $O(k \times |E(G)|)$ polynomial delay, Peng et al. showed that their method runs much faster than T-DFS
and T-DFS2 in practice because their pruning strategy incurs lower overhead \cite{peng2019towards}. HPI \cite{qiu2018real}
enumerates hop-constrained cycles triggered by incoming edges in dynamic graphs. It builds an index maintaining paths
between vertices with a high degree to reduce the cost of enumeration. However, the index can consume a large amount of memory
due to the exponential number of paths between each pair of these vertices.

\setlength{\textfloatsep}{0pt}
\begin{algorithm}[t]
    \footnotesize
	\caption{Generic DFS based Framework}
	\label{algo:basic_method}
	\SetKwFunction{Search}{Search}
	\SetKwProg{proc}{Procedure}{}{}
	 \KwIn{a graph $G$, two distinct vertices $s, t$, hop constraint $k$\;}
	 \KwOut{all $k$ hop-constrained paths from $s$ to $t$\;}
	 $M \leftarrow (s)$\;
	 \Search{$t, k, M$}\;
	 \proc{\Search{$t, k, M$}}{
	      $v \leftarrow $ the last vertex in $M$\;
	      \lIf{$v = t$}{$emit(M)$, \KwRet}
	      \ForEach{$v' \in N(v)$}{
	            \If{$v' \notin M$ and $L(M) + 1 + B(v') \leqslant k$}{
	                \Search{$t, k, M \cup \{v'\}$}\;
	            }
	      }
	 }
\end{algorithm}

\subsection{Other Related Work} \label{sec:related_work}

\textbf{\emph{s-t} Path (or Cycle) Enumeration.} Another kind of algorithms \cite{bohmova2018computing,nishino2017compiling,yasuda2017fast}
focus on developing construction methods to compile \emph{s-t} paths into a representation structure such that these paths can be quickly listed without explicitly
storing each individual result. These algorithms can only handle graphs with hundred vertices because compiling
\emph{s-t} paths of large graphs can consume a large amount of memory. Enumerating all \emph{s-t} paths (or cycles)
without the hop constraint is a classical problem \cite{tarjan1973enumeration,johnson1975finding,birmele2013optimal,kumar20182scent}. However,
these algorithms cannot be easily extended to the scenarios with the hop constraint because their enumeration procedure does
not consider the impact of the hop constraint. Additionally, there are also a variety of works \cite{bhattacharya2020improved,haeupler2012incremental,bender2015new}
that focus on detecting the existence of cycles in dynamic graphs instead of enumerating the results.

\textbf{Top-K Shortest Path Enumeration.} We can evaluate a query $q(s, t, k)$ with the Top-K shortest path
algorithms \cite{yen1971finding,eppstein1998finding,gao2010fast,chang2015efficiently,singh2015implementation,martins2003new}.
In particular, we set $K$ as a sufficient large value and terminate the query when the length of results is greater than $k$. Despite that these algorithms can find
the results of $q(s, t, k)$, they enumerate results along the ascending order of the length of results, which is unnecessary
for the HcPE problem and incurs overhead.

\SUN{\textbf{Subgraph Matching.} Given a data graph and a query graph, subgraph matching finds all embeddings in the data graph that are identical to
the query graph \cite{sun2020memory,lai2019distributed}. Existing graph database systems such as EmptyHeaded \cite{aberger2017emptyheaded}
and GraphFlow \cite{mhedhbi2019optimizing} enumerate all results by performing self-joins on $G$, and propose variant join plan optimization methods
in which the cardinality estimation plays an important role. Existing estimation methods \cite{park2020g} work on input relations of the
join query \cite{li2016wander} or catalogs \cite{mhedhbi2019optimizing} that are built in an offline preprocessing step and summarize the global statistics of
$G$. For example, given $G$, GraphFlow uses sampling methods to build a catalog by collecting the number of subgraphs with some specific structures appearing
in $G$. Given a query, GraphFlow optimizes the join plan by considering different plans of constructing the query graph from its subgraphs, estimating
the cost (e.g., the number of partial results) based on the catalog and selecting the plan with the minimum cost. GraphFlow evaluates
the query according to the plan and adopts the intersection caching
to reduce the cost of set intersections. In summary, the query optimizer and the computation of existing systems are optimized for reducing the cost of finding
all subgraphs in $G$ with a specific structure (e.g., a path).}

\SUN{In contrast, the HcPE problem targets at paths from $s$ to $t$ that satisfy the length constraint. Moreover, our method evaluates the query on
a query-dependent index, which is built online for each query based on distances to $s,t$, without building relations. The index rules out many
invalid candidates for the query in $G$. Our query optimizer as well as the
cardinality estimation method is designed specially to work with the index.}

\SUN{\textbf{Distance Queries.} A distance query asks the distance between two vertices in a graph, which receives a lot of research
interests \cite{potamias2009fast,akiba2013fast,jin2019pruned,cohen2003reachability,qiao2012approximate,cheng2009line}. Existing methods
such as the pruned landmark labeling \cite{akiba2013fast} construct an index in an offline preprocessing step to serve all queries, and
evaluate the query with the pre-computed results. The index records the distance to a set of vertices for each vertex in the graph,
which maintains the global statistics of $G$. They focus on balancing the cost of building indexes and query efficiency. In contrast,
the light-weight index proposed in this paper is query-dependent, which is built based on the distance to $s,t$ and maintains the local
statistics for the given query.
}

%% file: 3_algorithm_overview.tex
\section{Algorithm Overview}

In this section, we first propose a join-based model to the HcPE problem, and then give an overview of our \textbf{PathEnum}.

\subsection{A Join-based Model} \label{sec:join_based_model}

Although existing algorithms \cite{peng2019towards,grossi2018efficient,rizzi2014efficiently} provide comprehensive analysis
to the HcPE problem in terms of the time complexity, there lacks a method to model the practical computation cost of evaluating
a query. To reveal the problem, we formulate a HcPE query $q(s, t, k)$ on $G$ as a chain join $Q$.
The edge list $E(G)$ can be viewed as a binary relation $R(u, u') = \{(v, v') | e(v,v') \in E(G)\}$.
At first glance, the query $q(s, t, k)$ can be easily translated to a chain join $Q = R_1(u_0, u_1) \Join R_2(u_1, u_2) \Join \dots \Join R_k(u_{k - 1}, u_{k})$ where
$R_1 = \{(s, v)| e(s, v) \in E(G)\}$, $R_k = \{(v, t)|e(v, t) \in E(G)\}$ and $R_i = \{(v, v')|e(v, v') \in E(G)\}$ when
$1 < i < k$. For the ease of presentation, we use $Q$ to represent the results of evaluating $Q$ as well.
To obtain $\mathcal{P}(s, t, k, G)$, we first evaluate $Q$, and then eliminate the tuples $r \in Q$ that have duplicate vertices.
However, this method only returns the paths from $s$ to $t$ the length of which are exactly $k$.
Although we can solve this problem by launching $k$ chain join queries to compute paths with different lengths, this
approach incurs a large amount of redundant computations. To solve the problem,
we propose to generate relations of $Q$ as follows.

 \begin{figure}[t]\small
    \setlength{\abovecaptionskip}{0pt}
    \setlength{\belowcaptionskip}{0pt}
    \centering
    \begin{subfigure}{0.3\textwidth}
        \centering
        \includegraphics[scale=0.4]{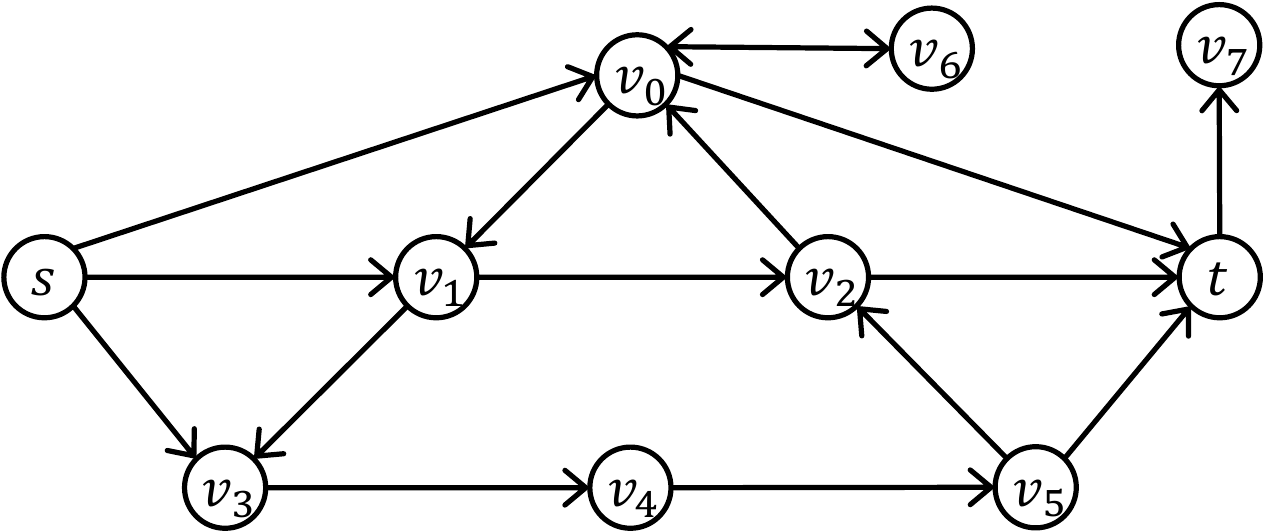}
        \caption{Graph $G$.}
        \label{fig:graph}
    \end{subfigure}	
    \begin{subfigure}{0.16\textwidth}
        \centering
        \includegraphics[scale=0.4]{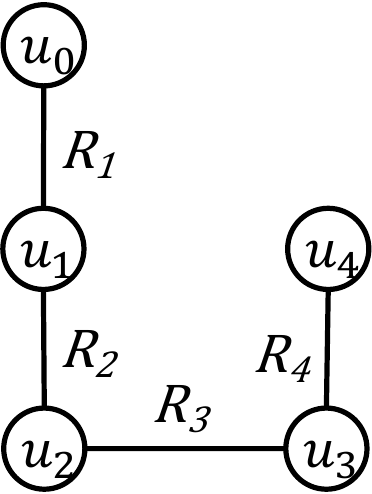}
        \caption{Join query $Q$.}
        \label{fig:join_query}
    \end{subfigure}	
    \caption{A query $q(s, t, 4)$ on the graph $G$.}
    \label{fig:example_graphs}
\end{figure}

\begin{enumerate}
    \item $R_1 = \{(s, v)|e(s, v) \in E(G) \}$ and $R_k = \{(v, t)|e(v, t) \in E(G) \wedge v \neq s\}$;
    \item $R_i = \{(v, v')|e(v, v') \in E(G - \{s\}) \wedge v \neq t \}$ when $1 < i < k$;
    \item $R_i = R_i \cup \{(t, t)\}$ for $1 < i \leqslant k$.
\end{enumerate}

The first two properties ensure that each tuple in $Q$ starts and ends at $s$ and $t$, while the third property
avoids eliminating the paths $P \in \mathcal{P}(s, t, k, G)$ that satisfy $L(P) < k$.
With the generation method, we can get Theorem \ref{theorem:correctness}. Example \ref{exmp:join_query} presents a running example.
\SUN{Due to space limit, the proof of Theorem \ref{theorem:correctness} and other propositions in this paper is presented in the appendix.}


\begin{theorem} \label{theorem:correctness}
    Evaluating $Q$ and eliminating tuples in $Q$ having duplicate vertices results in $\mathcal{P}(s, t, k, G)$.
\end{theorem}

\begin{example} \label{exmp:join_query}
	Given $G$ in Figure \ref{fig:graph} and a query $q(s, t, 4)$, the join query $Q$ is represented as a graph in Figure \ref{fig:join_query} where each edge is a relation and each node is an attribute. The relations of $Q$ are shown
	in Figure \ref{fig:initial_relations}. The path $(s, v_0, t)$ corresponds to the tuple $(s, v_0, t, t, t)$ in $Q$.
	$(s, v_0, v_6, v_0, t)$ belongs to $Q$ as well. However, it is a walk from $s$ to $t$, but not a path.
\end{example}

The cost function of evaluating $Q$ is shown in Equation \ref{eq:cost_model}, which is the total number of intermediate results
generated during the computation. If $Q$ is a basic relation, the cost is the size of the relation as we need to read it.
Otherwise, the cost is the sum of the cost of evaluating $Q_1$ and $Q_2$ and
the number of results of $Q$ where $Q = Q_1 \Join Q_2$. From the cost model, we can see that the cost of
evaluating a query is closely related to (1) the number of edges of involving in the enumeration; and (2)
the join order (i.e., search order) of evaluating the query.

\begin{equation} \label{eq:cost_model}
    T(Q) = \begin{cases}
        |R| & \text{If $Q$ is a base relation $R$.} \\
        |Q| + T(Q_1) + T(Q_2) & \text{If $Q = Q_1 \Join Q_2$.}
    \end{cases}
\end{equation}

\subsection{An Overview of PathEnum}

Figure \ref{fig:overview_pathenum} gives an overview of our PathEnum algorithm. Given a graph $G$ and a HcPE query $q(s, t, k)$,
we first build a light-weight index to reduce the number of edges involving in the subsequent search. Next, we generate a join order based on the statistics of the index.
We observe that the running time of different queries varies greatly because of the diverse size of the search space. Therefore, we optimize
the search order in two steps. In the first step, we use a preliminary cardinality estimator
to estimate the size of the search space. If the estimated size is small, then we directly invoke a depth-first search based method on the index to find results.
Otherwise, we optimize the join order with a full-fledged cardinality estimation. This method makes more accurate estimation than the
coarse-grained one at higher cost. However, the overhead is negligible when the
running time of queries is long. The query optimizer selects the method with a lower cost to evaluate the query.

\begin{figure}[t]\small
    \setlength{\abovecaptionskip}{0pt}
    \setlength{\belowcaptionskip}{0pt}
    \includegraphics[scale=0.43]{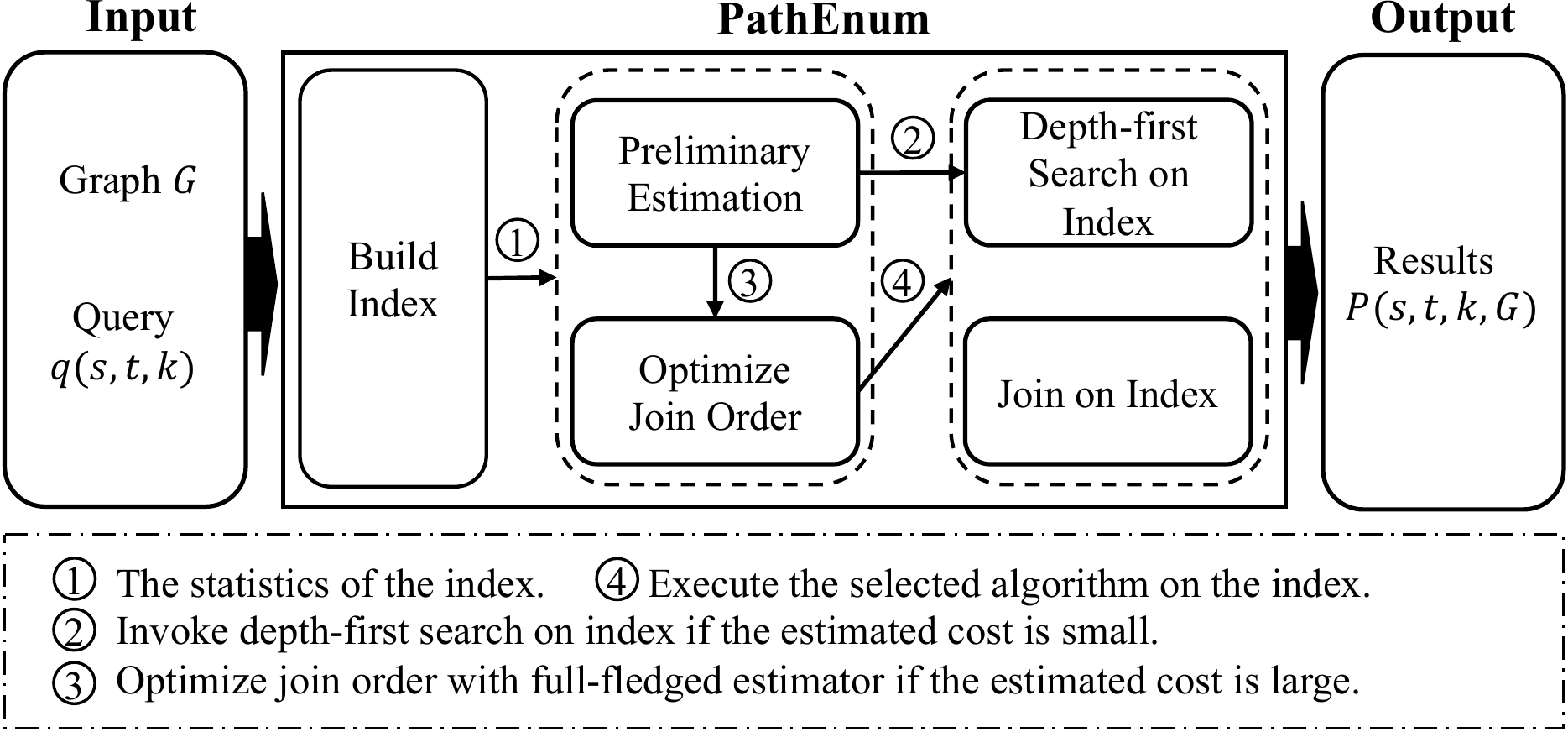}
    \centering
    \caption{An overview of PathEnum.}
    \label{fig:overview_pathenum}
\end{figure}

%% file: 4_build_index.tex
\section{Index Construction} \label{sec:build_index}


\subsection{Relation Construction}

Benefiting from the join-based model, we can evaluate $q(s, t, k)$ on $G$ as a chain join $Q$.
To reduce the cost, we can eliminate the \emph{dangling tuples}, i.e., the tuples not existing in any results, from each relation of $Q$
with the \emph{full reducer}, which is a classical dangling tuple elimination method in relational databases \cite{abiteboul1995foundations}.
For ease of understanding, we present the algorithm in graph context. Algorithm \ref{algo:build_relations} illustrates the details.
Lines 1-4 generate relations $R$ of $Q$ based on the construction method introduced in Section \ref{sec:join_based_model}.
Lines 5-12 remove dangling tuples from these relations. Specifically, Lines 5-8 prune relations $R_i$ along the increasing order of $i$.
Line 6 obtains the end vertex of edges in $R_i$. Next, Lines 7-8 remove edges $(v, v')$ in $R_{i + 1}$ such that $v$ does not belong to $C$.
After that, Lines 9-12 filter relations $R_i$ along the decreasing order of $i$ with the same method as Lines 5-8. Finally, Line 13 returns the relations.
The following is a running example.

\setlength{\textfloatsep}{0pt}
\begin{algorithm}[t]
    \footnotesize
	\caption{Build Relations}
	\label{algo:build_relations}
	\KwIn{a graph $G$, two distinct vertices $s, t$, hop constraint $k$\;}
	\KwOut{a set of relations $R$\;}
	\tcc{Initialize relations.}
	$R_1 \leftarrow \{(s, v)|e(s, v) \in E(G)\}$\;
	$R_k \leftarrow \{(v, t)| e(v, t) \in E(G) \wedge v \neq s \} \cup \{(t, t)\} $\;
	\For{$i \leftarrow 2 \text{ to } k - 1$}{
	    $R_i \leftarrow \{(v, v')|e(v, v') \in E(G-\{s\}) \wedge v \neq t\} \cup \{(t,t)\}$\;
	}
	
	\tcc{Perform full reducer.}
	\For{$i \leftarrow 1$ to $k - 1$}{
	    $C \leftarrow \{v'|(v, v') \in R_i\}$\;
	    \ForEach{$(v, v') \in R_{i + 1}$}{
	        \lIf{$v \notin C$}{$R_{i + 1}\leftarrow R_{i + 1} - \{(v, v')\}$}
	    }
	}
	\For{$i \leftarrow k - 1$ to $1$}{
	    $C \leftarrow \{v|(v, v') \in R_{i + 1}\}$\;
	    \ForEach{$(v, v') \in R_{i}$}{
	        \lIf{$v' \notin C$}{$R_{i}\leftarrow R_{i} - \{(v, v')\}$}
	    }
	}
	\KwRet{$R_{1-k}$\;}
\end{algorithm}

\setlength{\textfloatsep}{0pt}
\begin{figure}[t]\small
    \centering
    \captionsetup[subfigure]{aboveskip=0pt,belowskip=0pt}
    \setlength{\abovecaptionskip}{0pt}
    \setlength{\belowcaptionskip}{0pt}
    \begin{minipage}[b]{0.15\textwidth}
        \begin{subfigure}[t]{\textwidth}
            \centering
            \includegraphics[scale=0.4]{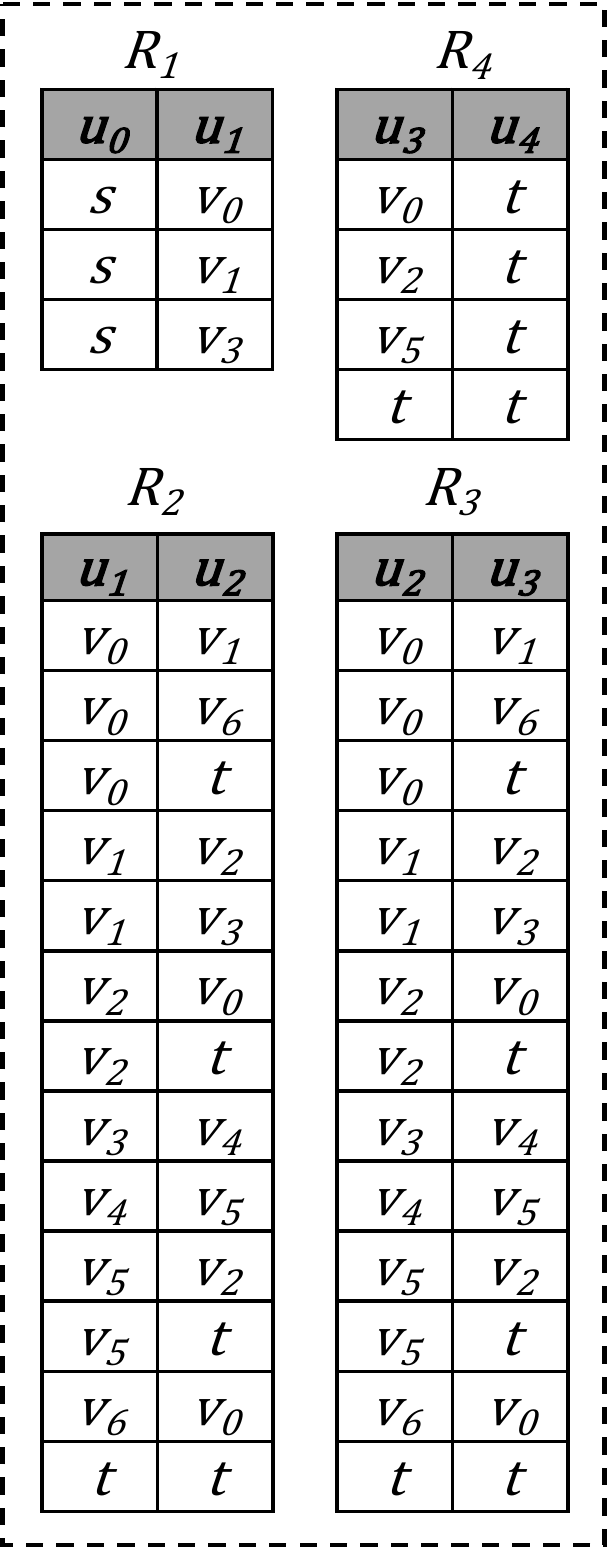}
            \caption{Initial $R$.}
            \label{fig:initial_relations}
        \end{subfigure}
    \end{minipage}
    \begin{minipage}[b]{0.3\textwidth}
        \centering
        \begin{subfigure}[t]{\textwidth}
            \centering
            \includegraphics[scale = 0.4]{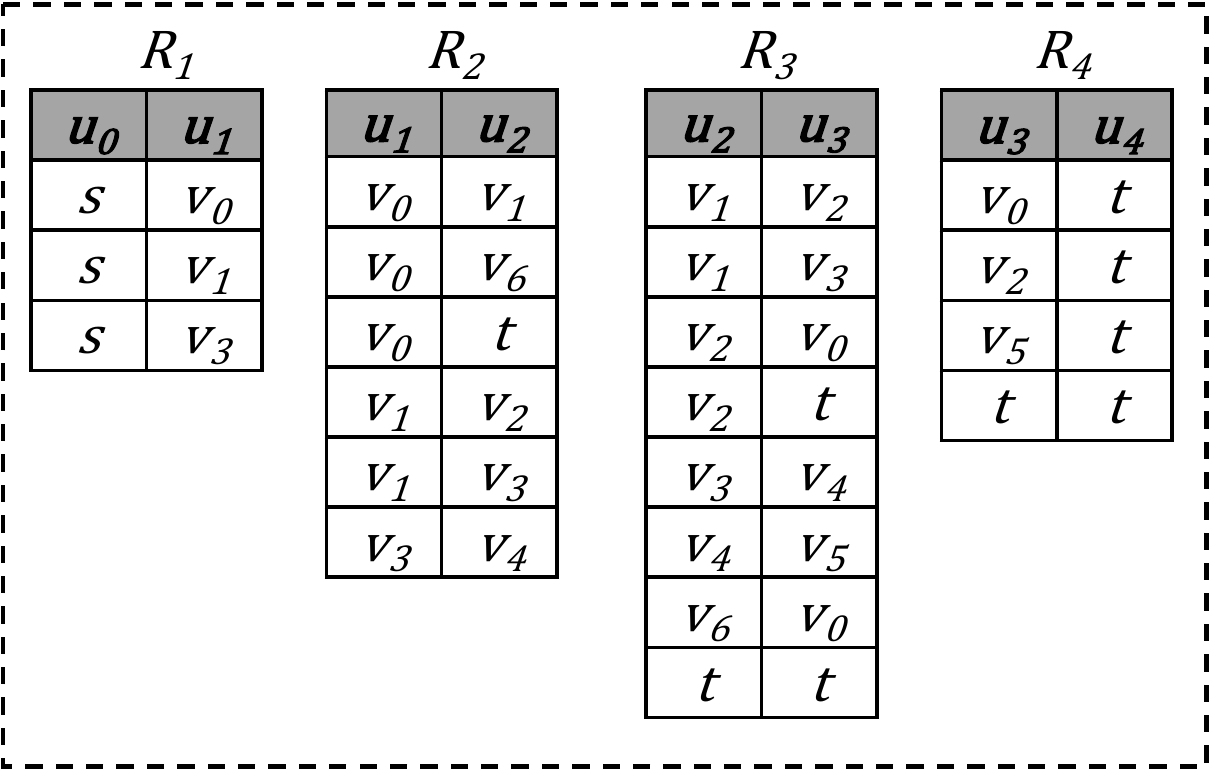}
            \caption{$R$ after pruning from $R_2$ to $R_4$.}
            \label{fig:left_to_right_pruning}
        \end{subfigure}
          \\[2ex]
        \begin{subfigure}[t]{\textwidth}
            \centering
            \includegraphics[scale = 0.4]{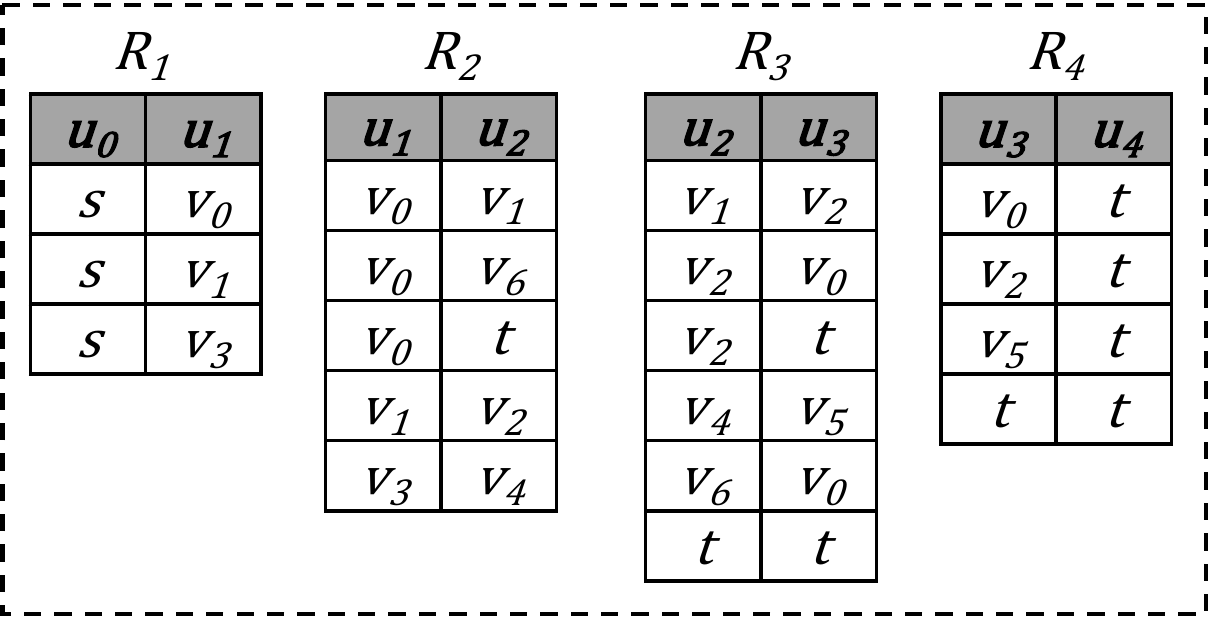}
            \caption{$R$ after pruning from $R_3$ to $R_1$.}
            \label{fig:right_to_left_pruning}
        \end{subfigure}
    \end{minipage}
    \caption{Build relations $R$ of $Q$.}
    \label{fig:sample_graph_1}
\end{figure}

\begin{example}
   Given $G$ and $q(s, t, 4)$ in Figure \ref{fig:example_graphs}, the relations generated by Lines 1-4 are shown in Figure \ref{fig:initial_relations}.
   Figure \ref{fig:left_to_right_pruning} illustrates the relations after pruning from $R_2$ to $R_k$. For example, $(v_4, v_5)$ is removed from $R_2$ because $v_4$ does not
   appear in values of $u_1$ in $R_1$. Figure \ref{fig:right_to_left_pruning} presents the relations after filtering from $R_{k - 1}$ to $R_1$.
   For example, $(v_1, v_3)$ is eliminated from $R_3$ since $v_3$ does not exist in values of $u_3$ in $R_4$.
\end{example}

The space and time complexities are both $O(k \times |E(G)|)$. The relations returned by Algorithm \ref{algo:build_relations} satisfy
the following proposition.

\begin{proposition} \label{lemma:full_reducer}
    Each tuple in relations of $Q$ appear in the final results of evaluating $Q$ \cite{abiteboul1995foundations}.
\end{proposition}

\subsection{Light-Weight Index}

Algorithm \ref{algo:build_relations} removes dangling tuples at the cost of scanning $G$ and each relation several times. The cost can dominate the execution time of
some queries, especially on large graphs. This makes the algorithm hard to meet the real-time constraint.
In order to reveal the problem, we propose a light-weight index, which is built at a small overhead but provides competitive pruning power.

\textbf{General Idea.} Based on the definition of a path from $s$ to $t$,
we have the following proposition.

\begin{proposition} \label{prop:prune_vertex}
    Given a vertex $v \in V(G)$, if there exists a path $P \in \mathcal{P}(s, t, k, G)$ such that $P[i] = v$ where $0 \leqslant i \leqslant |P| - 1$, then $S(s, v | G - \{t\}) \leqslant i$
    and $S(v, t| G - \{s\}) \leqslant k - i$.
\end{proposition}

Let $C_i$ denote the set of vertices $v \in V(G)$ that satisfy $S(s, v|G - \{t\}) \leqslant i$ and $S(v, t|G - \{s\}) \leqslant k - i$. According to Proposition \ref{prop:prune_vertex},
if $v$ appears at position $i$ of $P \in \mathcal{P}(s, t, k, G)$, then $v$ belongs to $C_i$. Moreover, given a vertex $v$,
suppose that the remaining budget traveling to $t$ is $b$. We only need to consider the neighbors $v'$ of $v$ such that $S(v', t|G - \{s\}) \leqslant b - 1$ to meet the hop constraint. Based on the observation, we want to build an index supporting two kinds of lookup operations to serve the subsequent enumeration: (1) given $i$ where $0 \leqslant i \leqslant k$,
retrieve $C_i$; and (2) given $v \in C_i$ and an integer $b$
where $0 \leqslant b \leqslant k$, retrieve the neighbors $v'$ of $v$ such that $S(v', t|G - \{s\}) \leqslant b$ (or the in neighbors $v'$ of $v$ such that $S(s, v'|G - \{t\}) \leqslant b$).

\setlength{\textfloatsep}{0pt}
\begin{algorithm}[t]
    \footnotesize
	\caption{Build Index}
	\label{algo:build_index}
	 \KwIn{a graph $G$, two distinct vertices $s, t$, hop constraint $k$\;}
	 \KwOut{a light-weight index $\mathcal{I}$\;}
	 Set $v.s = S(s, v|G - \{t\})$ and $v.t = S(v, t| G - \{s\})$ for $v \in V(G)$\;
	 $X \leftarrow$ a $(k + 1) \times (k + 1)$ matrix of sets\;
	 \ForEach{$v \in V(G)$}{
	    \lIf{$v.s + v.t \leqslant k$}{Add $v$ to $X[v.s, v.t]$}
	 }
	 $H \leftarrow $ a hash table\;
	 \ForEach{$v \in X - \{t\}$}{
	    Add a key-value pair $(v,\{\})$ to $H$\;
	    \ForEach{$v' \in N(v)$}{
	        \lIf{$v.s + v'.t + 1 \leqslant k$}{Add $v'$ to $H[v]$}
	    }
	 }
	 Add a key-value pair $(t, \{t\})$ to $H$\;
	 Sort the values $v' \in H[v]$ of $v \in X$ by ascending order of $v'.t$\;
	 \KwRet{$\mathcal{I}(X, H)$}\;
\end{algorithm}

\textbf{Implementation.} Algorithm \ref{algo:build_index} presents the details of building the index. For each $v \in V(G)$, Line 1
sets $v.s$ and $v.t$ as $S(s, v| G - \{s\})$ and $S(v, t| G - \{t\})$, respectively. We implement this by performing two breadth-first search from $s$
and $t$, respectively. Lines 2-4 divide the vertices $v \in V(G)$ into disjoint sets based on $v.s$ and $v.t$. The partition considers
the vertices $v$ such that $v.s + v.t \leqslant k$ only. After that we build a hash table $H$ to maintain
the relationship between vertices in $X$ and their neighbors (Lines 5-10). The key is the vertex $v$ in $X$ and the value is the set of neighbors $v'$
of $v$ that satisfy $v.s + v'.t + 1 \leqslant k$. Line 11 sorts the neighbors $v'.t$ of $v$ in $H$ by the ascending order of $v'.t$. 
Finally, we return the index $\mathcal{I}$ that contains $X$ and $H$. In practice,
$H$ has three components that is the \emph{Neighbors} array storing neighbors $v'$ of each vertex $v \in X$ by the ascending order of $v'.t$,
the \emph{Offset} array indexing the neighbor set of each vertex by the distance to $t$, and the \emph{Hash Table} the key and value of which are the vertices
in $X$ and a pointer to the beginning position at the \emph{Offset} array. The following is an example.


\begin{example}
    Figure \ref{fig:build_index} presents $\mathcal{I}$ on $G$ and $q(s, t, 4)$ in Figure \ref{fig:example_graphs}. Figure \ref{fig:vertex_bins}
    shows the partitions $X$ of $v \in V(G)$ based on $S(s, v|G - \{t\})$ and $S(v, t|G - \{s\})$. For example, $X[2, 2] = \{v_4, v_6\}$. Figure \ref{fig:index}
    demonstrates the implementation of $H$. Take $v_0$ as an example. It has three neighbors $\{t, v_1, v_6\}$, which are stored in the \emph{Neighbors} array
    by the ascending order of the distance to $t$. As $k = 4$, $v_0$ has five slots in the \emph{Offset} array to index $\{t, v_1, v_6\}$ based
    on the distance to $t$. The value in $H$ is 0, which points to the begin position of slots belonging to $v_0$ in the \emph{Offset} array. Suppose
    that we want to retrieve the neighbors $v$ of $v_0$ such that $S(v, t|G - \{s\}) \leqslant 2$. We first get the beginning position of the neighbor set
    of $v_0$ from the first slot of the \emph{Offset} array, which is 0. Next, we get the end position of the neighbors satisfying the distance constraint
    from the fourth slot, which is 3. Then, we get the results from the \emph{Neighbors} array, which are $\{t, v_1, v_6\}$.
\end{example}

\textbf{Index Lookup Operations} The index $\mathcal{I}$ supports two kinds of operations listed below.

\begin{itemize}
    \item $\mathcal{I}(i)$: Retrieve $C_i$, i.e., the vertices $v \in V(G)$ satisfy that $S(s, v|G - \{t\}) \leqslant i$ and $S(v, t|G - \{s\}) \leqslant k - i$.
    \item $\mathcal{I}_t(v, b)$ (or $\mathcal{I}_s(v, b)$): Retrieve the neighbors $v'$ of $v$ such that $S(v', t| G - \{s\}) \leqslant b$ (or the in neighbors
    $v'$ of $v$ such that $S(s, v' | G - \{t\}) \leqslant b$).
\end{itemize}

$\mathcal{I}(i)$ and $\mathcal{I}_t(v, b)$ are implemented based on $X$ and $H$, respectively. The time complexity of the two operations are both $O(1)$.

\setlength{\textfloatsep}{0pt}
\begin{figure}[t]\small
	\setlength{\abovecaptionskip}{0pt}
	\setlength{\belowcaptionskip}{0pt}
	\centering
	\begin{subfigure}[t]{0.19\textwidth}
		\centering
		\includegraphics[scale=0.4]{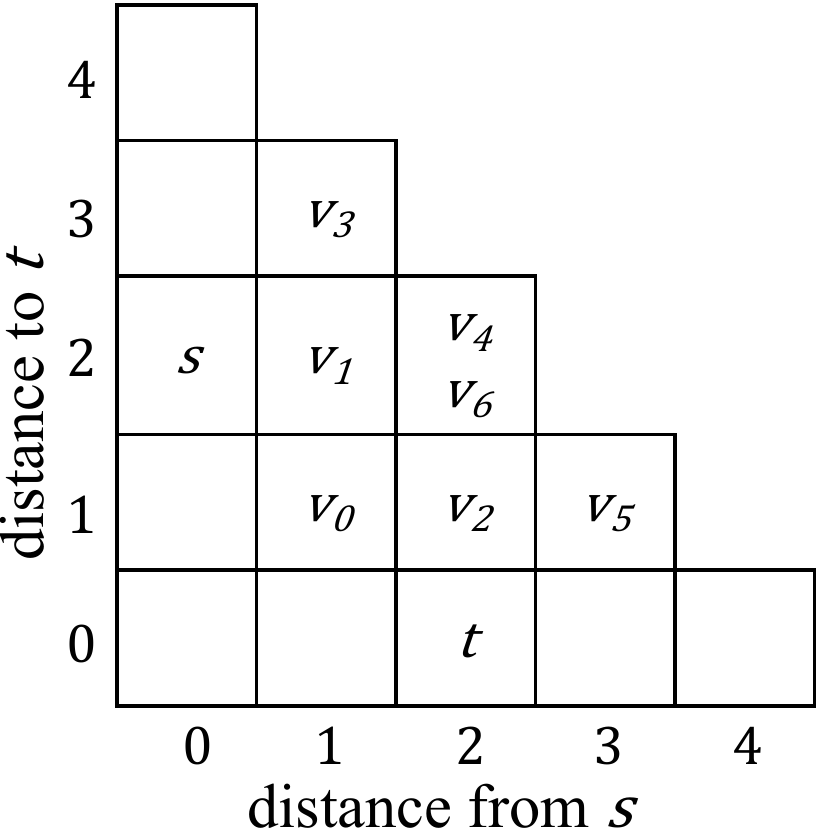}
		\caption{Partitions $X$.}
		\label{fig:vertex_bins}
	\end{subfigure}	
	\begin{subfigure}[t]{0.27\textwidth}
		\centering
		\includegraphics[scale=0.4]{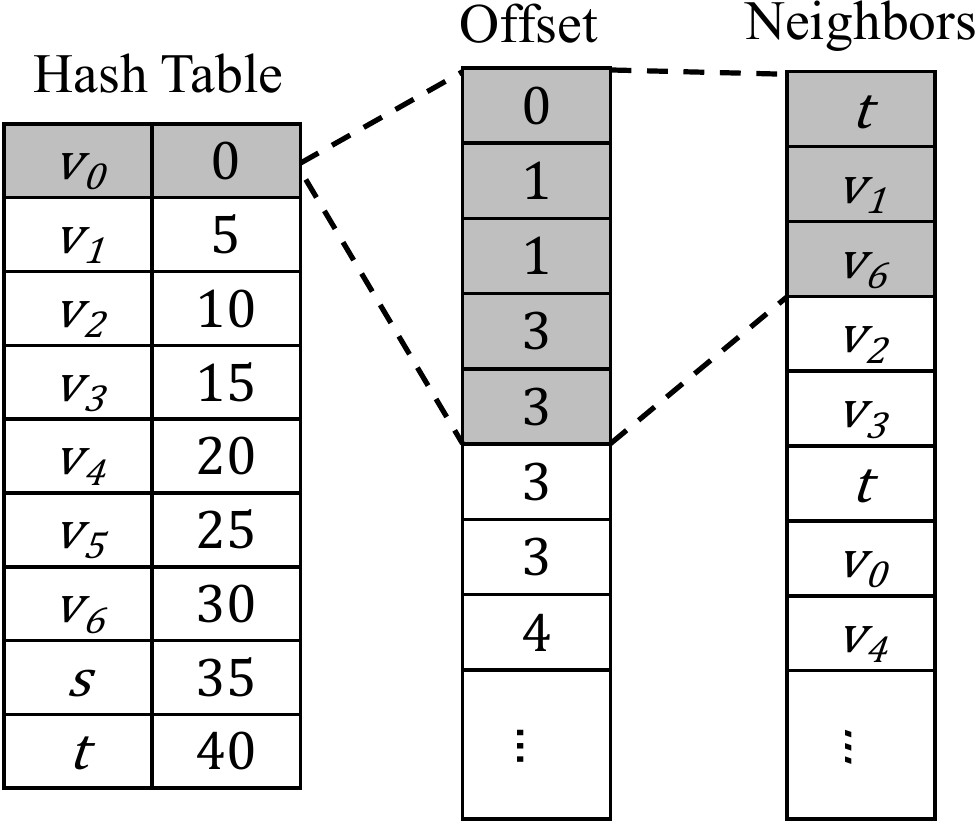}
		\caption{Hash table $H$.}
		\label{fig:index}
	\end{subfigure}	
	\caption{Build the index $\mathcal{I}$ on $G$.}
	\label{fig:build_index}
\end{figure}

\subsection{Analysis} \label{sec:index_analysis} 

\SUN{We compare the pruning power of Algorithm \ref{algo:build_index} with Algorithm \ref{algo:build_relations} in the appendix.}
In this following, we analyze the space and time complexities of Algorithm \ref{algo:build_index}.

\textbf{Space.} The space complexity of $X$ is $O((k + 1) ^ 2 + |V(G)|)$
because $X$ contains all vertices of $G$ at most. The space complexity of $H$ is $O(|E(G)| + k \times |V(G)|)$ because
we store all edges in $E(G)$ at most and each vertex has $k + 1$ slots in the \emph{Offset} array. Therefore, the space
complexity of constructing $\mathcal{I}$ is $O(|E(G)| + k \times |V(G)|)$.

\textbf{Time.} Line 1 in Algorithm \ref{algo:build_index} takes $O(|E(G)| + |V(G)|)$ time because
we perform two breadth-first searches. Lines 2-4 takes $|V(G)|$ time. Since Lines 6-9 loop
over the neighbors of each vertex in $X$, the cost is $O(|E(G)|)$. We implement the sort at Line 11 with the counting sort because
$k$ is small. Therefore, the cost is $O(|E(G)|)$ as well. In summary, the time complexity of Algorithm \ref{algo:build_index} is $O(|E(G)| + |V(G)|)$.
The cost in practice is small because many vertices and edges are ruled out by the distance constraint.

%% file: 5_search_on_index.tex
\section{Search on Index} \label{sec:generate_join_plan}

\subsection{Depth-First Search on Index}

Algorithm \ref{algo:dfs_on_index} presents the depth-first search method on the index $\mathcal{I}$.
The \emph{Search} procedure recursively enumerates all hop-constrained paths from $s$ to $t$ based on $\mathcal{I}$ (Lines 3-7). If the last vertex $v$ in $M$ is $t$,
then we find a result and emit it (Line 4). Otherwise, we consider the neighbors $v'$ of $v$ such that $S(v', t) \leqslant k - L(M) - 1$ as the next vertex in $M$ to
meet the hop constraint. In particular, we loop over $\mathcal{I}_t(v, k - L(M) - 1)$, add $v'$ to $M$, and continue the search.
The check at line 7 ensures that there are no duplicate vertices in $M$.

\subsection{Analysis} \label{sec:dfs_analysis}

\SUN{Algorithm \ref{algo:dfs_on_index} can be easily extended to support HcPE queries with variant constraints such as accumulative values and
a sequence of actions, which is detailed in 
the appendix.} In the following, we focus on the space consumption and time complexity. 


\textbf{Space.} Algorithm \ref{algo:dfs_on_index} spends $O(k)$ space to store partial results because it maintains one partial result at a time during the search.

\textbf{Time.} Algorithm \ref{algo:dfs_on_index}
performs a DFS on the search tree where nodes are partial results $M$ and edges are operations of adding a vertex to $M$.
We analyze the time complexity based on the search tree. Let $\mathcal{M}_i$ represent the nodes at depth $i$ in the search tree, which are
the set of partial results $M$ containing $i + 1$ vertices. Each internal node $M \in \mathcal{M}_i$ corresponds to an invocation of the \emph{Search} procedure.
The cost of an invocation is $|\mathcal{I}_t(M[i], k - i - 1)|$ (i.e., the for loop at Lines 6-7). Then, the running time $T$ of Algorithm \ref{algo:dfs_on_index}
is computed by Equation \ref{eq:dfs_time_complexity}. 

\begin{equation} \label{eq:dfs_time_complexity}
    T = \sum_{0 \leqslant i \leqslant k - 1} \sum_{M \in \mathcal{M}_i} |\mathcal{I}_t(M[i], k - i - 1)|.
\end{equation}

As Algorithm \ref{algo:dfs_on_index} generates partial results incrementally, $\mathcal{M}_{i + 1}$ and $\mathcal{M}_i$ have the following relation where $0 \leqslant i \leqslant k - 1$.

\begin{equation} \label{eq:dfs_relationship}
    \mathcal{M}_{i + 1} = \bigcup_{M \in \mathcal{M}_i} \{M \cup \{v\} | v \in  \mathcal{I}_t(M[i], k - i - 1) - M \}.
\end{equation}

\begin{algorithm}[t]
	\footnotesize
	\caption{Depth-First Search on Index}
	\label{algo:dfs_on_index}
	\SetKwFunction{Search}{Search}
	\SetKwProg{proc}{Procedure}{}{}
	\KwIn{two distinct vertices $s, t$, hop constraint $k$, index $\mathcal{I}$\;}
	\KwOut{all $k$ hop-constrained paths from $s$ to $t$\;}
	$M \leftarrow (s)$\;
	\Search{$t, k, M, \mathcal{I}$}\;
	\proc{\Search{$t, k, M, \mathcal{I}$}}{
		$v \leftarrow$ the last vertex in $M$\;
		\lIf{$v = t$}{$emit(M)$, \KwRet}
		\ForEach{$v' \in \mathcal{I}_t(v, k - L(M) - 1)$}{
			\lIf{$v' \notin M$}{\Search{$t, k, M \cup \{v'\}, \mathcal{I}$}}
		}
	}
\end{algorithm}

\SUN{Because $M$ is generated during the enumeration, it is hard to estimate the number of vertices $v \in  \mathcal{I}_t(M[i], k - i - 1)$
belonging to $M$. For the ease of analysis, we relax the constraint of Algorithm \ref{algo:dfs_on_index} by removing the check at line 7 because $k$ is
small and $M$ contains a few vertices.}
Let $\widetilde{\mathcal{M}}_i$ denote the set of partial results containing $i + 1$
vertices, which are generated by the algorithm after relaxation. According to Equation \ref{eq:dfs_relationship},
$\widetilde{\mathcal{M}}_{i + 1} = \bigcup_{M \in \widetilde{\mathcal{M}}_i} \{M \cup \{v\} | v \in  \mathcal{I}_t(M[i], k - i - 1)\}$ and $\mathcal{M}_i \subseteq \widetilde{\mathcal{M}}_i$. The algorithm after relaxation satisfies the following proposition.

\begin{proposition} \label{prop:walk_correctness}
    Algorithm \ref{algo:dfs_on_index} without the check at line 7 finds all $k$ hop-constrained walk $\mathcal{W}(s, t, k, G)$ from $s$ to $t$ in $G$. Given $M \in \widetilde{\mathcal{M}}_i$
    where $0 \leqslant i \leqslant k$, $M$ must appear in a walk $W \in \mathcal{W}(s, t, k, G)$.
\end{proposition}

The proposition shows that each leaf in the search tree of Algorithm \ref{algo:dfs_on_index} after relaxation is a walk $W \in \mathcal{W}(s, t, k, G)$.
Then, $|\widetilde{\mathcal{M}}_i| \leqslant \delta_W$ where $\delta_W = |\mathcal{W}(s , t, k, G)|$. Together with Equations \ref{eq:dfs_time_complexity} and \ref{eq:dfs_relationship},
we get the following equation.

\begin{equation}
    \begin{split}
            T & = \sum_{0 \leqslant i \leqslant k - 1} \sum_{M \in \mathcal{M}_i} |\mathcal{I}_t(M[i], k - i - 1)| \\
              & \leqslant \sum_{0 \leqslant i \leqslant k - 1} \sum_{M \in \widetilde{\mathcal{M}}_i} |\mathcal{I}_t(M[i], k - i - 1)| \\
              & = \sum_{0 \leqslant i \leqslant k - 1} |\bigcup_{M \in \widetilde{\mathcal{M}}_i} \{M \cup \{v\} | v \in  \mathcal{I}_t(M[i], k - i - 1)\}| \\
              & = \sum_{1 \leqslant i \leqslant k} |\widetilde{\mathcal{M}}_i| \leqslant k \times \delta_W.
    \end{split}
\end{equation}

In summary, given $G$ and $q(s, t, k)$, the running time of Algorithm \ref{algo:dfs_on_index} is $O(k \times \delta_W)$. The analysis indicates that when
most of walks in $\mathcal{W}(s, t, k, G)$ belong to $\mathcal{P}(s, t, k, G)$, Algorithm \ref{algo:dfs_on_index} generates a few \emph{invalid partial results}, i.e., the partial
results do not exist in any final results, and its running time is very close to the lower bound $\Omega(\delta_P)$ for the problem where $\delta_P = |\mathcal{P}(s, t, k, G)|$.
In contrast, when most walks in $\mathcal{W}(s, t, k, G)$ are not paths, the algorithm can result in a large number of invalid partial results. Example \ref{exmp:gap}
presents an example.
From the example, we can also see that the gap between $\delta_P$ and $\delta_W$ (i.e., $\frac{\delta_P}{\delta_W}$) depends on the query and the graph topology.

\begin{example} \label{exmp:gap}
    We execute a query $q(s, t, 4)$ on $G_0$ and $G_1$ in Figure \ref{fig:sample_graphs}, respectively.
    $|\mathcal{W}(s, t, 4, G_0)|$ is equal to 8, and each walk in $\mathcal{W}(s, t, 4, G_0)$ is a path in
    $\mathcal{P}(s, t, 4, G_0)$, for example, $W = (s, v_0, v_2, v_4, t)$. In contrast, $|\mathcal{W}(s, t, 4, G_1)|$ is equal to 6, and only $W = (s, v_0, t)$ belongs to
    $\mathcal{P}(s, t, 4, G_1)$.
\end{example}

%% file: 6_search_optimization.tex
\section{Query Optimization}

\subsection{General Idea}

Algorithm \ref{algo:dfs_on_index} enumerates all results through extending the partial result from $s$ by one vertex at a step. It is equivalent to
evaluating $Q$ along the join order $(R_1, R_2, ... R_k)$, which is a left-deep join tree. Because the join order has an important impact on the cost of the enumeration,
we want to optimize it to further accelerate the query. On the other hand, an important observation we made on the HcPE problem is that the running time of different queries varies greatly. We can easily answer the query
with a small search space regardless of join orders. Consequently, the benefit of optimizing join orders on these queries is
limited. Even worse the optimization time dominates the query time if the optimization method is complex.

In order to reveal the problem, we propose an optimizer generating join orders based on the index in two phases. In particular,
we first use a preliminary cardinality estimator to roughly but quickly estimate the size of the search space.
If the search space size is small, then we directly invoke Algorithm \ref{algo:dfs_on_index}. Otherwise,
we generate a join order with a full-fledged cardinality estimator, which provides more accurate estimation but at a higher cost.
The optimizer selects the method (Algorithm \ref{algo:dfs_on_index} versus. Algorithm \ref{algo:join_on_index}) with lower cost to evaluate the query.

\begin{figure}[t]\small
	\setlength{\abovecaptionskip}{0pt}
	\setlength{\belowcaptionskip}{0pt}
	\centering
	\begin{subfigure}{0.23\textwidth}
		\centering
		\includegraphics[scale=0.4]{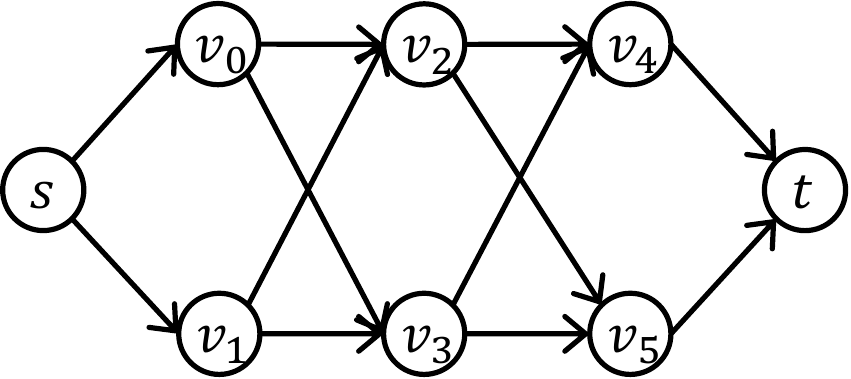}
		\caption{Graph $G_0$.}
		\label{fig:graph_0}
	\end{subfigure}	
	\begin{subfigure}{0.22\textwidth}
		\centering
		\includegraphics[scale=0.4]{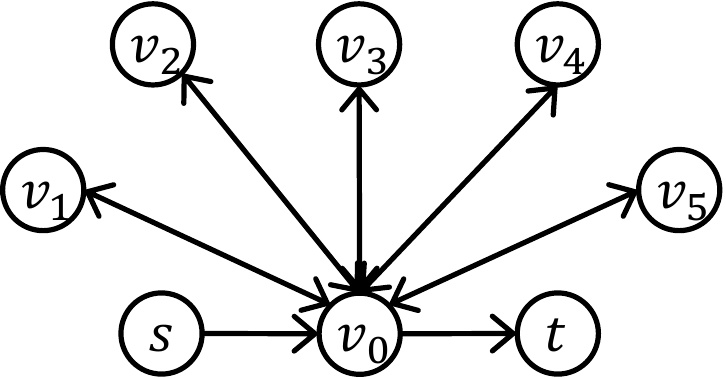}
		\caption{Graph $G_1$.}
		\label{fig:graph_1}
	\end{subfigure}	
	\caption{Sample graphs.}
	\label{fig:sample_graphs}
\end{figure}

\subsection{Cardinality Estimator} \label{sec:coarse-grained_cardinality_estimator}

\textbf{Preliminary Cardinality Estimator.}
Given $q(s, t, k)$ on $G$, the preliminary estimator aims to estimate the size of the search space roughly but quickly. Based on the analysis in Section \ref{sec:dfs_analysis},
the size of the search space can be estimated as $\sum_{1 \leqslant i \leqslant k} |\widetilde{\mathcal{M}}_i|$ where $\widetilde{\mathcal{M}}_{0} = \{(s)\}$ and
$\widetilde{\mathcal{M}}_{i} = \bigcup_{M \in \widetilde{\mathcal{M}}_{i - 1}} \{M \cup \{v\} | v \in  \mathcal{I}_t(M[i - 1], k - i)\}$. Suppose that the average number of immediate partial
results derived from partial results in $\widetilde{\mathcal{M}}_{i}$ is $\gamma_i$. Then, $|\widetilde{\mathcal{M}}_{i + 1}| = \gamma_i \times |\widetilde{\mathcal{M}}_{i}|$.

To calculate the size of the search space, we would like to estimate $\gamma_i$. Given $M \in \widetilde{\mathcal{M}}_{i}$,
the number of immediate partial results derived from $M$ is $|\mathcal{I}_t(M[i], k - L(M) - 1)|$. $M[i]$ must belong to $C_i$ based on Proposition \ref{prop:prune_vertex}.
Then, an intuitive method assessing $\gamma_i$ is to estimate it as the average number of neighbors $v'$ of vertices $v$ in $C_i$ that satisfy $S(v', t| G - \{s\}) \leqslant k - L(M) - 1$,
i.e., $\frac{1}{|C_i|}\sum_{v \in C_i} |\mathcal{I}_t(v, k - L(M) - 1)|$. The estimated value is denoted by $\hat{\gamma}_i$. In total,
the estimated size $\hat{T}$ of the search space is computed by Equation \ref{eq:preliminary_estimation}. The value of $\hat{\gamma}_i$ is a basic statistics of $\mathcal{I}$,
which is collected during the index construction. Then, the time complexity of computing the equation is $O(k ^ 2)$, which incurs a small cost.

\begin{equation} \label{eq:preliminary_estimation}
    \begin{split}
        \hat{T} &= \sum_{1 \leqslant i \leqslant k} |\widetilde{\mathcal{M}}_i| \approx \sum_{0 \leqslant i \leqslant k - 1} \prod_{0 \leqslant j \leqslant i} \hat{\gamma}_j \\
                &= \sum_{0 \leqslant i \leqslant k - 1} \prod_{0 \leqslant j \leqslant i} \frac{1}{|C_j|}\sum_{v \in C_j} |\mathcal{I}_t(v, k - L(M) - 1)|.
    \end{split}
\end{equation}

We compare $\hat{T}$ with a threshold $\tau$. If $\hat{T} > \tau$, then we optimize the join order with the full-fledged cardinality estimator. Otherwise,
we evaluate the query with Algorithm \ref{algo:dfs_on_index}. Therefore, we set $\tau$ such that the cost of the optimization is neglected compared with the search
time when $\hat{T} > \tau$. \SUN{In particular, given $G$, $\tau$ is measured by pre-executing some random queries with Algorithm \ref{algo:dfs_on_index}
and testing $\tau$ from $10$, $10 ^ 2$,..., till the time finding $\tau$ results is longer than the join plan optimization time for most of queries.
In our experiments, we execute 100 queries and set $\tau$ as $10 ^ 5$, which works well in our workloads. This is because the optimization time on these graphs
is generally shorter than the time of finding $10^5$ results. Moreover, if the queries have fewer than $10^5$ results, then the enumeration time
is small (several milliseconds), which makes the gain of optimization limited. As such, we directly use Algorithm \ref{algo:dfs_on_index} to answer them.}

\textbf{Full-fledged Cardinality Estimator.} The full-fledged estimator gives an accurate estimation on the size of the search space. To avoid performing the Cartesian product
of two relations, we require that the two relations in a join operation must have common attributes. Therefore, each sub-query $Q'$ is a sub-chain of $Q$.
$Q[i:j]$ denotes a sub-query $Q' = R_{i + 1}(u_i, u_{i+1}) \Join \dots \Join R_{j}(u_{j - 1}, u_j)$. We want to estimate $|Q[i:j]|$ based on the index.

We first consider a simple case $Q' = Q[i:i+1]$ where $0 \leqslant i < k$, i.e., $Q'$ is a base relation $R_{i + 1}(u_i, u_{i + 1})$. Based on $\mathcal{I}$, we obtain
that $R_{i + 1} = \bigcup_{v \in C_i} \{(v, v')| v' \in \mathcal{I}_t(v, k - i - 1)\}$. Thus, $|Q'| = |R_{i + 1}| = \sum_{v \in C_{i}} |\mathcal{I}_t(v, k - i - 1)|$.
Let $c_{i + 1} ^ {i}(v)$ denote the number of tuples starting with $v$ in $Q[i : i+1]$. Given $Q' = Q[i - 1:i + 1]$, we have
$|Q'| = |R_{i} \Join R_{i+1}| = \sum_{v \in C_{i - 1}}c_{i + 1} ^ {i - 1}(v) = \sum_{v \in C_{i - 1}}\sum_{v' \in \mathcal{I}_t(v, k - i)}c_{i + 1} ^ {i}(v')$. Therefore,
we compute $|Q[i:j]|$ as follows.

\begin{equation} \label{eq:full_fledged_estimation}
        |Q[i:j]| = \sum_{v \in C_i}c_j ^ {i}(v).
\end{equation}

\begin{equation} \label{eq:full_fledged_estimation_details}
     c_j ^ {i}(v) =
        \begin{cases}
                    1    & \text{If $i = j$}. \\
                    \sum_{v' \in \mathcal{I}_t(v, k - i - 1)} c_j ^ {i + 1}(v') & \text{If  $i < j$}.
        \end{cases}
\end{equation}

\SUN{We estimate $|Q[0:k]|$ with a dynamic programming method based on the index $\mathcal{I}$.
Given $0 \leqslant i \leqslant k$, we store $c_k ^ i(v)$ for each vertex $v \in \mathcal{I}(i)$. As such, the space cost
of the full-fledged cardinality estimator is $\sum_{0 \leqslant i \leqslant k}|\mathcal{I}(i)|$. Based on Equation \ref{eq:full_fledged_estimation_details},
we first set $c_k ^ k(v)$ to $1$ for each $v \in \mathcal{I}(k)$. Given $0 \leqslant i \leqslant k$, $|Q[i : k]|$ is equal
to $\sum_{v \in I(i)}c_k ^ i(v)$ where $c_k ^ i (v) = \sum_{v' \in \mathcal{I}_t(v, k-i-1)}c_k ^ {i + 1}(v')$ according
to Equations \ref{eq:full_fledged_estimation} and \ref{eq:full_fledged_estimation_details}. Then, the cost of estimating
$|Q[i:k]|$ based on $Q[i+1:k]$ is $\sum_{v \in I(i)}|\mathcal{I}_t(v, k-i-1)|$.
With the method, we calculate $|Q[0:k]|$ along the order from $k - 1$ to $0$. Therefore, the cost is equal to
$\sum_{0 \leqslant i \leqslant k - 1} \sum_{v \in I(i)}|\mathcal{I}_t(v, k-i-1)|$. As $\mathcal{I}(i) \leqslant |V(G)|$,
the space complexity is $O(k \times |V(G)|)$. Given $v \in \mathcal{I}(i)$, $|\mathcal{I}_t(v, k - i - 1)| \leqslant d(v)$.
As such, $\sum_{v \in \mathcal{I}(i)}|\mathcal{I}_t(v, k-i-1)| \leqslant \sum_{v \in \mathcal{I}(i)}d(v) \leqslant |E(G)|$,
and the time complexity is $O(k \times |E(G)|)$. The number of edges in the index is generally smaller than $|E(G)|$ because
of the filtering. The time complexity can be met when $G$ is a clique. The implementation of the estimator
will be introduced in Algorithm \ref{algo:generate_join_order}.}

\subsection{Join On Index} \label{sec:join_on_index}

\setlength{\textfloatsep}{0pt}
\begin{algorithm}[t]
    \footnotesize
	\caption{Join Order Optimization}
	\label{algo:generate_join_order}
	\SetKwProg{proc}{Procedure}{}{}
	 \KwIn{two distinct vertices $s,t$, hop constraint $k$, index $\mathcal{I}$\;}
	 \KwOut{the cut position $i ^ *$\;}
	 \tcc{Estimate number of paths to $t$.}
	 Set $c_k^i(v)$ as 0 for each $v \in \mathcal{I}(i)$ where $0 \leqslant i \leqslant k - 1$\;
	 Set $c_k^k(v)$ as 1 for each $v \in \mathcal{I}(k)$\;
	 \For{$i \leftarrow k - 1\text{ to }0$}{
	    \ForEach{$v \in \mathcal{I}(i)$, $v' \in \mathcal{I}_t(v, k - i - 1)$}{
	        $c_k ^ i(v) \leftarrow c_k ^ i(v) + c_{k} ^ {i + 1}(v')$\;
	    }
	 }
	 \tcc{Estimate number of paths from $s$.}
	 Set $c_i ^ 0(v)$ as 0 for each $v \in \mathcal{I}(i)$ where $1 \leqslant i \leqslant k$\;
	 Set $c_0 ^ 0(v)$ as 1 for each $v \in \mathcal{I}(0)$\;
	 \For{$i \leftarrow 1 \text{ to }k$}{
	    \ForEach{$v \in \mathcal{I}(i)$, $v' \in \mathcal{I}_s(v, k - i - 1)$}{
	        $c_i ^ 0 (v) \leftarrow c_i ^ 0(v) + c_{i - 1} ^ {0}(v')$\;
	    }
	 }
	 \tcc{Find the cut position $i^*$.}
	 $i^* \leftarrow arg \min_{0 \leqslant i \leqslant k}(\sum_{v \in \mathcal{I}(i)}c_i ^ 0(v) + \sum_{v \in \mathcal{I}(i)}c_k ^ i(v))$\;
    \KwRet $i ^ *$\;
\end{algorithm}

\textbf{Join Order Optimization.} The cost of a join order can be estimated by the cost model in Equation \ref{eq:cost_model} with the full-fledged cardinality estimator.
Because it is prohibitively expensive to enumerate all orders, the number of which is exponential to the number of relations, to minimize the cost,
we design a greedy optimization method. In particular, we minimize Equation \ref{eq:cost_model} in a top-down manner by (1) cutting $Q$ into two
sub-queries $Q[0:i]$ and $Q[i : k]$ such that the sum of $|Q[0:i]|$ and $|Q[i:k]|$ is minimized; and (2) cutting $Q[0:i]$ and $Q[i : k]$ into smaller sub-queries, respectively,
and continuing the process until each sub-query is a base relation.

However, the benefit of optimizing the orders of evaluating $Q[0:i]$ and $Q[i:k]$ is limited for a query with a large search space. Specifically,
a large search space indicates that the number of partial results grows exponentially with the length of the path increasing. Given any sub-query $Q'$ of $Q[0:i]$ (or $Q[i:k]$),
$|Q'|$ is much less than $|Q[0:i]|$ (or $|Q[i:k]|$). Consequently, the last join operation $Q = Q[0:i] \Join Q[i:k]$ dominates the evaluation cost. Therefore,
we simplify the join order optimization as follows: (1) find a cut position $i^*$ of $Q$ such that the sum of $Q[0:i^*]$ and $Q[i^*:k]$ is minimized; (2)
evaluate $Q[0:i^*]$ and $Q[i^*:k]$ with the depth-first search method, respectively; and (3) perform $Q[0:i^*] \Join Q[i^*:k]$ to find final results.

Algorithm \ref{algo:generate_join_order} illustrates the method finding the cut position $i^*$. Given $v \in \mathcal{I}(i)$ (i.e., $C_i$) where $0 \leqslant i \leqslant k$,
Lines 1-5 estimate the number of paths from $v$ to $t$ based on the full-fledged cardinality estimator.
Similarly, Lines 6-10 estimate the number of paths from $s$
to $v$. Finally, Lines 11-12 find the cut position $i ^ *$ such that the sum of $|Q[0:i]|$ and $|Q[i:k]|$ is minimized, and return it.
The time complexity of Algorithm \ref{algo:generate_join_order} is $O(k \times |E(G)|)$ and the space complexity is $O(k \times |V(G)|)$.

\setlength{\textfloatsep}{0pt}
\begin{algorithm}[t]
	\footnotesize
	\caption{Join On Index}
	\label{algo:join_on_index}
	\SetKwFunction{Search}{Search}
	\SetKwProg{proc}{Procedure}{}{}
	\KwIn{two distinct vertices $s,t$, hop constraint $k$, the cut position $i ^ *$, index $\mathcal{I}$\;}
	\KwOut{all $k$ hop-constrained paths from $s$ to $t$\;}
	$R_a \leftarrow \{\}$, $R_b \leftarrow \{\}$\;
	\Search{$M \leftarrow (s), t, 0, k, k - i^*, \mathcal{I}, R_a$}\;
	$C\leftarrow \{r[i^*]|r \in R_a\}$\;
	\ForEach{$v \in C$}{
		\Search{$M \leftarrow (v), t, i^*, k, k - i^* + 1, \mathcal{I}, R_b$}\;	   
	}
	$R \leftarrow R_a \Join_{HJ} R_b$\;
	\ForEach{$r \in R$}{
		\lIf{$r$ is a $k$ hop-constrained path from $s$ to $t$}{\emph{emit}($r$)}
	}
	
	\proc{\Search{$M, t, i, k, l, \mathcal{I}, R$}}{
		\lIf{$|M| = l$}{$R\leftarrow R \cup \{M\}$, \KwRet}
		$v \leftarrow $ the last vertex in $M$\;
		\ForEach{$v' \in \mathcal{I}_t(v, k - i - L(M) - 1)$}{
			\Search{$M \cup \{v'\}, t, i, k, l, \mathcal{I}, R$}\;
		}
	}
\end{algorithm}

After finding the cut position, we compare the cost of the new order with that of Algorithm \ref{algo:dfs_on_index}.
In particular, Algorithm \ref{algo:dfs_on_index} is equivalent to the left-deep join along the order $(R_1, R_2,\dots,R_k)$.
The cost is $T_{DFS} = \sum_{1 \leqslant i \leqslant k} |Q[0:i]|$ based on Equation \ref{eq:cost_model}.
In contrast, the cost of the new order is $T_{JOIN} = |Q| + T(Q[0:i^*]) + T(Q[i^*:k]) = |Q| + \sum_{1 \leqslant i \leqslant i^*}|Q[0:i]| + \sum_{i ^ * < i \leqslant k}|Q[i^*:k]|$.
Based on the intermediate results in Algorithm \ref{algo:generate_join_order}, we can get that $T_{DFS} = \sum_{1 \leqslant i \leqslant k} \sum_{v \in \mathcal{I}(i)}c_i ^ 0(v)$,
while $T_{JOIN} = \sum_{v \in \mathcal{I}(0)}c_k ^ 0 (v) + \sum_{1 \leqslant i  \leqslant i ^ *}\sum_{v \in \mathcal{I}(i)}c_i ^ 0(v) + \sum_{i ^ * \leqslant i  \leqslant k}\sum_{v \in \mathcal{I}(i)}c_k ^ i(v)$. If $T_{DFS} < T_{JOIN}$, then we adopt Algorithm \ref{algo:dfs_on_index}. Otherwise, we evaluate the query with the join-based method, which is introduced
in Algorithm \ref{algo:join_on_index}.

\textbf{Join Implementation.}
Algorithm \ref{algo:join_on_index} presents the join-based method on the index. $R_a$ and $R_b$ store the results of evaluating $Q[0:i^*]$ and $Q[i^*:k]$, respectively (Line 1). 
We first find the results of $Q[0:i^*]$ with a depth-first search from $s$ (Line 2). $M$ is a sequence of vertices. If $M$ contains $l$ vertices, then we add it to $R$ and return (Line 10).
Otherwise, we loop over the neighbors $v'$ of the last vertex $v$ in $M$ such that $S(v', t|G-\{s\}) \leqslant k - i - L(M) - 1$, add $v'$ to $M$, and continue the search (Lines 11-13).
After that, Line 3 collects all vertices appearing in the last position of tuples in $R_a$, i.e., the values of the join key $Q[i ^ *]$. Next, we find the results of $Q[i^*:k]$ with
a depth-first search from each $v \in C$ (Lines 4-5). Finally, we perform the hash join of $R_a$ and $R_b$, and output the valid path (Lines 6-8). In practical implementation, we
check whether a result is a valid path when performing the join operation.

\subsection{Analysis}

Algorithm \ref{algo:join_on_index} satisfies Proposition \ref{prop:join_walk}. Based on the proposition, we analyze its space and time complexities. 

\begin{proposition} \label{prop:join_walk}
    Each partial result $M$ generated by the \emph{Search} procedure appears in a tuple in $R$. Each tuple in $R$ corresponds to a walk in $\mathcal{W}(s, t, k, G)$.
\end{proposition}

\textbf{Space.} Algorithm \ref{algo:join_on_index} maintains intermediate results of evaluating $Q[0 : i^*]$ and $Q[i^* : k]$, which are $R_a$ and $R_b$, respectively.
Based on Proposition \ref{prop:join_walk}, each tuple in $R_a$ (or $R_b$) appears in a result of $R$. Therefore, the sizes of $R_a$ and $R_b$ are less than or equal to $|R| = |\mathcal{W}(s, t, k, G)|$,
and the space complexity is $O(k \times \delta_W)$.

\textbf{Time.} Because each partial result $M$ appears in $R_a$ (or $R_b$), the time complexity of evaluating $Q[0:i^*]$ and $Q[i^*:k]$ are $O(|R_a| \times (i^* + 1))$ and
$O(|R_b| \times (k - i^* + 1))$, respectively. The time complexity of a hash join is $O(|IN| + |OUT|)$ where $|IN|$ and $|OUT|$ are the sizes of the input and output, respectively.
Therefore, the cost of evaluating $R = R_a \Join_{HJ} R_b$ is $O(|R_a| \times (i^* + 1) + |R_b| \times (k - i^* + 1) + |R| \times k)$. As both $|R_a|$ and $|R_b|$ are less than or
equal to $|R|$, the cost is $O(k \times |R|) = O(k \times \delta_W)$. In summary, the time complexity of Algorithm \ref{algo:join_on_index} is $O(k \times \delta_W)$.

\SUN{\textbf{Discussion.} As discussed in Section \ref{sec:related_work}, existing cardinality estimation methods generally take catalogs or relations as
the input \cite{park2020g}, which cannot be directly applied to our light-weight index. Our full-fledged estimator works closely with the query-dependent index that
rules out many invalid edges that cannot appear in any results of the given query. Therefore, it is expected to give more accurate estimation
than working on the original graph or the global statistics of $G$. The method estimates the cardinality based on Equations \ref{eq:full_fledged_estimation} and \ref{eq:full_fledged_estimation_details}, which calculates the number of walks from $s$ to $t$ ($\delta_W$). As a result,
if the gap between the number of walks ($\delta_W$) and the number of paths ($\delta_P$) is small, then our method can give an accurate estimation.
Otherwise, the method can introduce some errors. Our extensive experiment results in the appendix show that the estimation
method works well in practice.}


%% file: 7_experiments.tex
\section{Experiments} \label{sec:experiments}

\subsection{Experimental Setup}

All experiments are conducted in a Linux machine equipped with two Intel Xeon E5-2660 v2 CPUs and 64GB RAM. The graph is first loaded entirely into the main memory
from the disk, and we focus on the scenario of queries on in-memory graphs. Thus we exclude the time on disk I/O.

\textbf{Datasets.} Table \ref{tb:real_world_graphs} lists the details of the 15 real-world graphs, most of which are used in previous work \cite{peng2019towards}. These graphs are from a variety of categories such as social networks, web graphs and biology graphs. The number of vertices ranges from thousands to tens of millions, and the number of edges varies from hundreds of thousands to billions. We use \emph{tm}, a graph with billions of edges, to evaluate the scalability of our algorithm.

\textbf{Queries.} For each graph $G$, we generate four query sets each of which contains 1,000 queries. Different query sets vary in the number of query results as well as search space in enumeration. Specifically, we divide $V(G)$ into two disjoint sets $V'$ and $V''$ based on the vertex degrees:
(1) $V'$ is the set of vertices within top 10\% in the descending order of their degrees; and (2) $V''$ is the remaining ones $V(G)$, excluding $V'$. Then, we have four settings
according to the locations of $s$ and $t$: $\{V', V''\} \times \{V', V''\}$. For each setting, we generate 1,000 queries by choosing $s$ and $t$ uniformly at random. We vary the hop-constraint $k$
from 3 to 8 in our experiments. To guarantee that there exists at least one result, we ensure that the distance between $s$ and $t$ is no larger than 3. We add this constraint because
the query is terminated by a breadth-first search if these is no result, which makes the enumeration problem trivial.
The query set where both $s$ and $t$ belong to $V'$ is generally more challenging than the other three query sets because there are more paths between vertices with large degrees. Therefore, we report the experiment results on the query set where $s, t \in V'$ and $k=6$ by default.

\setlength{\textfloatsep}{0pt}
\begin{table}[t]
  \small
  \centering
  \setlength{\abovecaptionskip}{0pt}
  \setlength{\belowcaptionskip}{0pt}
  \caption{Properties of real-world graphs.}
  \label{tb:real_world_graphs}
  \resizebox{0.48\textwidth}{!}{
\begin{tabular}{ccccccc}
\hline
\textbf{Name} &  \textbf{Dataset} & \textbf{|V|} & \textbf{|E|} & \textbf{$d_{avg}$} & \textbf{Type} \\ \hline
\emph{up}     &  US Patents\footnotemark[1]      & 4M        & 17M     &8.8     & Citation               \\
\emph{db}     &  DBpedia\footnotemark[2]        & 4M        &14M      &6.5     & Miscellaneous  \\
\emph{gg}     &  Web-google\footnotemark[1]       & 876K      & 5M      &11.1     & Web          \\
\emph{st}     &  Web-standford\footnotemark[1]  & 282K      & 2.3M    &16.4     & Web \\
\emph{tw}     &  Twitter-social\footnotemark[2]  & 465K      & 835K    &3.6     & Miscellaneous \\
\emph{bk}     &  Baidu-baike\footnotemark[2]      & 416K      & 3M      &15.8     & Web           \\
\emph{tr}     &  Wiki-trust\footnotemark[2]       & 139K      & 740K    &10.7     & Interaction \\
\emph{ep}     &  Soc-Epinsion1\footnotemark[1]    & 75K       & 508K    &13.4     & Social       \\
\emph{uk}     &  Web-uk-2005\footnotemark[2]     & 121K      & 334K    &181.2     & Web \\
\emph{wt}     &  WikiTalk\footnotemark[2]        & 2M        & 5M      &4.2     & Miscellaneous\\
\emph{sl}     &  Soc-Slashdot0922\footnotemark[1] & 82K       & 948K    &21.2   & Social \\
\emph{lj}     &  LiveJournal\footnotemark[1]     & 5M        & 69M     &28.3    & Social      \\
\emph{da}     &  Rec-dating\footnotemark[2]       & 169K      & 17M     &205.7     & Recommendation \\
\emph{ye}     &  Bio-grid-yeast\footnotemark[2] & 6K        & 314K    &104.5     & Biological \\  
\emph{tm}     &   Twitter-mpi\footnotemark[2]    & 52M       & 1.96B   &74.7     & Miscellaneous \\ \hline
\end{tabular}
}
\end{table}
\footnotetext[1]{http://snap.stanford.edu/data/}
\footnotetext[2]{http://networkrepository.com/networks.php}

\textbf{Metrics.} For each algorithm, we measure the \emph{query time}, \emph{throughput} and \emph{response time}
to process a query. The query time is the elapsed time from the beginning of a query to its end. The response time
is the elapsed time from the beginning of a query to finding the first 1000 results. Both of them are measured in milliseconds (ms).
The throughput is the number of results found per second. We report the arithmetic mean of these metrics on a query set
unless otherwise specified. To complete our experiments in a reasonable time, we set the time limit for a query as two minutes
($1.2 \times 10 ^ 5$ ms). If the query cannot be completed within the time limit, we terminate it and set its
query time as two minutes. The throughput is calculated based on the number of results found when the query is terminated.

\setlength{\textfloatsep}{0pt}
\begin{table*}[t]
  \small
 \centering
  \setlength{\abovecaptionskip}{0pt}
  \setlength{\belowcaptionskip}{0pt}
  \caption{Overall comparison of competing algorithms on different graphs. The star symbol besides the query time
  denotes that the algorithm runs out of time on $>20\%$ queries. The algorithm performing the best on each graph is marked in bold.}
  \label{tb:overall_comparison}
  \resizebox{0.96\textwidth}{!}{
\begin{tabular}{c|ccccc|ccccc|cc}
\hline
\textbf{Dataset} & \multicolumn{5}{c|}{\textbf{Query Time (ms)}}  & \multicolumn{5}{c|}{\textbf{Throughput (\#Results Per Second)}} & \multicolumn{2}{c}{\textbf{Response Time (ms)}}\\ \hline
 &BC-DFS  &BC-JOIN  &IDX-DFS  &IDX-JOIN  &PathEnum &BC-DFS  &BC-JOIN  &IDX-DFS  &IDX-JOIN  &PathEnum  &BC-DFS &IDX-DFS\\ \hline
\emph{up} &5.75e+0  &4.26e+0  &2.75e-1  &2.41e+1  &\textbf{2.28e-1}  &1.46e+3 &1.97e+3 &3.06e+4 &3.50e+2 &\textbf{3.68e+4} &5.75e+0 &\textbf{2.75e-1} \\
\emph{db} &1.13e+1  &1.05e+1  &7.97e-1  &4.06e+1  &\textbf{6.22e-1}  &1.14e+4 &1.24e+4 &1.63e+5 &3.20e+3 &\textbf{2.09e+5} &1.13e+1 &\textbf{7.97e-1} \\
\emph{gg} &1.83e+2  &6.64e+1  &\textbf{9.67e-1}  &8.08e+0  &1.16e+0  &9.38e+4 &2.58e+5 &\textbf{1.77e+7} &2.12e+6 &1.48e+7 &4.65e+1 &\textbf{6.67e-1} \\
\emph{st} &3.67e+3  &4.05e+2  &4.44e+0  &4.95e+0  &\textbf{3.28e+0}  &4.98e+4 &5.29e+5 &4.82e+7 &4.32e+7 &\textbf{6.52e+7} &1.21e+2 &\textbf{1.32e+0} \\
\emph{tw} &3.35e+2  &4.16e+2  &\textbf{1.72e+0}  &2.96e+0  &1.78e+0  &5.60e+1 &4.51e+1 &\textbf{1.09e+4} &6.33e+3 &1.05e+4 &3.35e+2 &\textbf{1.72e+0} \\
\emph{bk} &7.08e+3  &2.63e+3  &9.19e+1  &\textbf{7.57e+1}  &9.29e+1  &5.25e+4 &7.29e+5 &1.88e+8 &\textbf{2.29e+8} &1.87e+8 &4.68e+2 &\textbf{2.14e+0} \\
\emph{tr} &9.88e+4* &1.17e+4  &2.39e+2  &1.00e+2  &\textbf{9.83e+1}  &1.15e+4 &8.33e+5 &5.26e+7 &1.26e+8 &\textbf{1.28e+8} &1.07e+3 &\textbf{1.76e+1} \\
\emph{ep} &1.06e+5* &2.34e+4  &6.55e+2  &\textbf{2.78e+2}  &3.79e+2  &1.34e+4 &1.04e+6 &8.45e+7 &\textbf{2.00e+8} &1.46e+8 &7.35e+2 &\textbf{1.28e+1} \\
\emph{uk} &4.87e+4* &4.47e+4* &3.88e+3  &4.68e+3  &\textbf{3.84e+3}  &7.95e+5 &9.85e+5 &3.21e+8 &2.45e+8 &\textbf{3.24e+8} &1.61e+1 &\textbf{4.23e-1} \\
\emph{wt} &1.05e+5* &3.23e+4  &1.70e+3  &5.14e+2  &\textbf{4.79e+2}  &5.33e+3 &6.20e+5 &3.49e+7 &1.17e+8 &\textbf{1.26e+8} &1.08e+4 &\textbf{1.58e+2} \\
\emph{sl} &1.20e+5* &6.10e+4* &2.76e+3  &7.51e+2  &\textbf{7.18e+2}  &1.43e+4 &1.02e+6 &5.02e+7 &1.85e+8 &\textbf{1.93e+8} &1.42e+3 &\textbf{3.81e+1} \\
\emph{lj} &1.20e+5* &1.20e+5* &8.50e+2  &6.39e+2  &\textbf{4.99e+2}  &1.35e+3 &2.38e+4 &1.69e+7 &2.24e+7 &\textbf{2.88e+7} &1.57e+5 &\textbf{4.38e+2} \\
\emph{da} &1.20e+5* &1.20e+5* &1.26e+4  &3.84e+3  &\textbf{3.32e+3}  &2.10e+3 &4.14e+5 &2.88e+7 &1.19e+8 &\textbf{1.36e+8} &4.13e+4 &\textbf{5.78e+2} \\
\emph{ye} &1.20e+5* &1.20e+5* &7.88e+4* &1.18e+5* &\textbf{6.46e+4}  &6.67e+4 &9.40e+5 &1.87e+8 &4.44e+7 &\textbf{2.34e+8} &3.86e+2 &\textbf{1.01e+1} \\ \hline
\end{tabular}
}
\vspace*{-10pt}
\end{table*}

\textbf{Comparisons.} We study the following algorithms in comparison with PathEnum. We obtain the source code of BC-DFS and BC-JOIN from their original
authors \cite{peng2019towards}. All the competing algorithms are implemented in C++. We compile the code with g++ 7.3.1 with -O3 enabled.

\begin{itemize}[noitemsep,topsep=0pt]
    \item BC-DFS \cite{peng2019towards}: The state-of-the-art polynomial delay method.
    \item BC-JOIN \cite{peng2019towards}: A join-oriented algorithm based on BC-DFS.
    \item IDX-DFS: The proposed depth-first search method.
    \item IDX-JOIN: The proposed join method on the index.
\end{itemize}

The comparison between IDX-DFS/IDX-JOIN and PathEnum is to demonstrate the effectiveness of our cost-based selection. Also note, Peng et al. \cite{peng2019towards} showed that BC-DFS and BC-JOIN outperform T-DFS \cite{rizzi2014efficiently}, T-DFS2 \cite{grossi2018efficient}, KRE \cite{gao2010fast}, KPJ \cite{chang2015efficiently}
and HPI \cite{qiu2018real} by orders of magnitude.


\subsection{Comparison with Existing Algorithms} \label{sec:comparison_with_existing_algorithms}

\textbf{Overall Comparison.} Table \ref{tb:overall_comparison} gives an overall comparison of competing algorithms on different graphs.
We only report the response time of BC-DFS and IDX-DFS because the join-based methods have to obtain the results of each sub-query
before computing the final results, which have a long response time. As shown in the table, the query time on different graphs
varies greatly, which ranges from less than one millisecond to tens of seconds.


\emph{PathEnum vs. BC-DFS/BC-Join.} Our algorithms significantly outperform counterparts on all graphs, especially those with long query time.
For example, IDX-DFS runs 61X faster than BC-DFS on \emph{wt} in terms of query time and achieves 6547X
speedup in terms of throughput. The performance gap is different in terms of query time and throughput because
BC-DFS runs out of time on a number of queries. Additionally, we can see that BC-DFS runs out of time on more than 20\% queries
on a number of graphs, while our algorithms complete most of queries. IDX-DFS spends less than 1 second to find 1000 results
on all graphs, and achieves more than one order of magnitude speedup over BC-DFS in terms of response time.


\emph{IDX-DFS vs. IDX-JOIN.} IDX-DFS outperforms IDX-JOIN on graphs with short queries, but generally runs slower on graphs with long queries. For example,
IDX-DFS achieves up to two orders of magnitude speedup over IDX-JOIN on \emph{up} because there is a small number of results
(i.e., the search space is small) and the time spent on generating join orders can dominate the query time. In contrast,
IDX-JOIN achieves more than two times speedup over IDX-DFS on \emph{tr} and \emph{ep}. The results demonstrate that
optimizing the join order can significantly accelerate the query. 

\emph{Impact of cost optimizer.} PathEnum generally outperforms both IDX-DFS and IDX-JOIN,
especially on graphs with long query time. For example, PathEnum reduces both the query time and the number of queries
running out of time on \emph{ye}. The results prove the effectiveness of our query optimizer.
In a small number of cases, PathEnum can runs slightly slower than IDX-DFS or IDX-JOIN. 
This is because our cost model only considers the impact of the number of partial results, whereas some other factors (e.g.,
the overhead of materialization and the cost of checking whether a vertex belongs to $M$)
can affect the practical performance. 



In the following, we select \emph{ep} and \emph{gg} as representative graphs to demonstrate experiment results.
\emph{ep} takes long query time, while \emph{gg} takes short query time.

\textbf{Detailed Metrics.} To compare the pruning techniques,
we examine the detailed metrics of BC-DFS and IDX-DFS, which includes (1) the number of invalid partial results (\emph{\#Invalid}),
which are the partial results that do not appear in any path in $\mathcal{P}(s, t, k, G)$; (2) the number of edges accessed
(\emph{\#Edges}) during the enumeration; and (3) the number of results reported (\emph{\#Results}). Figure \ref{fig:overall_detailed_metrics}
presents the experiment results. We make the following observations. First, the number of edges accessed by
BC-DFS is around 100 times as many as that by IDX-DFS, which shows the effectiveness of our index. 
The gap narrows on \emph{ep} with $k$ varied from 6 to 8 because BC-DFS runs out of time on most queries and finds fewer results
than IDX-DFS. Second, the number of invalid partial results generated by the studied approaches is very close, which indicates that
the pruning techniques in BC-DFS provide limited extra pruning power compared with simply using the distance to $t$ in our method. 
Moreover, the number of invalid partial results accounts for a small portion of results. This implies that
the benefit of adopting complex pruning techniques during the enumeration to reducing the invalid partial results is limited.

\begin{figure}[h]\small
	\setlength{\abovecaptionskip}{0pt}
	\setlength{\belowcaptionskip}{0pt}
	\captionsetup[subfigure]{aboveskip=0pt,belowskip=0pt}
	\centering
	\begin{subfigure}[t]{0.23\textwidth}
		\centering
		\includegraphics[scale=0.23]{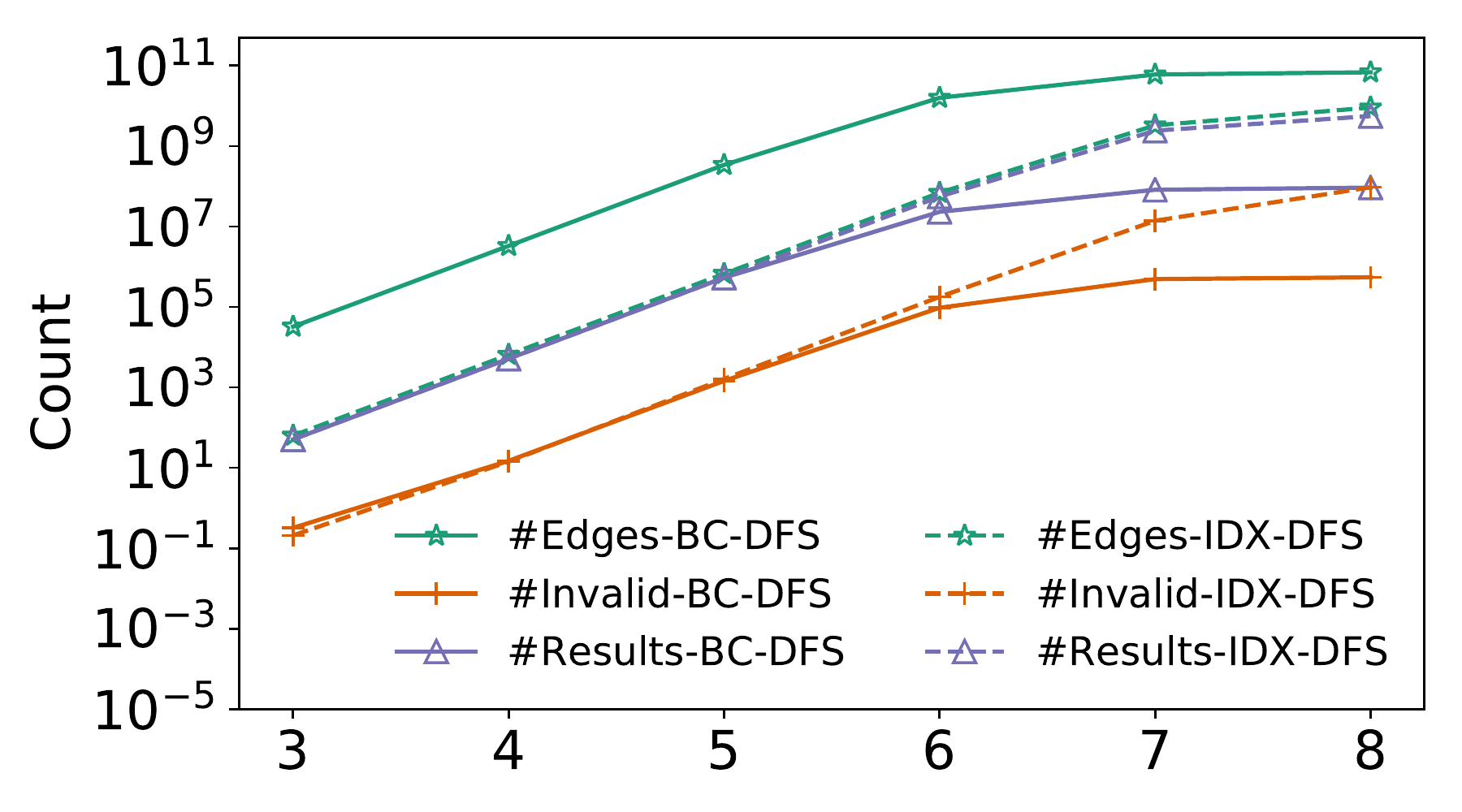}
		\caption{\emph{ep}.}
		\label{fig:detailed_metrics_socepinsion}
	\end{subfigure}
	\begin{subfigure}[t]{0.23\textwidth}
		\centering
		\includegraphics[scale=0.23]{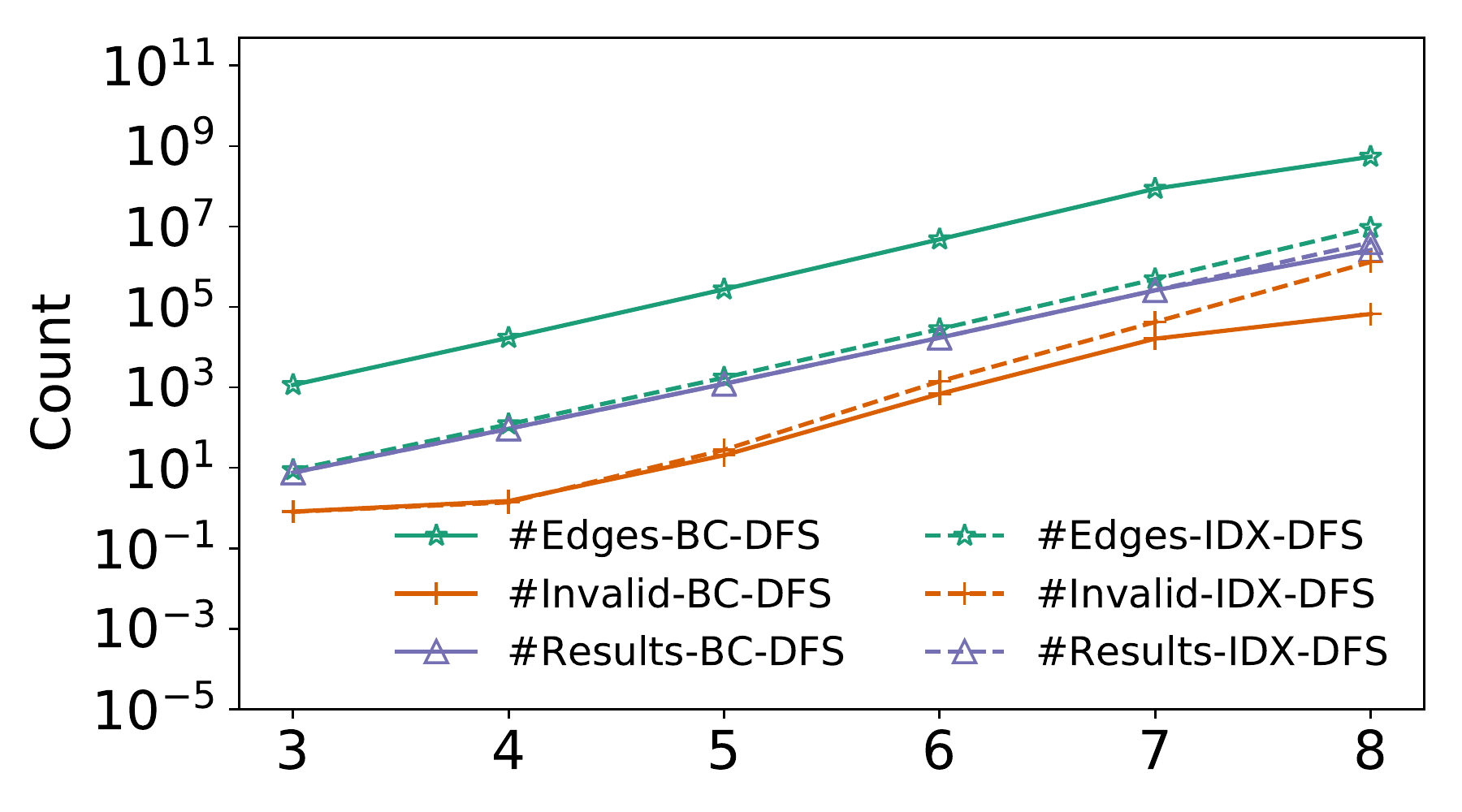}
		\caption{\emph{gg}.}
		\label{fig:detailed_metrics_webgoogle}
	\end{subfigure}
	\caption{Comparison of detailed metrics with $k$ varied.}
	\label{fig:overall_detailed_metrics}
\end{figure}

\SUN{\textbf{Query Time Breakdown.} Figure \ref{fig:query_time_breakdown} presents the query time breakdown of BC-DFS and IDX-DFS on \emph{ep}
and \emph{gg}. The \emph{preprocessing time} is the time on building index, while the \emph{enumeration time} is that on enumerating results. As shown in the figure, the preprocessing
dominates the query time when $k$ is small. IDX-DFS runs much faster than BC-DFS on both the preprocessing and enumeration (Note that y-axis is log-scale and
we terminate a query when it runs out of time).
The elapsed time of BC-DFS and IDX-DFS is close on \emph{ep} when $k = 8$ because a number of queries run out of time.}

\begin{figure}[h]\small
    \setlength{\abovecaptionskip}{0pt}
    \setlength{\belowcaptionskip}{0pt}
    \captionsetup[subfigure]{aboveskip=0pt,belowskip=0pt}
    \centering
    \begin{subfigure}[t]{0.23\textwidth}
        \centering
        \includegraphics[scale=0.23]{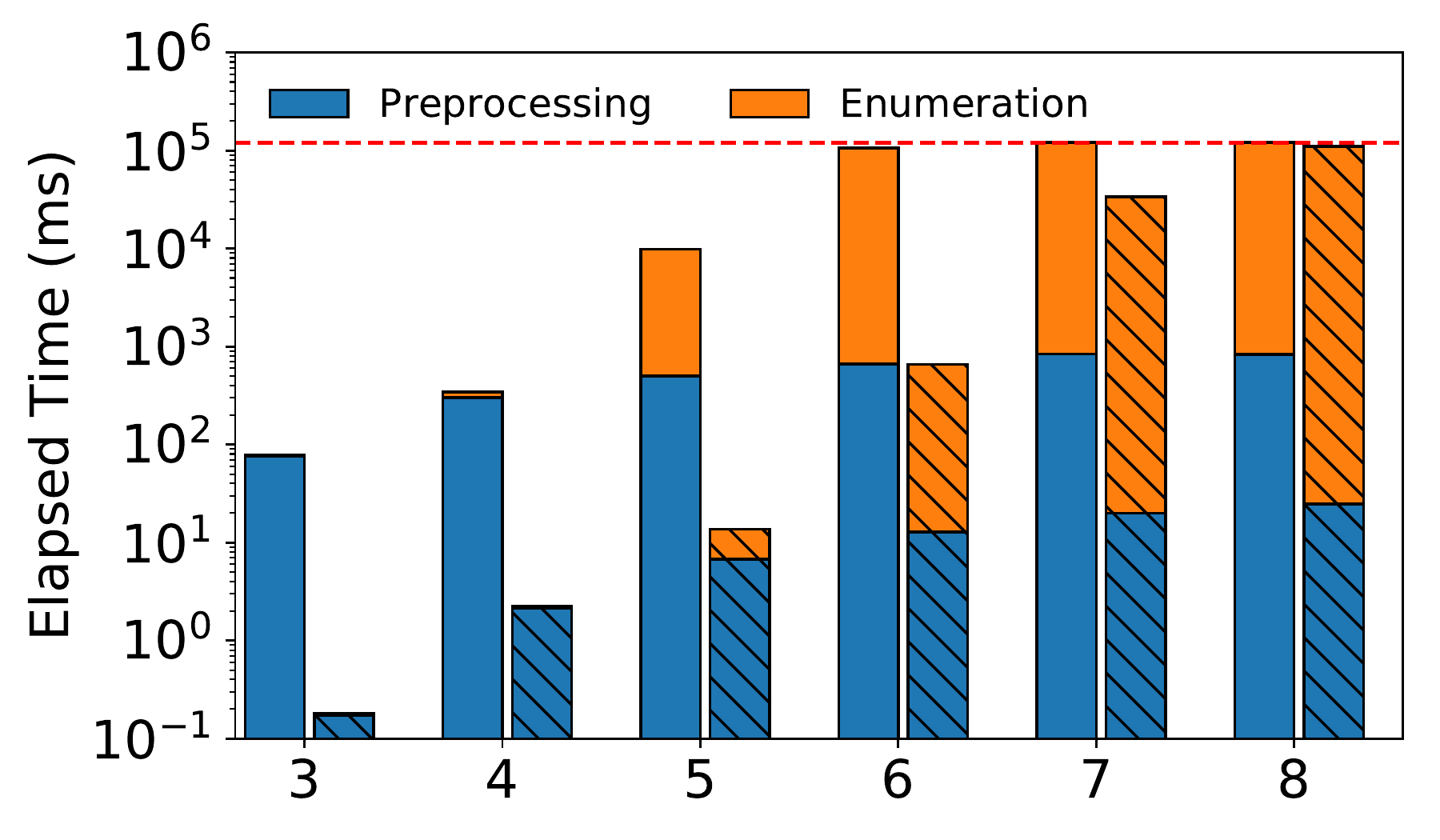}
        \caption{\emph{ep}.}
        \label{fig:query_time_breakdown_socepinsion}
    \end{subfigure}
    \begin{subfigure}[t]{0.23\textwidth}
        \centering
        \includegraphics[scale=0.23]{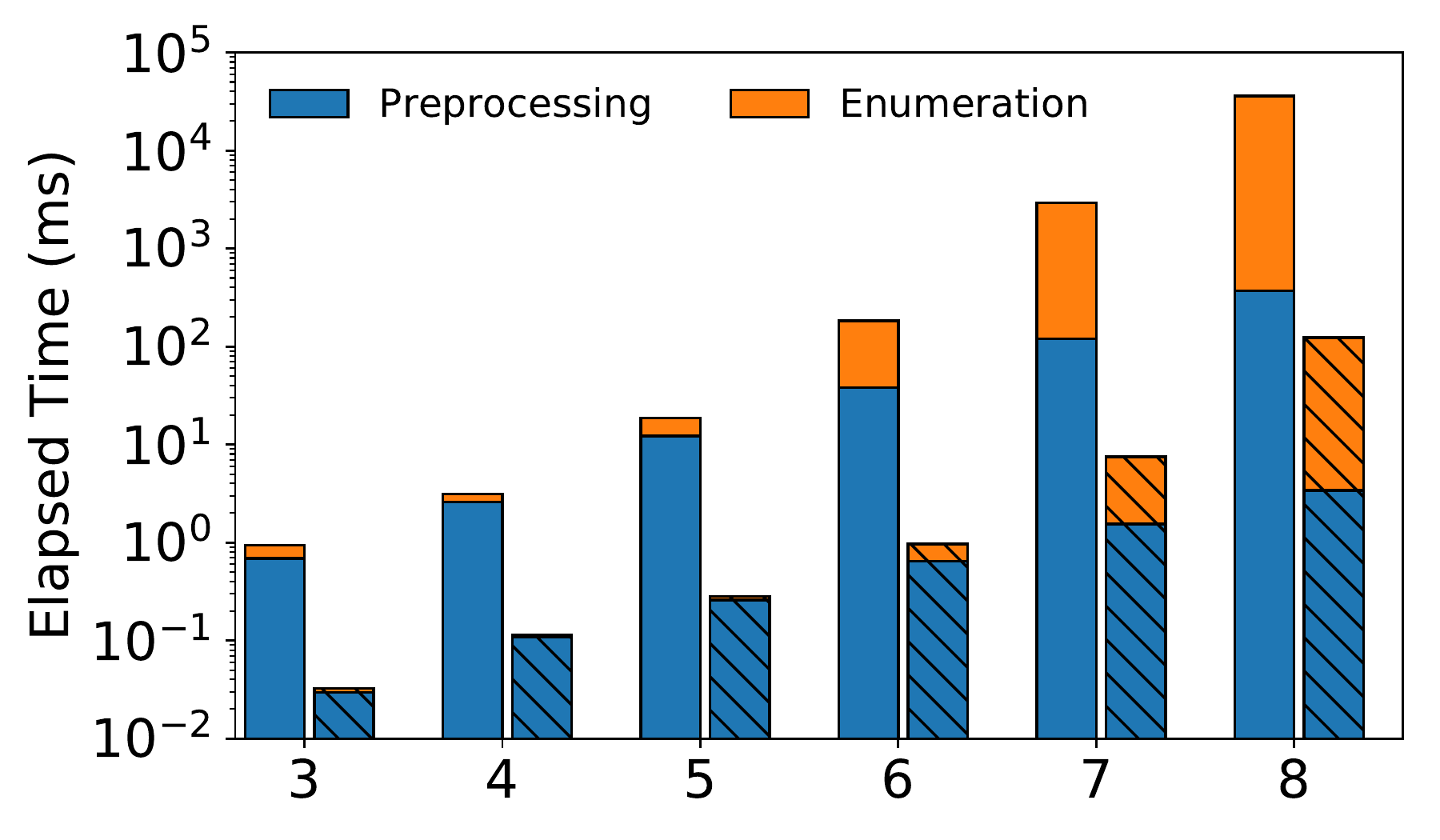}
        \caption{\emph{gg}.}
        \label{fig:query_time_breakdown_webgoogle}
    \end{subfigure}
    \caption{\SUN{Query time breakdown of BC-DFS (without hatches) and IDX-DFS (with hatches) with $k$ varied.}}
    \label{fig:query_time_breakdown}
\end{figure}

\SUN{\textbf{Query Time Distribution.} Table \ref{tb:query_time_distribution} counts the percentage of queries that can be completed within 60 seconds
(<60s) and that run out of time (>120s). Others can be finished between 60 seconds and 120 seconds.
We can see that the number of queries running out of time increases with $k$ varied from 3 to 8 on \emph{ep}. IDX-DFS significantly
outperforms BC-DFS, especially when $k$ is large. For example, IDX-DFS completes 23.1\% queries within 60 seconds, whereas BC-DFS only
completes 0.1\% queries. Furthermore, IDX-DFS completes all queries on \emph{gg} within 60 seconds.}

\setlength{\textfloatsep}{0pt}
\begin{table}[h]
 \small
 \centering
  \setlength{\abovecaptionskip}{0pt}
  \setlength{\belowcaptionskip}{0pt}
\caption{\SUN{Query time distribution on \emph{ep} and \emph{gg}.}}
\label{tb:query_time_distribution}
\begin{tabular}{c|cc|cc|cc|cc}
\hline
    & \multicolumn{4}{c|}{\emph{ep}}                                                   & \multicolumn{4}{c}{\emph{gg}}                                                   \\ \hline
    & \multicolumn{2}{c|}{BC-DFS}         & \multicolumn{2}{c|}{IDX-DFS}        & \multicolumn{2}{c|}{BC-DFS}         & \multicolumn{2}{c}{IDX-DFS}        \\ \hline
$k$ & \textless{}60s & \textgreater{}120s & \textless{}60s & \textgreater{}120s & \textless{}60s & \textgreater{}120s & \textless{}60s & \textgreater{}120s \\ \hline
3   & 1.00           & 0.00               & 1.00           & 0.00               & 1.00           & 0.00               & 1.00           & 0.00               \\
4   & 1.00           & 0.00               & 1.00           & 0.00               & 1.00           & 0.00               & 1.00           & 0.00               \\
5   & 0.959          & 0.016              & 1.00           & 0.00               & 1.00           & 0.00               & 1.00           & 0.00               \\
6   & 0.130          & 0.813              & 1.00           & 0.00               & 1.00           & 0.00               & 1.00           & 0.00               \\
7   & 0.003          & 0.997              & 0.877          & 0.066              & 0.994          & 0.003              & 1.00           & 0.00               \\
8   & 0.001          & 0.999              & 0.231          & 0.674              & 0.749          & 0.138              & 1.00           & 0.00               \\ \hline
\end{tabular}
\end{table}

\SUN{\textbf{Performance on Outlier Queries (queries running out of time).}
Moreover, we evaluate the performance of BC-DFS and IDX-DFS on short running queries (<60s) and long running
queries (>120s), respectively. Table \ref{tb:outlier_performance} presents throughput and response time on \emph{ep} with $k = 8$.
IDX-DFS runs much faster than BC-DFS in terms of both throughput and response time. The response time of IDX-DFS on short and long running queries is very close
and the value is small. Moreover, IDX-DFS has a high throughput on both short and long running queries, which indicates that the enumeration is efficient.
Therefore, IDX-DFS cannot complete the outlier queries because these queries have a large number of results.}

\setlength{\textfloatsep}{0pt}
\begin{table}[h]
 \small
 \centering
  \setlength{\abovecaptionskip}{0pt}
  \setlength{\belowcaptionskip}{0pt}
\caption{\SUN{Performance for queries with different query time on \emph{ep} with $k = 8$.}}
\label{tb:outlier_performance}
\begin{tabular}{c|cc|cc}
\hline
                & \multicolumn{2}{c|}{\textbf{Throughput}} & \multicolumn{2}{c}{\textbf{Response Time (ms)}} \\ \hline
\textbf{Method} & \textless{}60s    & \textgreater{}120s   & \textless{}60s        & \textgreater{}120s       \\ \hline
BC-DFS          & 1.52e+04          & 8.65e+03             & 6.11e+02              & 3.03e+03                 \\
IDX-DFS         & 7.83e+06          & 5.13e+07             & 2.12e+01              & 2.52e+01                 \\ \hline
\end{tabular}
\end{table}

\SUN{\textbf{Performance on Dynamic Graphs.} We compare the performance of BC-DFS and IDX-DFS on dynamic graphs. Following experiments in \cite{peng2019towards},
we randomly select 10\% edges of \emph{ep} and \emph{gg} as updates and keep subgraphs on remaining edges as initial graphs.
For each selected edge $e(v, v')$, we set $v'$ and $v$ as $s$ and $t$, respectively, and enumerate the hop-constrained paths. As the index
is built for each query online, our method can directly process dynamic graphs. We examine the 99.9\% latency of BC-DFS and IDX-DFS
in terms of the response time. Figure \ref{fig:latency} presents the results on \emph{ep} and
\emph{gg} with $k$ varied. As show in the figure, IDX-DFS significantly outperforms BC-DFS.
The 99.9\% latency of IDX-DFS on \emph{ep} with $k$ varied from $3$ to $7$ is within 0.1s. The 99.9\% latency of IDX-DFS on \emph{gg}
is less than 0.1s.}

\begin{figure}[h]\small
    \setlength{\abovecaptionskip}{0pt}
    \setlength{\belowcaptionskip}{0pt}
    \captionsetup[subfigure]{aboveskip=0pt,belowskip=0pt}
    \centering
    \begin{subfigure}[t]{0.23\textwidth}
        \centering
        \includegraphics[scale=0.23]{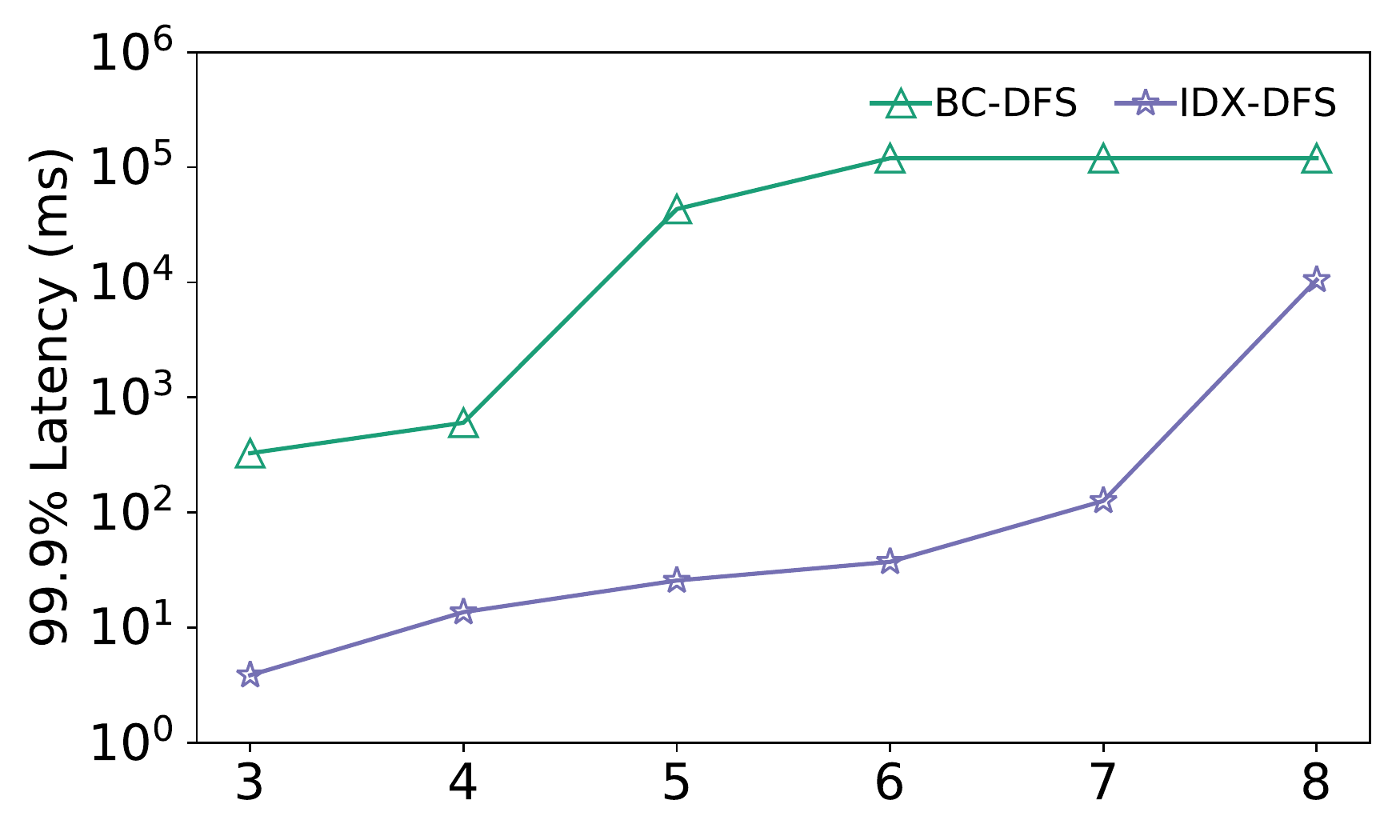}
        \caption{\emph{ep}.}
        \label{fig:latency_ep}
    \end{subfigure}
    \begin{subfigure}[t]{0.23\textwidth}
        \centering
        \includegraphics[scale=0.23]{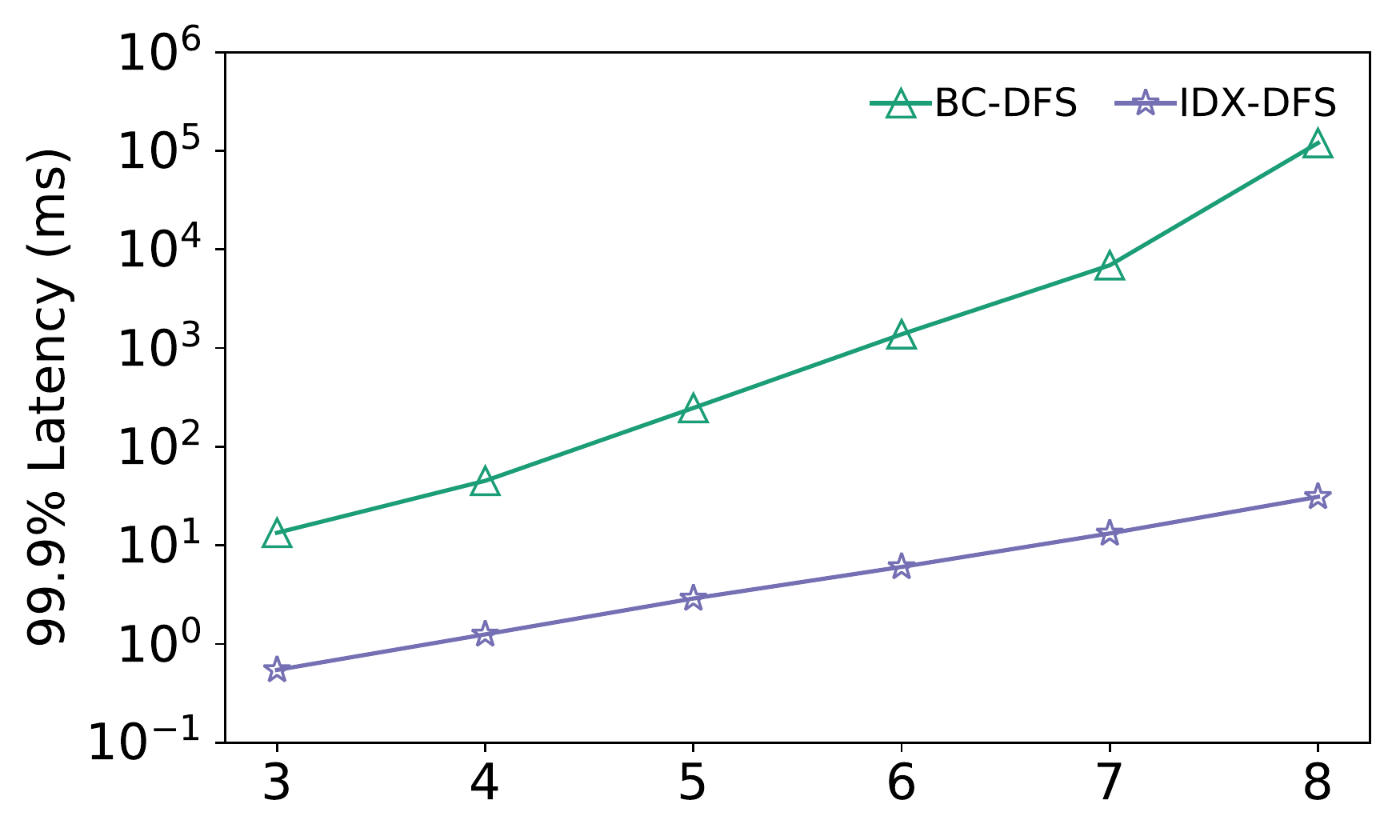}
        \caption{\emph{gg}.}
        \label{fig:latency_gg}
    \end{subfigure}
    \caption{\SUN{Comparison of 99.9\% latency with $k$ varied.}}
    \label{fig:latency}
\end{figure} 


\subsection{Evaluation of Individual Techniques}

\SUN{\textbf{Spectrum Analysis.} We conduct the spectrum analysis to study the effectiveness of our join plan optimization method. Particularly,
given $Q$, we categorize all join plans into the left deep tree and the bushy tree based on the shape of join trees. The left deep tree
extends partial results by a vertex at a step (e.g., Algorithm 4), whereas the bushy tree performs the join on partial results of two
sub-queries of $Q$ (e.g., Algorithm 6). We enumerate all left deep trees of $Q$ without the Cartesian product.
In contrast, for the bushy tree, we consider all cut positions $i$ where $0 < i < k$ that divides $Q$ into two sub-queries
$Q[0:i]$ and $Q[i:k]$ and evaluate $Q[0:i]$ and $Q[i:k]$ with the depth-first search method because the join of $Q[0:i]$ and $Q[i:k]$
dominates the cost.}

\begin{figure}[h]\small
    \setlength{\abovecaptionskip}{0pt}
    \setlength{\belowcaptionskip}{0pt}
    \captionsetup[subfigure]{aboveskip=0pt,belowskip=0pt}
    \centering
    \begin{subfigure}[t]{0.23\textwidth}
        \centering
        \includegraphics[scale=0.23]{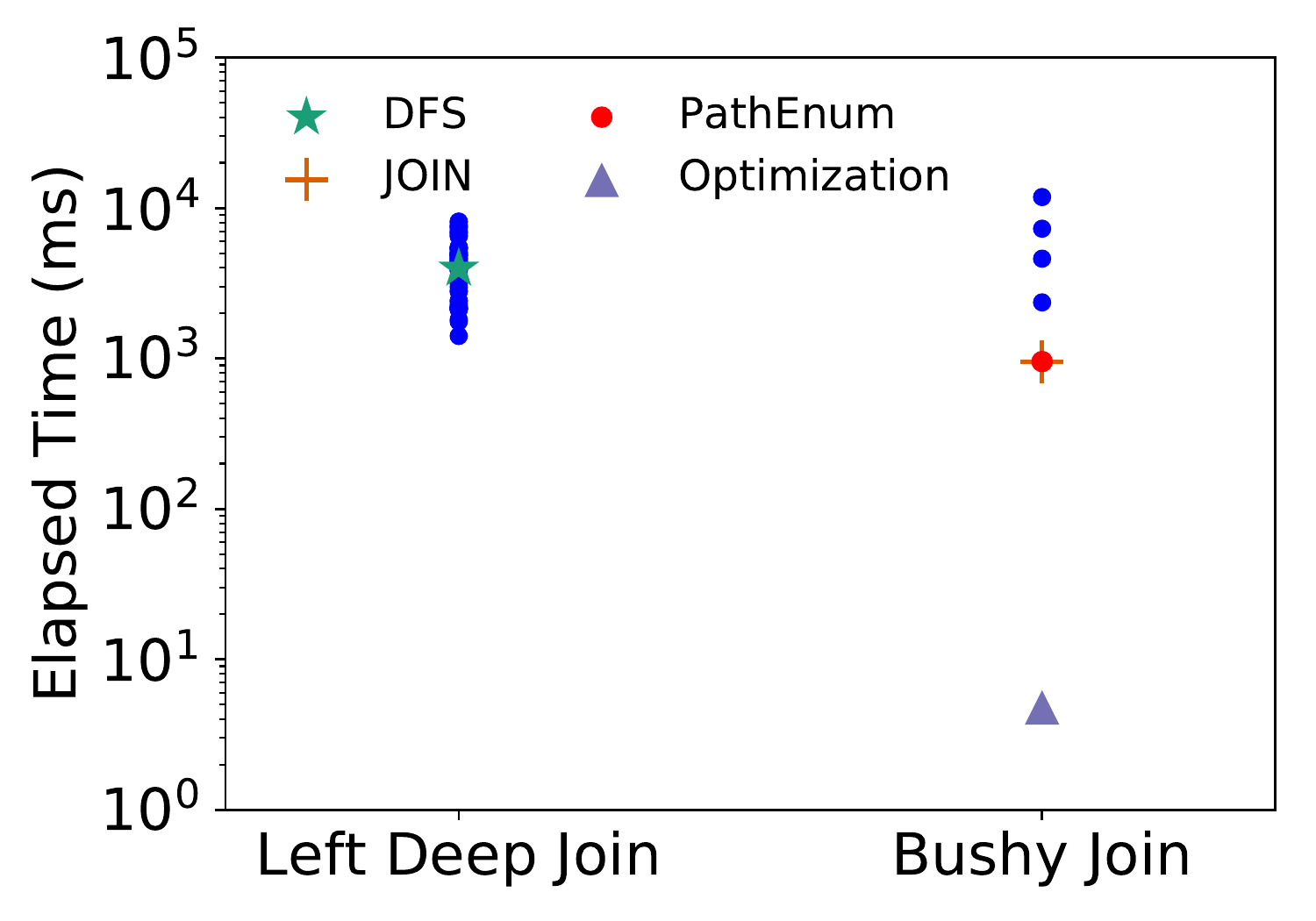}
        \caption{\emph{ep}.}
        \label{fig:spectrum_analysis_ep}
    \end{subfigure}
    \begin{subfigure}[t]{0.23\textwidth}
        \centering
        \includegraphics[scale=0.23]{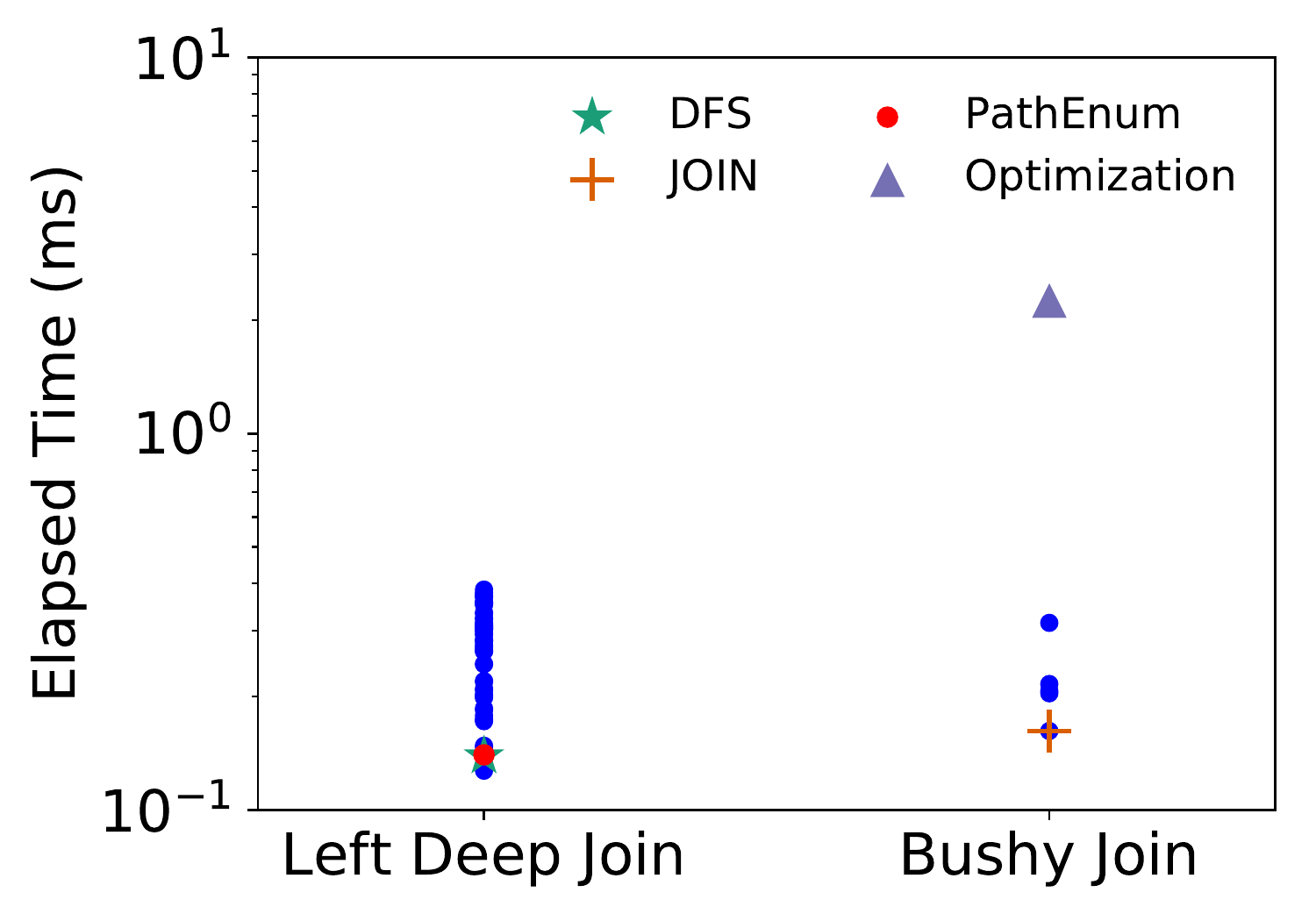}
        \caption{\emph{gg}.}
        \label{fig:spectrum_analysis_gg}
    \end{subfigure}
    \caption{\SUN{Spectrum analysis of join plan optimization.}}
    \label{fig:spectrum_analysis}
\end{figure}

\SUN{Figure \ref{fig:spectrum_analysis} presents the results of a query with $k = 6$ on \emph{ep} and \emph{gg}.
"\emph{DFS}" and "\emph{JOIN}" denotes the enumeration time of
Algorithms \ref{algo:dfs_on_index} and \ref{algo:join_on_index}, respectively.
"\emph{Optimization}" represents the time spent on optimizing the join order
(Algorithm \ref{algo:generate_join_order}). "\emph{PathEnum}" denotes the sum of the enumeration time
and the query optimization time of PathEnum. Each blue point denotes the time on enumerating all results based on the index with a join plan. 
In Figure \ref{fig:spectrum_analysis_ep}, the optimization time is much shorter than the enumeration time and the optimal plan is a bushy tree.
In Figure \ref{fig:spectrum_analysis_gg}, the optimization time is longer than the enumeration time. PathEnum takes shorter time than
the optimization because the preliminary estimator decides to use IDX-DFS directly. So our join optimizer is effective.
Nevertheless, the query optimizer can be further improved by considering a larger plan space because our method
considers only one plan with the left-deep tree (i.e., the order from $s$ to $t$) and the optimal plan
can fall outside of our plan space.}

\SUN{\textbf{Factors on Query Efficiency.} We examine the impact of the index size and the number of results on the enumeration time, respectively.
The index size is measured by the number of edges in the index. As the enumeration time, the index size and the number of results
vary greatly on different queries, we perform the linear regression analysis on the logarithm values of these metrics.
Figures \ref{fig:index_enumeration} and \ref{fig:result_enumeration} present the results of IDX-DFS on \emph{ep} and \emph{gg} with $k = 6$.
A blue point represents the result of a query and the red line denotes the underlying relationship obtained by the linear regression
model. The enumeration time increases with the index size and \#results increasing. Moreover, the
enumeration time has a closer relationship with \#results than the index size.}

\begin{figure}[h]\small
	\setlength{\abovecaptionskip}{0pt}
	\setlength{\belowcaptionskip}{0pt}
	\captionsetup[subfigure]{aboveskip=0pt,belowskip=0pt}
	\centering
	\begin{subfigure}[t]{0.23\textwidth}
		\centering
		\includegraphics[scale=0.23]{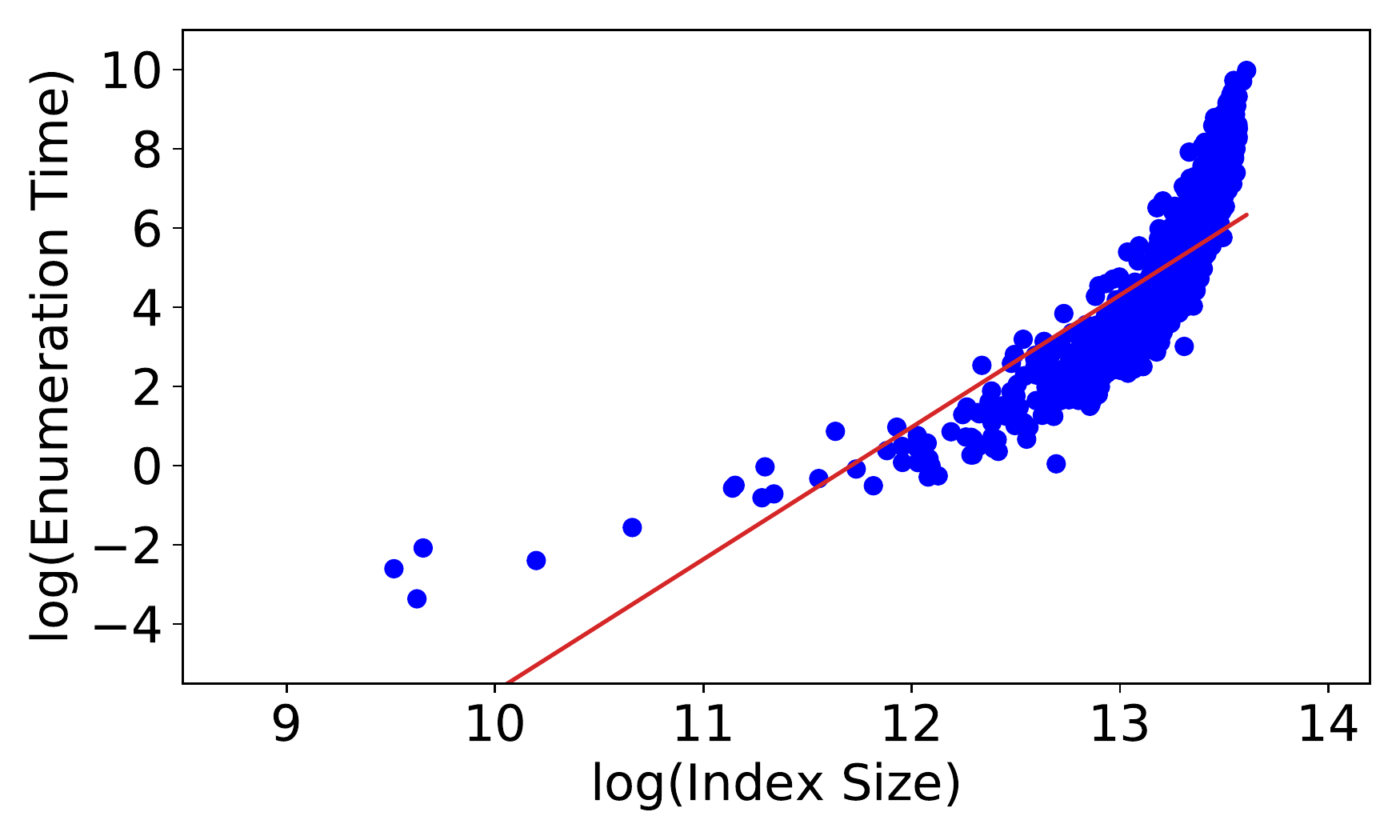}
		\caption{\emph{ep}.}
		\label{fig:index_enumeration_ep}
	\end{subfigure}
	\begin{subfigure}[t]{0.23\textwidth}
		\centering
		\includegraphics[scale=0.23]{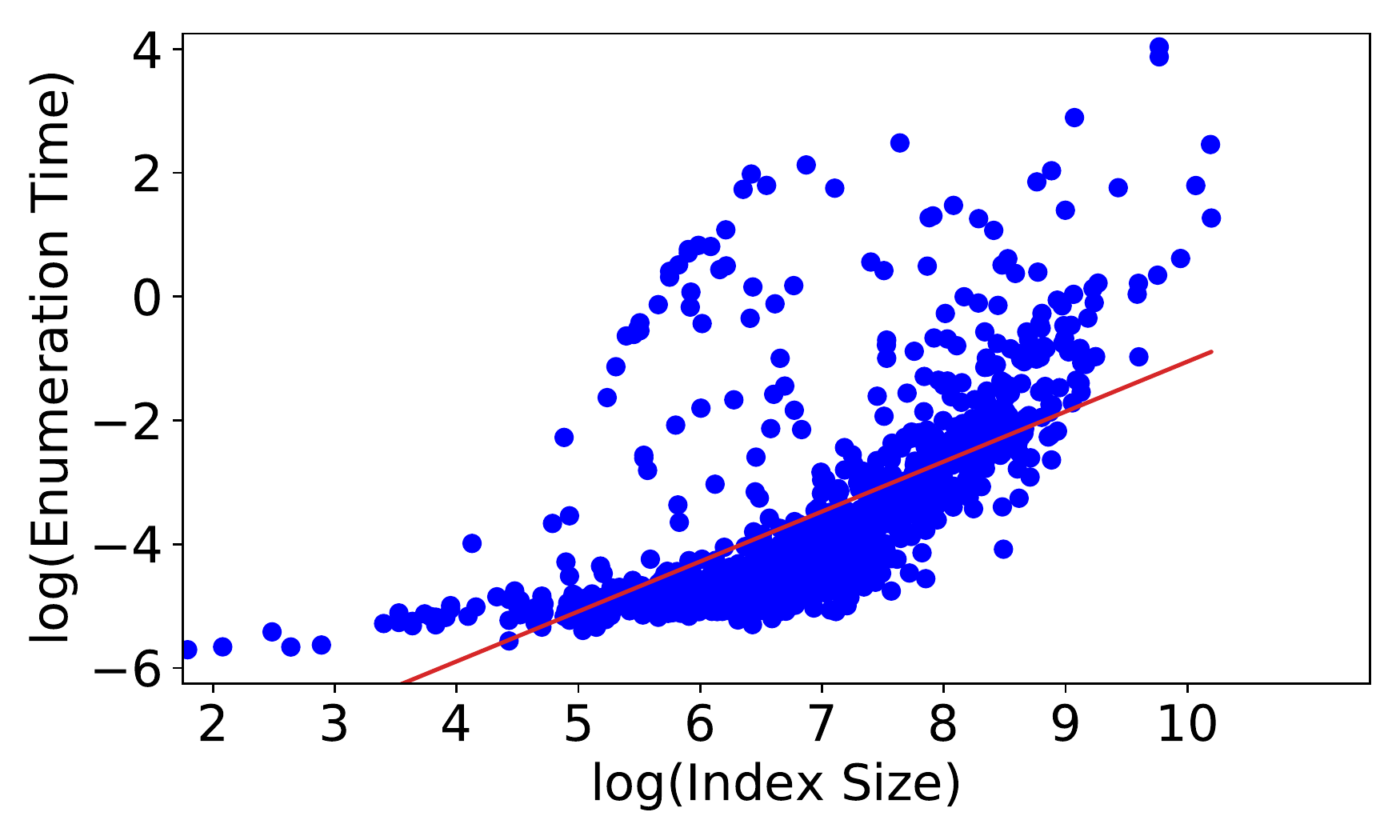}
		\caption{\emph{gg}.}
		\label{fig:index_enumeration_gg}
	\end{subfigure}
	\caption{\SUN{Impact of index size on enumeration time.}}
	\label{fig:index_enumeration}
\end{figure}

\begin{figure}[h]\small
	\setlength{\abovecaptionskip}{0pt}
	\setlength{\belowcaptionskip}{0pt}
	\captionsetup[subfigure]{aboveskip=0pt,belowskip=0pt}
	\centering
	\begin{subfigure}[t]{0.23\textwidth}
		\centering
		\includegraphics[scale=0.23]{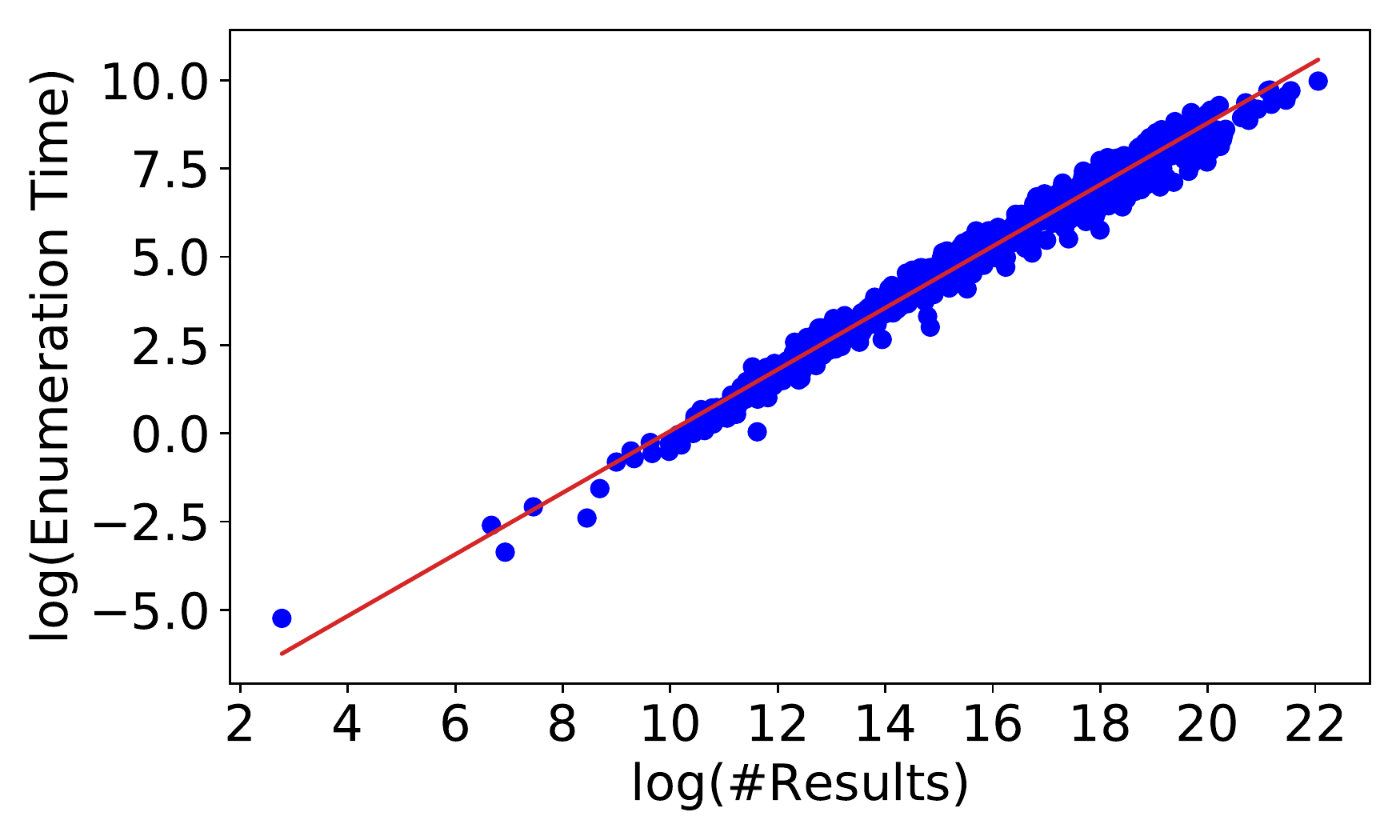}
		\caption{\emph{ep}.}
		\label{fig:result_enumeration_ep}
	\end{subfigure}
	\begin{subfigure}[t]{0.23\textwidth}
		\centering
		\includegraphics[scale=0.23]{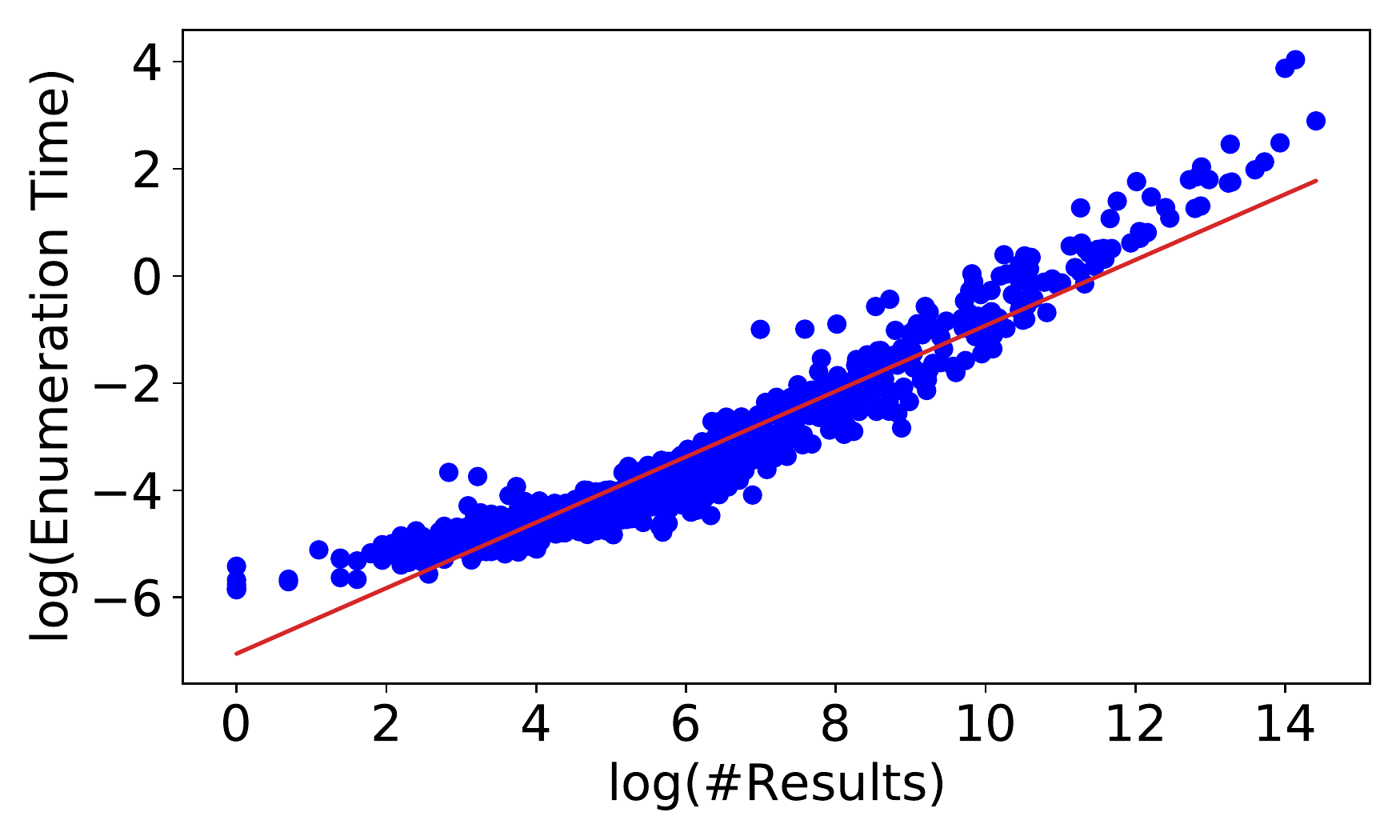}
		\caption{\emph{gg}.}
		\label{fig:result_enumeration_gg}
	\end{subfigure}
	\caption{\SUN{Impact of \#results on enumeration time.}}
	\label{fig:result_enumeration}
\end{figure}

\SUN{\textbf{Average and Maximum Number of Results.} Moreover, we examine the average and maximum number of results reported
on \emph{ep} and \emph{gg} with $k$ varied. Table \ref{tb:num_results} presents the experiment results.
The star symbol denotes that we cannot enumerate all results within 120 seconds
and report the value found by IDX-DFS within the time limit. We can see that the number of results significantly increases
with $k$ varied from 3 to 8, and the number of results on \emph{ep} is much more than taht \emph{gg}. Therefore, the query time on
\emph{ep} is longer than that on \emph{gg}, and the query time significantly increases with the increasing of $k$ as shown in Figure \ref{fig:query_time_breakdown}.}

\setlength{\textfloatsep}{0pt}
\begin{table}[h]
 \small
 \centering
  \setlength{\abovecaptionskip}{0pt}
  \setlength{\belowcaptionskip}{0pt}
\caption{\SUN{The average and maximum number of results reported on \emph{ep} and \emph{gg}. The star symbol denotes that we cannot enumerate all results within 120 seconds.}}
\label{tb:num_results}
\resizebox{0.48\textwidth}{!}{
\begin{tabular}{c|c|cccccc}
\hline
                    & $k$          & \textbf{3} & \textbf{4} & \textbf{5} & \textbf{6} & \textbf{7} & \textbf{8}     \\ \hline
\multirow{2}{*}{\emph{ep}} & avg & 5.06e+1   & 5.16e+3   & 5.34e+5   & 5.54e+7   & 3.00e+9   & ${1.24e+10}^*$       \\
                    & max & 3.30e+3   & 3.48e+5   & 3.64e+7   & 3.79e+9   & 2.74e+10   & ${2.28e+10}^*$ \\ \hline
\multirow{2}{*}{\emph{gg}} & avg & 7.58e+0   & 9.33e+1   & 1.22e+3   & 1.71e+4   & 2.59e+5   & 4.03e+6       \\
                    & max & 3.78e+2   & 6.69e+3   & 1.13e+5   & 1.81e+6   & 3.40e+7   & 7.24e+8       \\ \hline
\end{tabular}
}
\end{table}

\SUN{\textbf{Memory Cost.} Table \ref{tb:memory_cost} presents the maximum memory consumption on indexes and partial results
of IDX-JOIN with $k$ varied. The index consumes a small amount of memory space because
the space complexity of the index is $O(|E(G)| + k \times |V(G)|)$ and the filtering can effectively prune some
vertices and edges. The partial results of IDX-JOIN on \emph{ep} consume much more space than that on \emph{gg} because
there are more results on \emph{ep} than \emph{gg} as shown in Table \ref{tb:num_results}. In summary, the index takes small space, while
the partial results of IDX-JOIN can consume a large amount of space due to the large number of results.}

\setlength{\textfloatsep}{0pt}
\begin{table}[h]
 \small
 \centering
  \setlength{\abovecaptionskip}{0pt}
  \setlength{\belowcaptionskip}{-5pt}
\caption{\SUN{Maximum memory consumption (MB) on \emph{ep} and \emph{gg}.}}
\label{tb:memory_cost}
\begin{tabular}{c|c|cccccc}
\hline
\textbf{}                                                                           & \textbf{$k$} & \textbf{3} & \textbf{4} & \textbf{5} & \textbf{6} & \textbf{7} & \textbf{8} \\ \hline
\multirow{2}{*}{\textbf{Index}}                                                     & \textit{ep}  & 0.15       & 1.68       & 3.28       & 4.60       & 5.45       & 5.91       \\
                                                                                    & \textit{gg}  & 0.01       & 0.10       & 0.11       & 0.21       & 0.44       & 0.63       \\ \hline
\multirow{2}{*}{\textbf{\begin{tabular}[c]{@{}c@{}}Partial\\ Results\end{tabular}}} & \textit{ep}  & 0.03       & 0.66       & 26.51      & 138.32     & 4561.59    & 21479.32   \\
                                                                                    & \textit{gg}  & 0.01       & 0.02       & 0.37       & 1.05       & 17.80      & 55.70      \\ \hline
\end{tabular}
\end{table}

\SUN{\textbf{Supplement Experiments.} The appendix presents more experiment results including the comparison of throughput, query time
and response time with $k$ varied, the cumulative distribution function of query time, the time efficiency of individual
techniques (e.g., index construction) with $k$ varied, and the effectiveness of cardinality estimators.}

\subsection{Scalability Evaluation}

We evaluate the scalability of IDX-DFS and IDX-JOIN with \emph{tm} that has around two billion edges.
Figure \ref{fig:scalability_evaluation} presents the execution time of each individual technique and
the throughput with $k$ varied from 3 to 6. IDX-JOIN runs out of memory when $k = 6$. Therefore, we omit
its results on this case. "\emph{Index construction}" denotes the time spent on building the index (Algorithm
\ref{algo:build_index}). Additionally,
we report the time of computing the distance of each vertex to $s,t$, which is denoted by \emph{BFS}. \emph{BFS} is included
in \emph{Index construction}. As shown in Figure \ref{fig:scalability_execution_time}, Algorithm \ref{algo:build_index}
spends tens of seconds on the index construction, which is dominated by \emph{BFS}. The time spent on building the index
and generating join orders is more than that on enumerating results when $k$ varied from 3 to 4. Despite the long
preprocessing time, the throughput of both IDX-DFS and IDX-JOIN is up to $10 ^ 7$ when $k = 5$, which demonstrates
the efficiency of the index-based enumeration.

\begin{figure}[h]\small
    \setlength{\abovecaptionskip}{0pt}
    \setlength{\belowcaptionskip}{0pt}
    \captionsetup[subfigure]{aboveskip=0pt,belowskip=0pt}
    \centering
    \begin{subfigure}[t]{0.23\textwidth}
        \centering
        \includegraphics[scale=0.23]{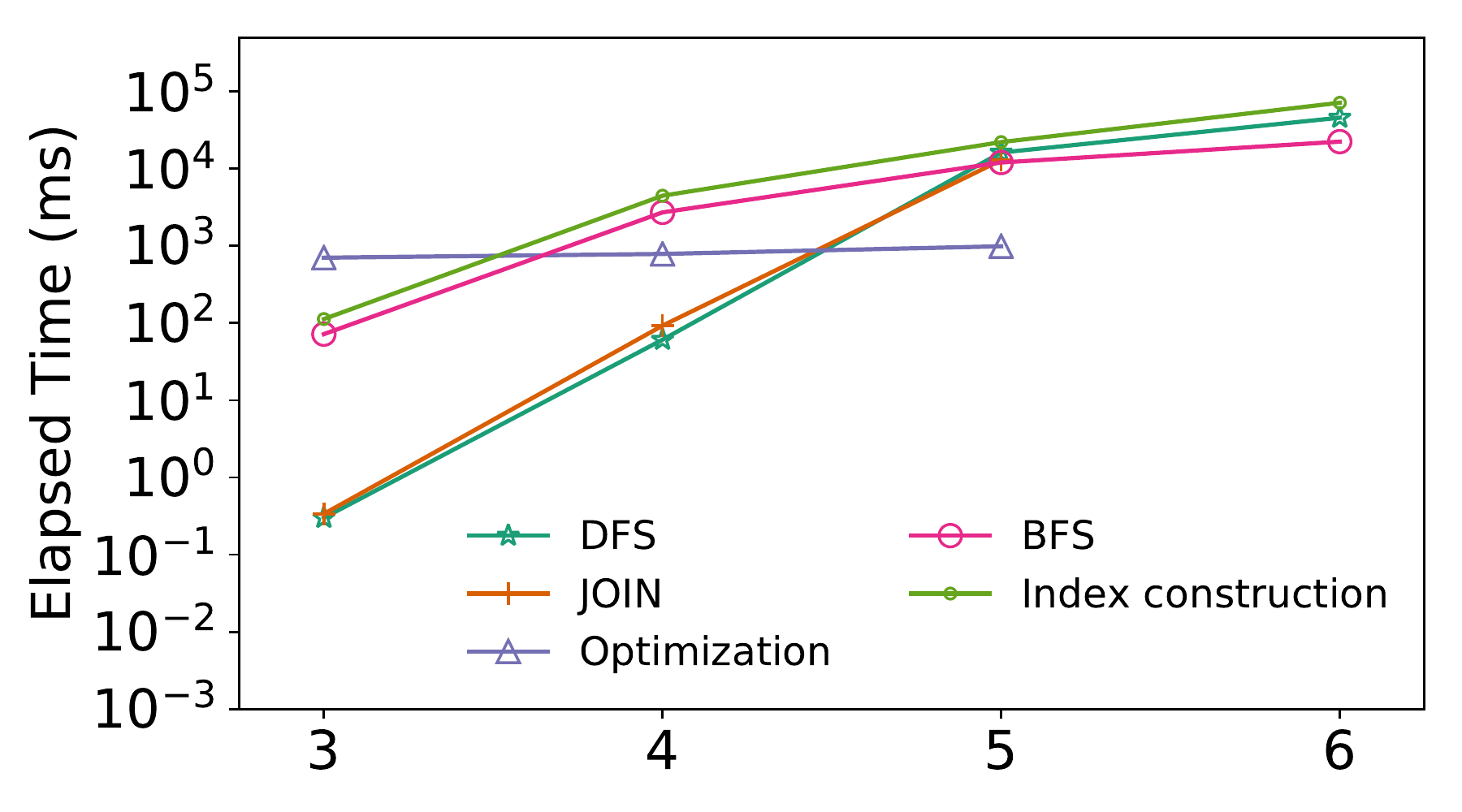}
        \caption{Execution time.}
        \label{fig:scalability_execution_time}
    \end{subfigure}
    \begin{subfigure}[t]{0.23\textwidth}
        \centering
        \includegraphics[scale=0.23]{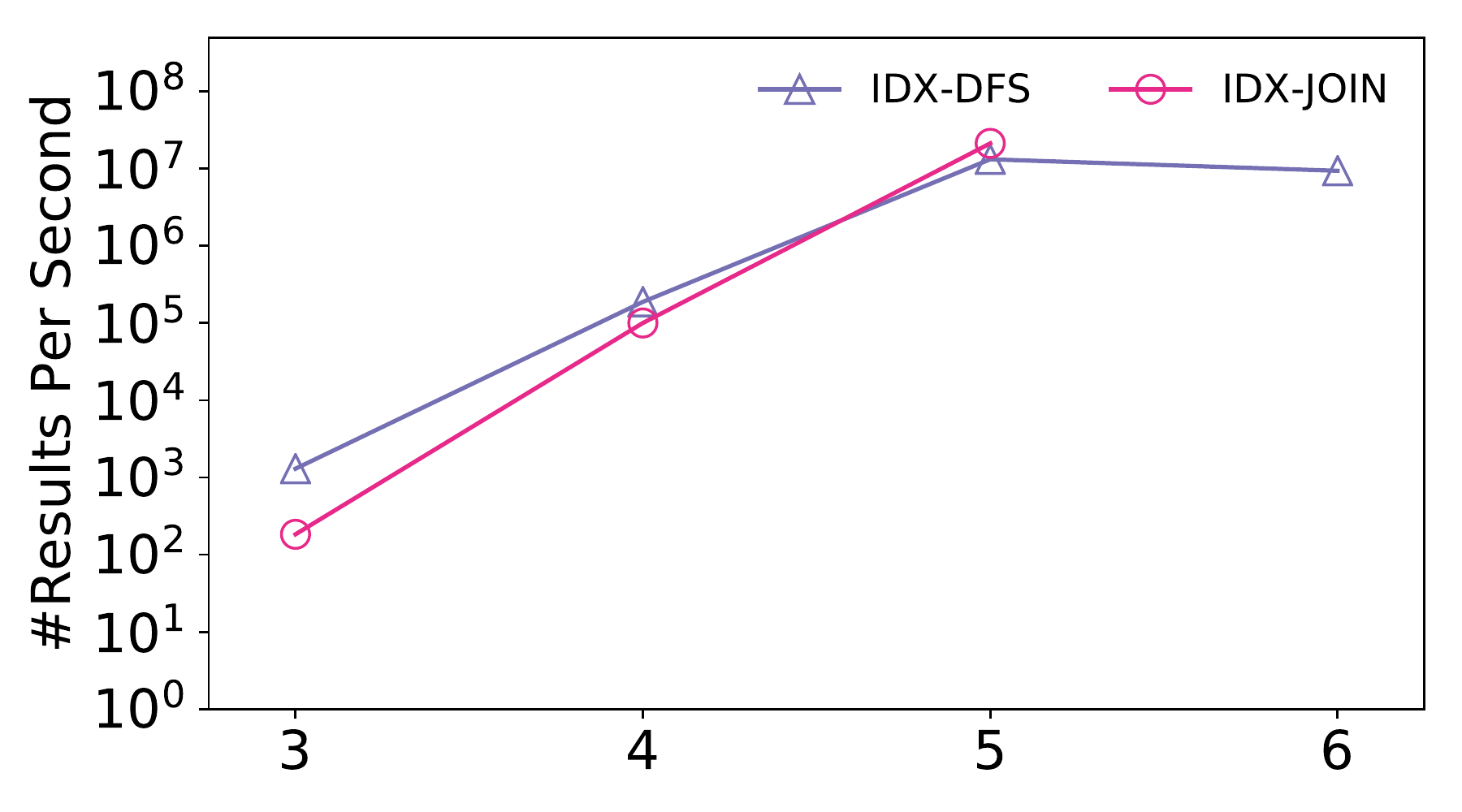}
        \caption{Throughput.}
        \label{fig:scalability_throughput}
    \end{subfigure}
    \caption{Scalability evaluation on \emph{tm} with $k$ varied.}
    \label{fig:scalability_evaluation}
\end{figure}

\SUN{\subsection{Discussions}}

\SUN{Although PathEnum significantly accelerates HcPE queries, there still leaves interesting future work. First, our join optimizer
can be further improved by searching the optimal plan in a larger plan space and considering more metrics such as the cost of materializing
partial results. Second, developing algorithms having a short response time on very large graphs
is an interesting research direction because building the index from scratch on very large graphs can take a long time (e.g., tens of seconds on \emph{tm}).
A promising approach is to build a global index in an offline preprocessing step to reduce the cost of construing the query-dependent index.
However, designing an effective global index is challenging because (1) such an index has to maintain the global statics of $G$ to
serve all queries and therefore must balance the cost of the index and query efficiency (e.g., recording distance between all vertex pairs
is unacceptable due to the large space overhead); and (2) the index needs to support efficient update operations to serve dynamic graphs.}

\SUN{Additionally, we observe some opportunities for graph database systems \cite{aberger2017emptyheaded,mhedhbi2019optimizing} to explore.
The query-dependent index can reduce elements involved in the computation and provide accurate statistics to the query optimizer.
This gives graph databases an alternative way of evaluating queries by dividing the evaluation into two phases: (1) builds a query-dependent index;
and (2) generates the query plan and computes based on the index. Moreover, the systems can adopt an adaptive query optimizer to
process queries because query time of different queries can vary greatly.}



%% file: 8_conclusion.tex
\section{Conclusions}

In this paper, we study the hop-constrained \emph{s-t} path enumeration problem, and propose PathEnum, an efficient algorithm towards addressing real-time requirements from many on-line applications. We design a light-weight index, and two index-based approaches for efficient enumerations. We further develop a query optimizer to optimize the join order and
decide which approach to use at per query basis. We conduct extensive experiments with a variety of real-world graphs, and show that PathEnum achieves orders of magnitude speedup over the state-of-the-art approaches.


%% file: 9_appendix.tex
\appendix

\section{Correctness of Join-based Model}

In this section, we prove that the join-based model in Section \ref{sec:join_based_model} is correct.
Given a graph $G$ and a HcPE query $q(s, t, k)$, suppose that the relations of $Q$ is generated based on the method in
Section \ref{sec:join_based_model}. Given a tuple $r \in Q$, $r[i]$ represents
the vertex at position $i$, and $l ^ *$ denotes the first position that $t$ appears in $r$. Let $r[0:i]$ be the vertices
from positions $0$ to $i$. We first prove that $Q$ satisfies the following lemma.

\begin{lemma} \label{lemma:result_to_walk}
    Given a tuple $r \in Q$, $r[0:l ^ *]$ is a walk from $s$ to $t$, and each vertex $v \in r[l^*:k]$ is $t$.
\end{lemma}

\begin{proof}
    Based on the first property of relations, $r[0]$ and $r[k]$ must be $s$ and $t$, respectively.
    Given $0 < i \leqslant l ^ *$, $(r[i-1], r[i])$ is an edge in $E(G)$ because it is a tuple belonging to $R_i$.
    Thus, $r[0:l^*]$ is a walk from $s$ to $t$. Next, we prove that each vertex $v \in r[l^*:k]$ is $t$ by contradiction.
    Without loss of generality, assume that $r[i]$ is the first vertex not equal to $t$ where $l^* < i < k$.
    Then, $(t, r[i])$ belongs to $R_i$. However, the second property guarantees that there is no tuple starting
    from $t$ except $(t, t)$, which contradicts the assumption. Thus, the lemma is proved.
\end{proof}
 
Moreover, the following lemma holds.
 
\begin{lemma} \label{lemma:walk_to_result}
    Given a walk $W \in \mathcal{W}(s, t, k, G)$, there is a tuple $r \in Q$ such that $W$ is equal to $r[0:l ^ *]$.
\end{lemma}

\begin{proof}
    Given $W \in \mathcal{W}(s, t, k, G)$, we first construct a tuple $r$ as follows: $r[0:l^*] = W$ and set
    each vertex in $r[l^*:k]$ as $t$. Next, we prove that $r$ belongs to $Q$. Given $1 \leqslant i \leqslant l^*$, $(r[i-1], r[i])$
    exists in $R_i$ according to the generation method of relations. Let $Q'$ be $R_1 \Join \dots \Join R_{l^*}$. Then,
    $r[0:l^*]$ belongs to $Q'$. As $(t,t) \in R_i$ where $l^* < i \leqslant k$, $r[0:l^*]$ will be extended by adding
    $t$ when performing the join operation on $Q'$ and the remaining relations. Therefore, $r$ appears in the final results.
    Moreover, as the join operation satisfies the commutative and associative laws, $r$ belongs to $Q$ regardless of the join order.
    Thus, the lemma is proved.
\end{proof}

Based on Lemmas \ref{lemma:result_to_walk} and \ref{lemma:walk_to_result}, we can prove the correctness of Theorem \ref{theorem:correctness} as follows.

\begin{proof}
    $Q$ has no duplicate results because tuples in each relation of $Q$ are distinct. According to Lemmas \ref{lemma:result_to_walk}
    and \ref{lemma:walk_to_result}, $Q$ contains all walks in $\mathcal{W}(s, t, k, G)$. Therefore, we can obtain all paths from $s$ to $t$ by
    eliminating results in $Q$ contain duplicate vertices except $t$. The theorem is proved.
\end{proof}

\section{Pruning Power Comparison}

We compare the pruning power of Algorithm \ref{algo:build_relations}
with that of Algorithm \ref{algo:build_index} in this section. Let $R_i(u_{i - 1}, u_i)$ be the relation generated by
Algorithm \ref{algo:build_relations}. Given $1 \leqslant i \leqslant k$, $R_i(u_{i - 1}:v, u_i)$ represents the
neighbors of $v$ in $R_i$, and $C(u_{i - 1})$ denotes all values of $u_{i - 1}$ in $R_i$, i.e., $\{v|(v, v')\in R_i\}$.
Because $v \in C(u_{i - 1})$, there exists a walk $W \in \mathcal{W}(s, t, k, G)$ such that $W[i - 1] = v$
according to Proposition \ref{lemma:full_reducer}. Therefore, $v$ belongs to $X$ in $\mathcal{I}$ according to Proposition
\ref{prop:prune_vertex}. Next, we prove that given $v \in C(u_{i - 1})$ where $v \neq t$, $R_i(u_{i - 1}:v, u_i)$ is
equal to $\mathcal{I}_t(v, k - i)$ by contradiction.

Assume that $v' \in R_i(u_{i - 1}:v, u_{i})$ but $v' \notin \mathcal{I}_t(v, k - i)$.
Then, $S(v', t|G - \{s\}) \leqslant k - i$ because there exists $W \in \mathcal{W}(s, t, k, G)$ such that $W[i] = v'$ according
to Proposition \ref{lemma:full_reducer}. Based on Algorithm \ref{algo:build_index}, if
$v' \in N(v)$ and $S(v', t|G - \{s\}) \leqslant k - i$, then $v'$ belongs to $\mathcal{I}_t(v, k-i)$, which contradicts the assumption.
Therefore, $R_i(u_{i - 1}:v, u_{i}) \subseteq \mathcal{I}_t(v, k-i)$. Next, assume that $v' \in \mathcal{I}_t(v, k - i)$
but $v' \notin R_i(u_{i - 1}:v, u_{i})$. Therefore, there is a walk $W$ from $v'$ to $t$ such that $L(W) \leqslant k - i$
because $S(v', t| G - \{s\}) \leqslant k - i$ according to Algorithm \ref{algo:build_index}. Moreover, there is a walk $W'$
from $s$ to $v$ such that $W[i - 1] = v$ because $v \in C(u_{i - 1})$. Then, we can construct a
walk from $s$ to $t$ by concatenating $W'$ and $W$. According to Proposition \ref{lemma:full_reducer},
$(v, v')$ must exist in $R_i(u_{i - 1}, u_{i})$, which contradicts the assumption.
So we have $\mathcal{I}_{t}(v, k - i) \subseteq R_i(u_{i - 1}:v, u_{i})$.
Based on the analysis, we can see that our index provides competitive pruning power with relations generated by
Algorithm \ref{algo:build_relations}.

\section{Correctness of Our Algorithms}

In the following, we first prove the correctness of Algorithms \ref{algo:dfs_on_index}.

\begin{proposition} \label{prop:search_correctness}
    Algorithm \ref{algo:dfs_on_index} finds all $k$ hop-constrained paths $\mathcal{P}(s, t, k, G)$ from $s$ to $t$ in $G$.
\end{proposition}

\begin{proof}
    We first prove that $M$ emitted by the algorithm is a path $P \in \mathcal{P}(s, t, k, G)$. Lines 1 and 4 ensure that $M$ begins with $s$, while ends with $t$.
    The check at Line 7 keeps that $M$ contains no duplicate vertices. Based on the construction method of $\mathcal{I}$, Line 6 guarantees
    that $L(M) \leqslant k$ and there is an edge in $E(G)$ between any two successive vertices in $M$. So $M$ reported by Algorithm \ref{algo:dfs_on_index} is a
    path $P \in \mathcal{P}(s, t, k, G)$. Additionally, Algorithm \ref{algo:dfs_on_index} does not report duplicate results because $\mathcal{I}_t(v, b)$ returns a set of vertices.
    
    Next, we show that given any path $P \in \mathcal{P}(s, t, k, G)$, Algorithm \ref{algo:dfs_on_index} can find it. We prove this by induction. Without loss of generality, suppose that
    $P = (v_0 = s, v_1,..., v_{j} = t)$ where $1 \leqslant j \leqslant k$.
    Initially, $M = (v_0)$ holds (Line 1). Assume that $M = (v_0, v_1,...,v_i)$ where $1 \leqslant i \leqslant j - 1$ is constructed by the algorithm. We prove that the algorithm
    can generate $M' = (v_0, v_1,...v_i, v_{i + 1})$ from $M$. As $v_{i + 1}$ appears in $P$, $S(v_{i + 1}, t | G - \{s\}) \leqslant k - i - 1$.
    Based on the construction method of $\mathcal{I}$, $v_{i + 1}$ must belong to $\mathcal{I}_t(v_i, k - i - 1)$. Therefore, $M'$ can be generated from $M$ by the
    for loop (Lines 6-7). Thus, the algorithm can find $P$, and the proposition is proved.
\end{proof}

Algorithm \ref{algo:join_on_index} first evaluates $Q[0:i^*]$ and $Q[i^*:k]$, respectively. Then, it joins them to find final results.
The analysis in Section \ref{sec:join_based_model} guarantees the correctness of Algorithm \ref{algo:join_on_index}.

\begin{proposition} \label{prop:join_path}
    Algorithm \ref{algo:join_on_index} finds all $k$ hop-constrained paths $\mathcal{P}(s, t, k, G)$ from $s$ to $t$ in $G$.
\end{proposition}

Next, we prove Proposition \ref{prop:walk_correctness}, which shows that (1) Algorithm \ref{algo:dfs_on_index} without the check at Line 7 finds all $k$
hop-constrained walk $\mathcal{W}(s, t, k, G)$ from $s$ to $t$ in $G$; and (2) Given $M \in \widetilde{\mathcal{M}}_i$ where $0 \leqslant i \leqslant k$, $M$ must appear in a walk $W \in \mathcal{W}(s, t, k, G)$.

\begin{proof}
We can prove that Algorithm \ref{algo:dfs_on_index} after relaxation finds $\mathcal{W}(s, t, k, G)$ with the same proof method of Proposition \ref{prop:search_correctness}. Therefore,
    we omit a detailed proof for brevity, and focus on the second part of the proposition. We prove this by contradiction.
    
    Assume that there exist $M \in \widetilde{\mathcal{M}}_i$ that does not appear in any walk in $\mathcal{W}(s, t, k, G)$. Let $v$ and $v'$ represent the last two vertices in $M$. Then,
    $v' \neq t$ and $v' \in \mathcal{I}_t(v, k - L(M - \{v'\}) - 1)$, which means $S(v', t | G - \{s\}) \leqslant k - L(M - \{v'\}) - 1$. Let $P'$ denote the shortest path from $v'$
    to $t$ in $G - \{s\}$. We construct a walk $W$ by concatenating $M$ and $P'$. $L(W) = L(M) + S(v', t | G - \{s\}) \leqslant k$. Then, $W$ belongs to $\mathcal{W}(s, t, k, G)$,
    which contradicts the assumption. Thus, the proposition is proved.
\end{proof}

We can prove the correctness of Proposition \ref{prop:join_walk} with the same method as that in Proposition \ref{prop:walk_correctness}. For brevity, we omit the details.

\SUN{\section{Comparison with BC-DFS/BC-JOIN \cite{peng2019towards}}}

\SUN{\textbf{Comparison with BC-DFS.} The barrier optimization technique in BC-DFS \cite{peng2019towards} extracts a subgraph $G'$ of $G$ based on the rule:
if $v$ belongs to a result, then $S(s, v|G) + S(v, t|G) <= k$. Then, it enumerates all results on $G'$ with Algorithm \ref{algo:basic_method}.
In contrast, we construct a light-weight index based on Proposition \ref{prop:prune_vertex}, which puts a constraint on which vertices appear at
position $i$ of a result, and finds all results with the assistance of the index, which can efficiently get $v'$ in $N(v)$ such that $S(v', t| G - \{s\}) <= b$.
The index accelerates the enumeration because the index reduces the number of edges accessed at each step and eliminates the distance check (see the difference
between Algorithms \ref{algo:basic_method} and \ref{algo:dfs_on_index}). Benefiting from the index, Algorithm \ref{algo:dfs_on_index}
achieves a good time complexity ($O(k \times \delta_W)$). The following is an example demonstrating the difference.} 

\SUN{
\begin{example}
    Given $G$ and $q$ in Figure \ref{fig:example_graphs}, $G'$
    contains all vertices except $v_7$ based on the barrier optimization technique \cite{peng2019towards}.
    Suppose that $M = (s, v_0, v_1)$. Algorithm \ref{algo:basic_method} loops over $v_2$ and $v_3$ and checks
    their distances to $t$, whereas Algorithm \ref{algo:dfs_on_index} directly gets $v_2$ from the index and continues the search.
\end{example}
}

Both Algorithms \ref{algo:basic_method} and \ref{algo:dfs_on_index} adopt the backtracking search to enumeration all results. The exploration process can be
viewed as conducting a depth-first search in a search tree. Suppose that the average cost of expanding a node in the search tree is $\alpha$ and the search tree has $\beta$ nodes. The
cost of exploring the search tree is $T = \alpha \times \beta$. Existing algorithms \cite{peng2019towards,grossi2018efficient,rizzi2014efficiently} directly traverse on
$G$ to enumerate results. Given the last vertex $v$ of a partial result $M$, they visit all neighbors $v'$ of $v$ in $G$, check whether $S(v', t| G - M) \leqslant k - L(M) - 1$ and
dynamically update $S(v, t|G - M)$ at each step. In contrast, our algorithm explores the index constructed in a preprocessing step to find results. At each step, we only consider the neighbors $v'$ 
of $v$ such that $S(v', t|G - \{s\}) \leqslant k - L(M) - 1$ with the assistance of the index, but eliminate the explicit distance verification and complex filtering techniques. In other words,
the difference between existing algorithms and ours is the trade-off between $\alpha$ and $\beta$.
Exiting algorithms optimize the search by decreasing $\beta$ at the cost of increasing $\alpha$, while we make the \emph{Search} procedure simple and efficient to reduce $\alpha$, but can
generate more invalid partial results. However, it is challenging to make a direct comparison of our algorithm and existing algorithms in terms of
the time complexity ($O(k \times \delta_W)$ versus. $O(k \times |E(G)| \times \delta_P)$) because the value of $\frac{\delta_P}{\delta_W}$ depends on each query instance. Instead,
we conduct extensive experiments with a variety of real-world graphs to compare them (see Figure \ref{fig:overall_detailed_metrics}). Our experiment
results show that our method significantly outperforms BC-DFS.

\SUN{\textbf{Comparison with BC-JOIN.} Join is a common methodology for graph queries \cite{aberger2017emptyheaded,mhedhbi2019optimizing}.
The key factors leading to the performance differences are strategies reducing partial results, e.g., invalid elements filtering and
join order optimization. BC-JOIN, which is a join-oriented algorithm proposed in \cite{peng2019towards}, first computes
the set $V$ of vertices appearing in the middle position ($\lceil \frac{k}{2} \rceil$) of paths in $\mathcal{P}(s, t, k, G)$. Then, it computes the
paths from $s$ to $v \in V$, and that from $v \in V$ to $t$ with Algorithm \ref{algo:basic_method}. Finally, it joins the intermediate results to
find the final results. In contrast, our algorithm first adopts a light-weight index to prune invalid candidates and then performs the join
on the basis of the index. Moreover, we design a cost-based query optimizer with the preliminary and full-fledged estimators to process queries
with variant execution time.}

\SUN{\section{Extension of Our Algorithms}}

\SUN{Our method on the HcPE problem can be easily extended to support variant constraints to capture the complexities of real-world applications.}

\SUN{\textbf{Constraints on Predicates.} The first kind of constraints is based on the predicate $f_p(e)$
on attributes such as weights and labels of edges (or vertices) where $e$ is an edge and $f_p(.)$ is a
user-defined boolean function. For example, if we focus on large flow transactions by recently created companies \cite{force2013fatf},
then we can define a predicate based on edge properties. In addition to the length constraint, the predicate requires
that each edge $e$ in a result path satisfies the conditions in $f_p(.)$, i.e., the return value of $f_p(e)$ is true.}

\SUN{Given a query $q$ on a graph $G$ and a predicate $f_p(.)$, we first generate a subgraph $G'$ of $G$ by
applying the predicate on $G$ to filter invalid edges. Then, we evaluate $q$ on $G'$ with PathEnum to
find the results. The filtering phase guarantees that each edge in a result path meets the constraints
defined by $f_p(.)$. In practice, we do not need to materialize the subgraph $G'$. Instead, we can
conduct the filtering when computing the distance between vertices and $s,t$ with BFS in
the index building phase of PathEnum. Then, the edges in the index meet the requirement in the predicate.
}

\SUN{The constraints on predicates is the same as Definition 2 in Section 2.1 of \cite{qiu2018real}, which requires
the paths to satisfy both the length and attribute constraints. Therefore, PathEnum can support the
scenario in the second motivation example in our paper, which comes from \cite{qiu2018real}.}

\SUN{\textbf{Constraints on Accumulative Values.} Let $\bigoplus$ denote a binary operation having commutative law
and associative law. $\alpha(e)$ represents a value associating with an edge $e$, for example, the edge weight.
Given a result path $P$, the accumulative value $\beta$ of $P$ is equal to $\bigoplus_{e \in P} \alpha(e)$.
The constraints on accumulative values require that the accumulative value $\beta$ of $P$
satisfies a user-defined boolean function $f_a(.)$, i.e., $f_a(\beta)$ returns true in addition
to the length constraint. For example, if we require that the sum of the transaction risky factors along the path is above a threshold \cite{force2013fatf},
then we can define $\bigoplus$ as a plus operation.}

\begin{algorithm}[h]
    \footnotesize
	\caption{Depth-First Search on Index with Constraints on Accumulative Values}
	\label{algo:constraints_on_accumulative_values}
	\SetKwFunction{Search}{Search}
	\SetKwProg{proc}{Procedure}{}{}
	 \KwIn{two distinct vertices $s, t$, hop constraint $k$, index $\mathcal{I}$\;}
	 \KwOut{all $k$ hop-constrained paths from $s$ to $t$ that satisfy constraints on accumulative values\;}
	 $M \leftarrow (s)$\;
	 $\beta \leftarrow $ an initial value\;
	 \Search{$t, k, M, \mathcal{I}, \beta$}\;
	 \proc{\Search{$t, k, M, \mathcal{I}, \beta$}}{
	      $v \leftarrow$ the last vertex in $M$\;
	      \lIf{$v = t$ and $f_{a}(\beta)$ is true}{$emit(M)$, \KwRet}
	      \ForEach{$v' \in \mathcal{I}_t(v, k - L(M) - 1)$}{
	            \lIf{$v' \notin M$}{\Search{$t, k, M \cup \{v'\}, \mathcal{I}, \beta \bigoplus \alpha(e(v,v'))$}}
	      }
	 }
\end{algorithm}

\SUN{Algorithm \ref{algo:constraints_on_accumulative_values} extends the depth-first search on the index
to support constraints on accumulative values. Line 2 assigns an initial value based on
the binary operation $\bigoplus$. For example, if $\bigoplus$ is the sum operation, then we set $\beta$
to 0, while setting it to 1 if $\bigoplus$ is the multiply operation. If the accumulative value of $P$
satisfies the constraint $f_a(.)$, then Line 6 outputs the path. When adding a new vertex to the partial
result $M$, Line 8 updates the accumulative value with the binary operation. In some cases, we can check
whether the accumulative value cannot lead to a valid solution and can be pruned when extending the partial results to reduce search space; for example,
if the constraint is that the accumulative edge weight must be below a threshold, and all edge weights are nonnegative.
However, if edge weights are negative, we cannot add the check when extending $M$ because the accumulative value is not monotonic.}

\SUN{We can extend the join on index with the same method as Algorithm \ref{algo:constraints_on_accumulative_values}
because the binary operation satisfies commutative law and associative law and the accumulative value is
independent of the sequence of performing the binary operations. We omit the detail for brevity.}

\SUN{\textbf{Constraints on a Sequence of Actions.} $l(e)$ is the edge label, which represents an action.
Given a sequence of actions, the constraint requires that the edge label sequence along a result
path satisfies the given sequence. For example, if we require that the flow of transactions pass through at least
two high-risk counties \cite{force2013fatf}, then we can define a sequence schema.}

\SUN{We can model the given sequence of actions as an automata where nodes are states and edges
define the transition relationship among states based on actions. The automata can be represented
by a matrix $A$. Given a state $a$ and an action (i.e., label) $l(e)$, $A[a][l(e)]$ returns
the next state. If the action is invalid for the given state, then $A[a][l(e)]$ returns \emph{null}.
Algorithm \ref{algo:constraints_on_sequence_actions} presents the depth-first search on index
with constraints on sequences of actions, which is based on the automata $A$. Line 2 sets
the initial state as the start state in the automata $A$. When trying to extend $M$ with a vertex,
Line 8 retrieves the next state $a'$ based on $A$. If $a'$ is \emph{null}, then skip the vertex.
Otherwise, move to the state $a'$ and continue the search. Line 6 outputs the result path if
the vertex is $t$ and the transition reaches an end state in $A$.}

\begin{algorithm}[h]
    \footnotesize
	\caption{Depth-First Search on Index with Constraints on Sequences of actions}
	\label{algo:constraints_on_sequence_actions}
	\SetKwFunction{Search}{Search}
	\SetKwProg{proc}{Procedure}{}{}
	 \KwIn{two distinct vertices $s, t$, hop constraint $k$, index $\mathcal{I}$\;}
	 \KwOut{all $k$ hop-constrained paths from $s$ to $t$ that satisfy constraints on sequences of Actions\;}
	 $M \leftarrow (s)$\;
	 $a \leftarrow $ the start state in the automata $A$\;
	 \Search{$t, k, M, \mathcal{I}, a$}\;
	 \proc{\Search{$t, k, M, \mathcal{I}, a$}}{
	      $v \leftarrow$ the last vertex in $M$\;
	      \lIf{$v = t$ and $a$ is an end state}{$emit(M)$, \KwRet}
	      \ForEach{$v' \in \mathcal{I}_t(v, k - L(M) - 1)$}{
	            $a' \leftarrow A[a][l(e(v, v'))]$\;
	            \If{$v' \notin M$ and $a' \neq null$}{
	                \Search{$t, k, M \cup \{v'\}, \mathcal{I}, a'$}\;
	            }
	      }
	 }
\end{algorithm}

\SUN{As the join on index can evaluate the query with different orders, we use the automata to check whether
a path $P$ can meet the constraint after $P$ is generated by the join on index method. Therefore, the
depth-first search method can terminate the invalid search path at an earlier stage than the join method.}

\SUN{In summary, our method on the HcPE problem can be easily extended to support variant constraints. Adding the constraints can accelerate the query
because the constraints can reduce the size of the search space.}

\section{Supplement Experiment Results}

\textbf{Varying Hop Constraint $k$.} Figure \ref{fig:vary_k_query_time} shows the experiment results on query time with $k$
varied. When $k = 8$, BC-JOIN runs out of memory on \emph{ep} due to the maintenance of a large amount of intermediate results.
Thus, we omit the results of BC-JOIN on this case. As shown in the figures, PathEnum significantly outperforms BC-DFS and BC-JOIN.
\SUN{The preliminary estimation makes the overhead of PathEnum very small. Its time complexity is $O(k^2)$,
and it takes less than 0.01 ms in experiments. Because there are a small number of results on \emph{gg} with k varied from 3 to 6,
PathEnum directly invokes IDX-DFS after the preliminary estimation, the cost of which is negligible. Thus, PathEnum is close to IDX-DFS
despite that IDX-DFS dominates IDX-JOIN. If the join order optimization is executed, the preliminary estimation ensures that the optimization
time accounts for a small portion of query time. The benefit of the optimization can offset the overhead, and PathEnum can beat both IDX-DFS and 
IDX-JOIN.}

\begin{figure}[ht]\small
    \setlength{\abovecaptionskip}{0pt}
    \setlength{\belowcaptionskip}{0pt}
    \captionsetup[subfigure]{aboveskip=0pt,belowskip=0pt}
    \centering
    \begin{subfigure}[t]{0.23\textwidth}
        \centering
        \includegraphics[scale=0.23]{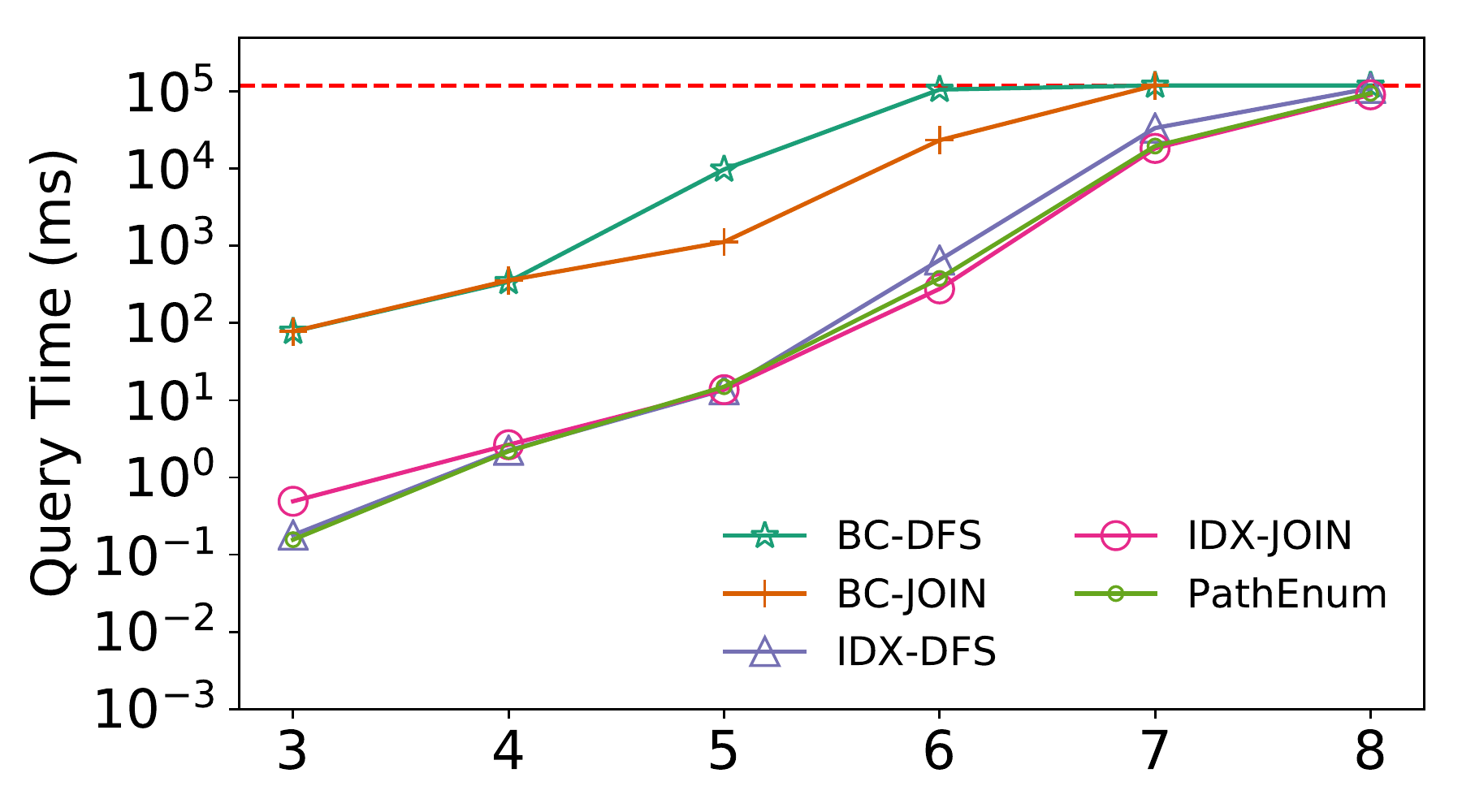}
        \caption{\emph{ep}.}
        \label{fig:vary_k_query_time_soc}
    \end{subfigure}
    \begin{subfigure}[t]{0.23\textwidth}
        \centering
        \includegraphics[scale=0.23]{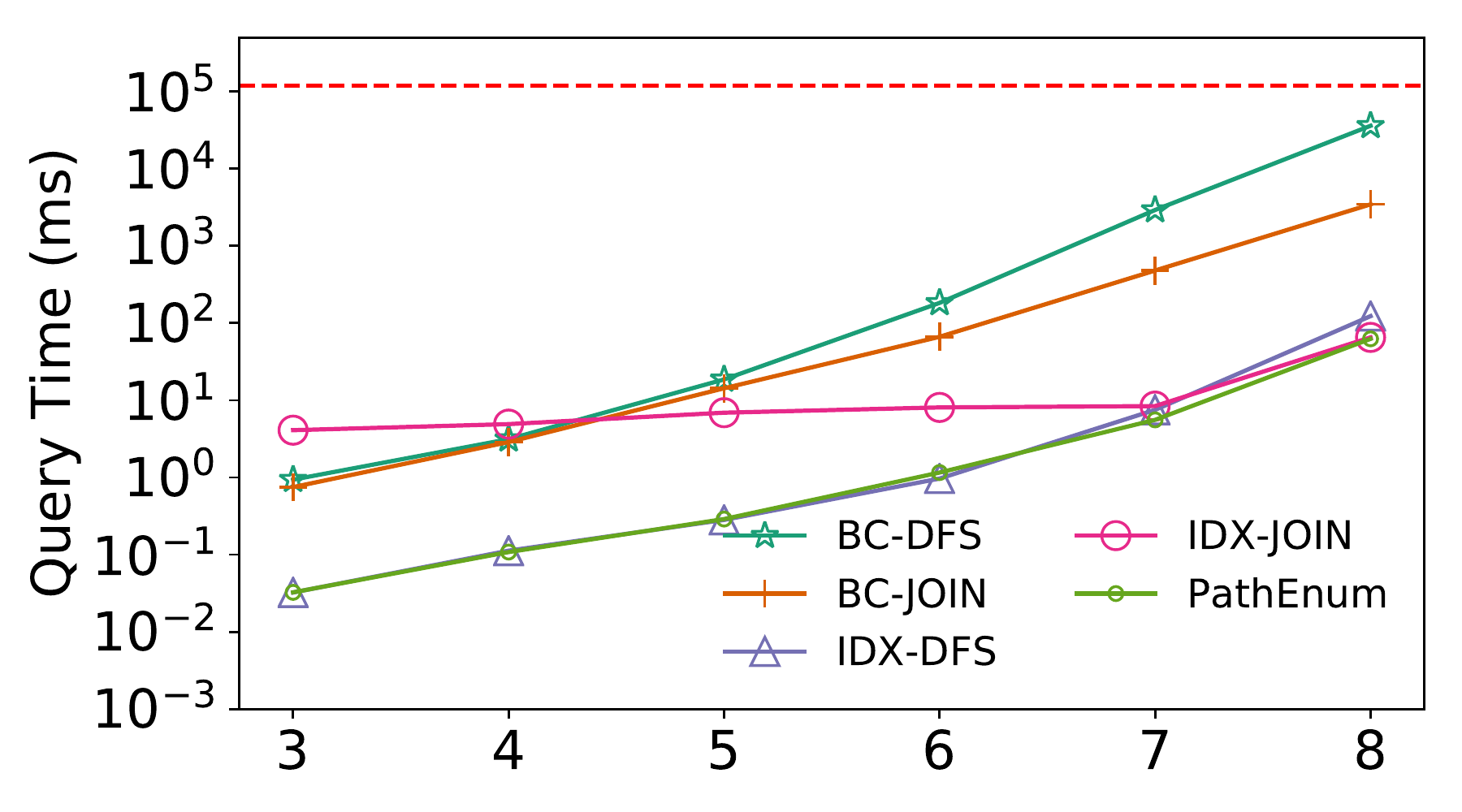}
        \caption{\emph{gg}.}
        \label{fig:vary_k_query_time_google}
    \end{subfigure}
    \caption{Comparison of query time with $k$ varied.}
    \label{fig:vary_k_query_time}
\end{figure}

Figure \ref{fig:vary_k_throughput} shows the experiment results
on throughput with $k$ varied.  Although PathEnum runs out of time on most queries against \emph{ep} when $k=8$,
it has much higher throughput than BC-DFS and BC-JOIN. The throughput of our three algorithms
increases on \emph{ep} with $k$ varied from 3 to 6, but keeps steady with $k$ increasing
from $6$ to $8$ because the time spent on building the index and optimizing join orders is neglected compared with
the time spent on enumeration when $k$ is large. Moreover, the results imply that the value of $k$
has little impact on the enumeration speed of our algorithms. In contrast, the throughput of BC-DFS decreases with $k$ varied
from 5 to 8 because the increasing of $k$ results in more overhead of dynamically updating the distance of vertices to $t$.

\begin{figure}[ht]\small
    \setlength{\abovecaptionskip}{0pt}
    \setlength{\belowcaptionskip}{0pt}
    \captionsetup[subfigure]{aboveskip=0pt,belowskip=0pt}
    \centering
    \begin{subfigure}[t]{0.23\textwidth}
        \centering
        \includegraphics[scale=0.23]{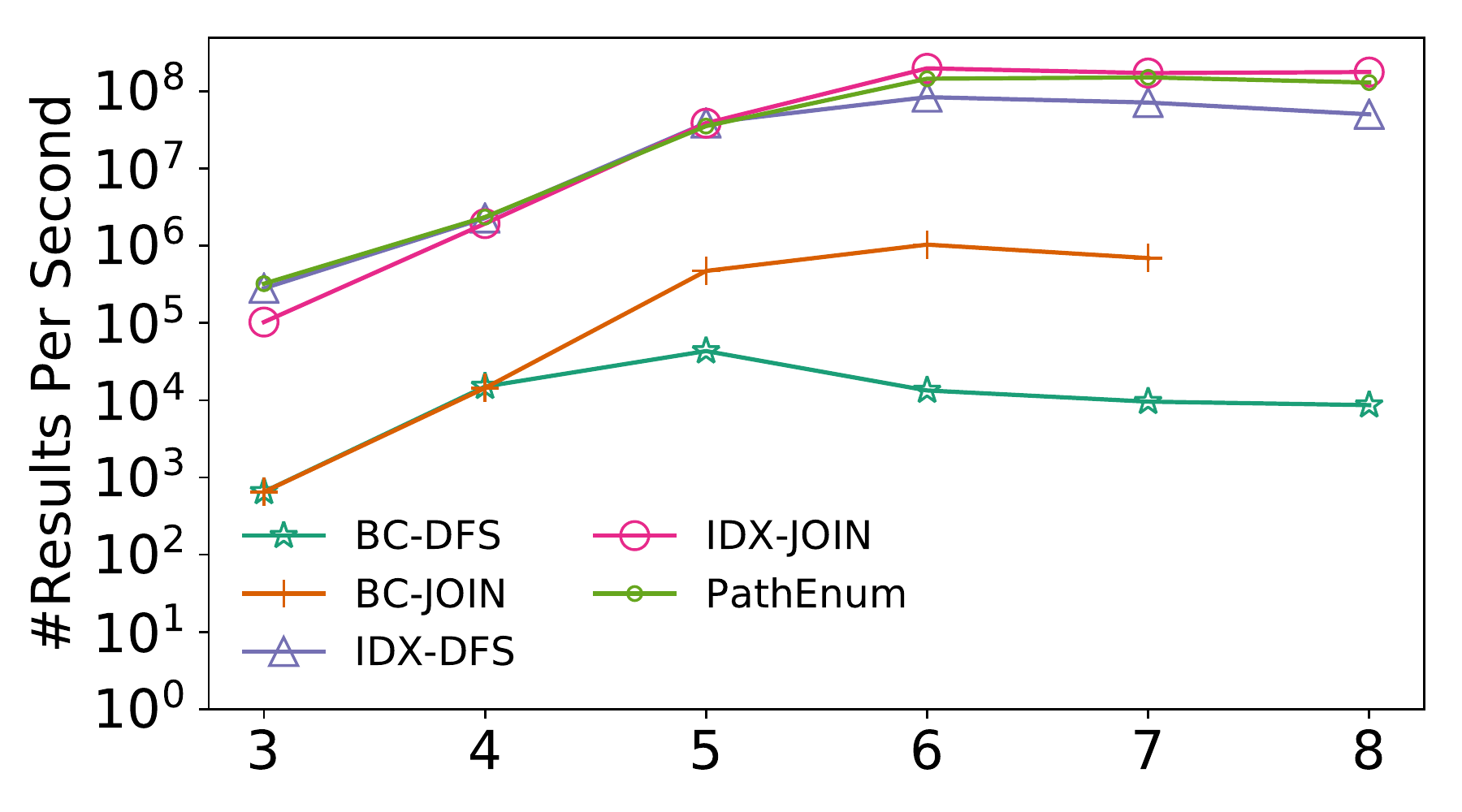}
        \caption{\emph{ep}.}
        \label{fig:vary_k_throughput_soc}
    \end{subfigure}
    \begin{subfigure}[t]{0.23\textwidth}
        \centering
        \includegraphics[scale=0.23]{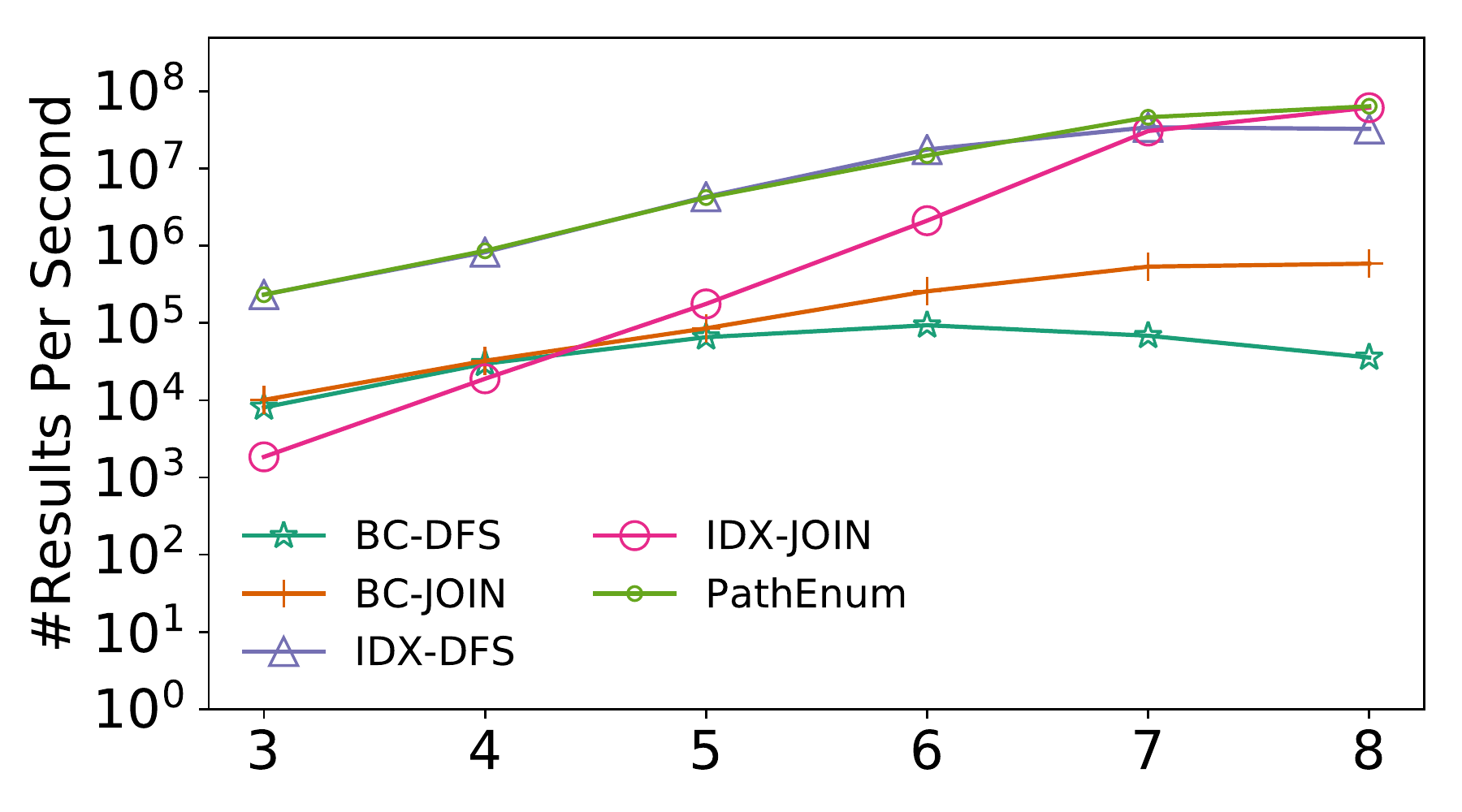}
        \caption{\emph{gg}.}
        \label{fig:vary_k_throughput_google}
    \end{subfigure}
    \caption{Comparison of throughput with $k$ varied.}
    \label{fig:vary_k_throughput}
\end{figure}

Figure \ref{fig:vary_k_response_time} presents the response time of BC-DFS and IDX-DFS with $k$ varied. IDX-DFS outperforms
BC-DFS by up to two orders of magnitude. Additionally, the response time of IDX-DFS slightly increases with $k$ varied from
3 to 8, and the value is less than 20 ms. The results show that IDX-DFS can be applied to the online scenarios having
real-time constraint \cite{qiu2018real}.

\begin{figure}[ht]\small
    \setlength{\abovecaptionskip}{0pt}
    \setlength{\belowcaptionskip}{0pt}
    \captionsetup[subfigure]{aboveskip=0pt,belowskip=0pt}
    \centering
    \begin{subfigure}[t]{0.23\textwidth}
        \centering
        \includegraphics[scale=0.23]{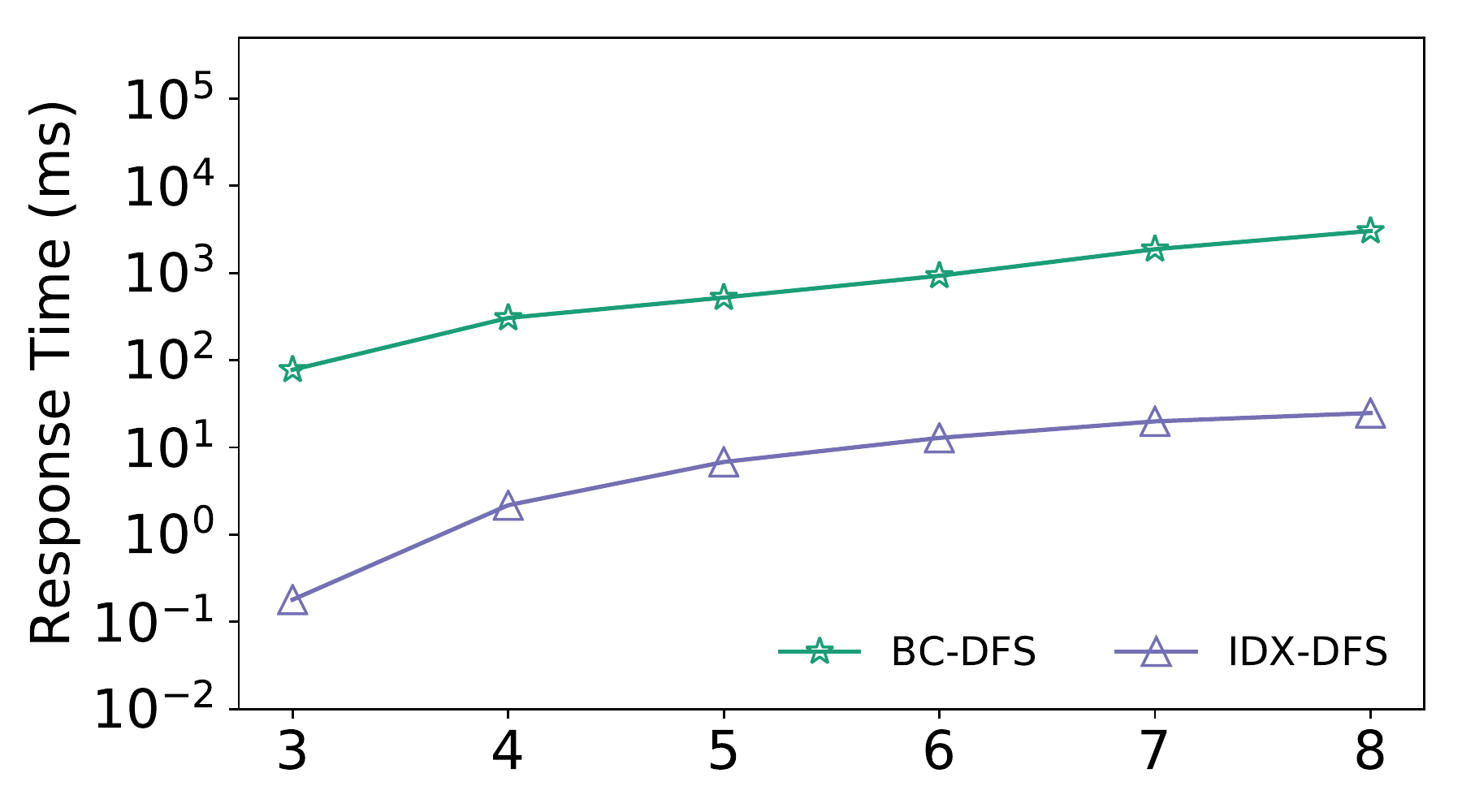}
        \caption{\emph{ep}.}
        \label{fig:vary_k_response_time_soc}
    \end{subfigure}
    \begin{subfigure}[t]{0.23\textwidth}
        \centering
        \includegraphics[scale=0.23]{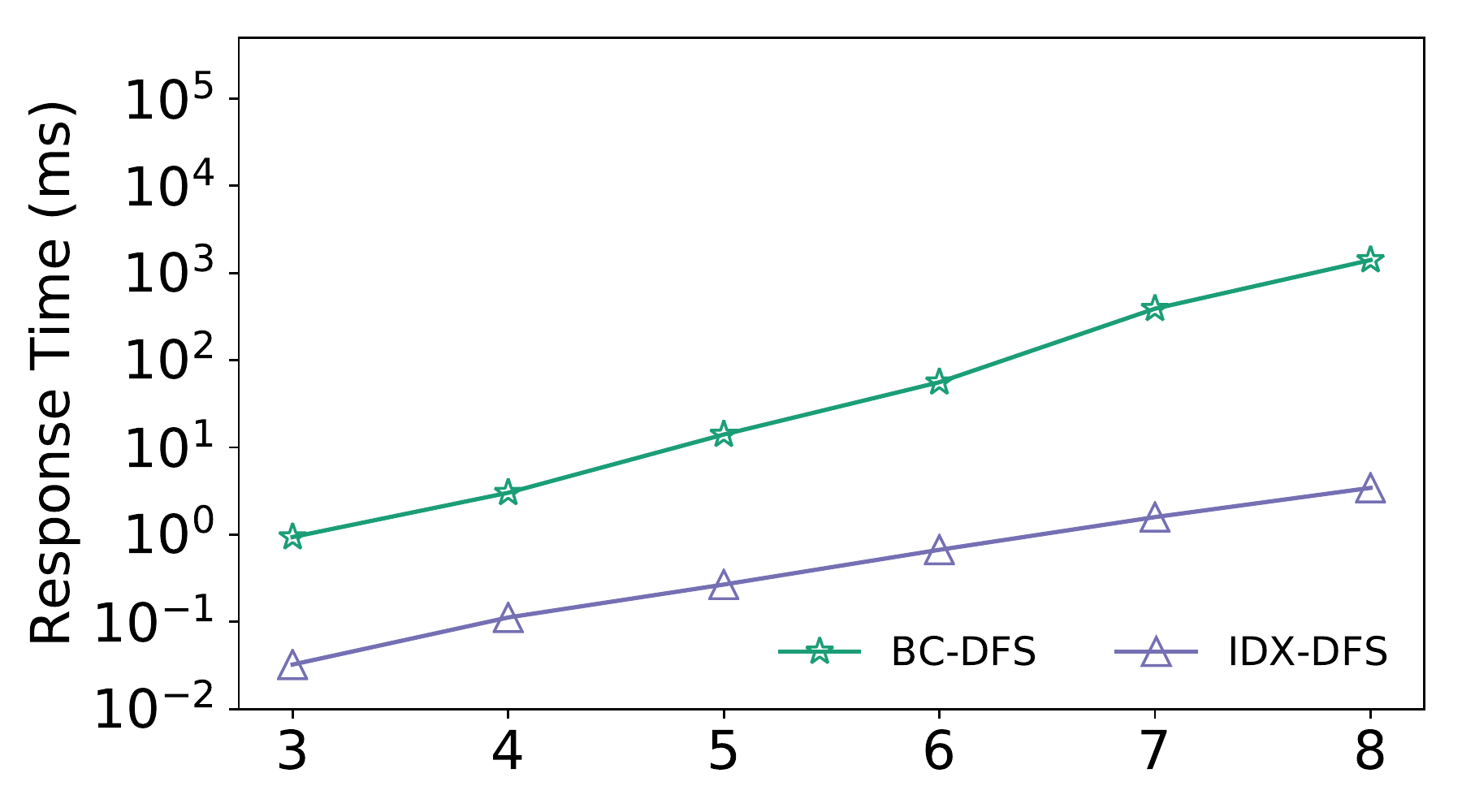}
        \caption{\emph{gg}.}
        \label{fig:vary_k_response_time_google}
    \end{subfigure}
    \caption{Comparison of response time with $k$ varied.}
    \label{fig:vary_k_response_time}
\end{figure}

\textbf{Individual Query Performance.} We demonstrate the cumulative distribution function of the query time in
Figure \ref{fig:overall_cdf} to examine the performance of completing individual query.
The query time of different queries varies greatly. IDX-JOIN performs better than IDX-DFS on \emph{ep}, but
worse on \emph{gg}. In contrast, PathEnum performs well on both of the two graphs. Nevertheless, our algorithms significantly
outperform BC-DFS and BC-JOIN. For example, BC-DFS and BC-JOIN run out of time on more than 80\% and 10\% queries on
\emph{ep}, respectively, while our algorithms complete all queries within around 10 seconds.

\begin{figure}[ht]\small
    \setlength{\abovecaptionskip}{0pt}
    \setlength{\belowcaptionskip}{0pt}
    \captionsetup[subfigure]{aboveskip=0pt,belowskip=0pt}
    \centering
    \begin{subfigure}[t]{0.23\textwidth}
        \centering
        \includegraphics[scale=0.23]{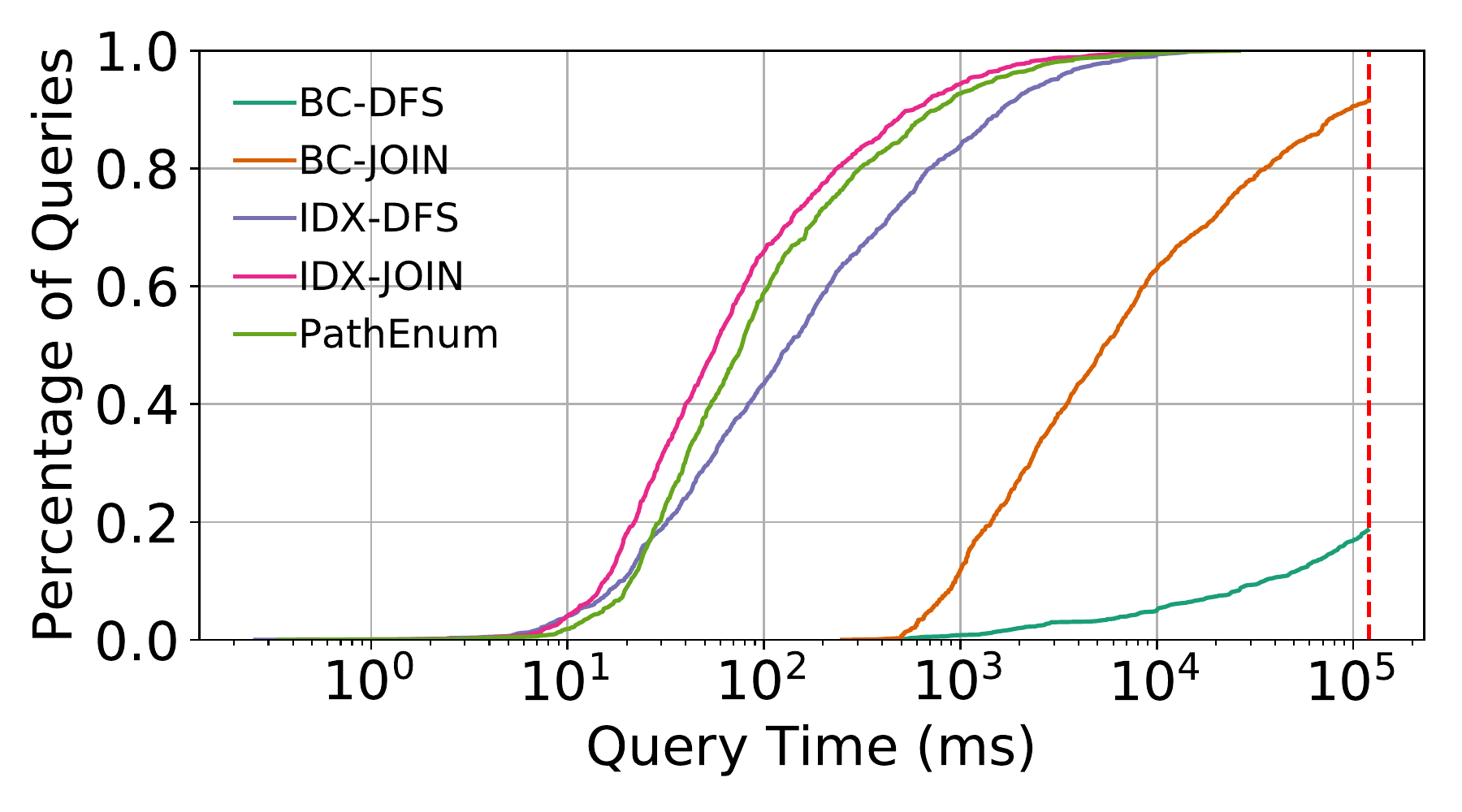}
        \caption{\emph{ep}.}
        \label{fig:cdf_socepinsion}
    \end{subfigure}
    \begin{subfigure}[t]{0.23\textwidth}
        \centering
        \includegraphics[scale=0.23]{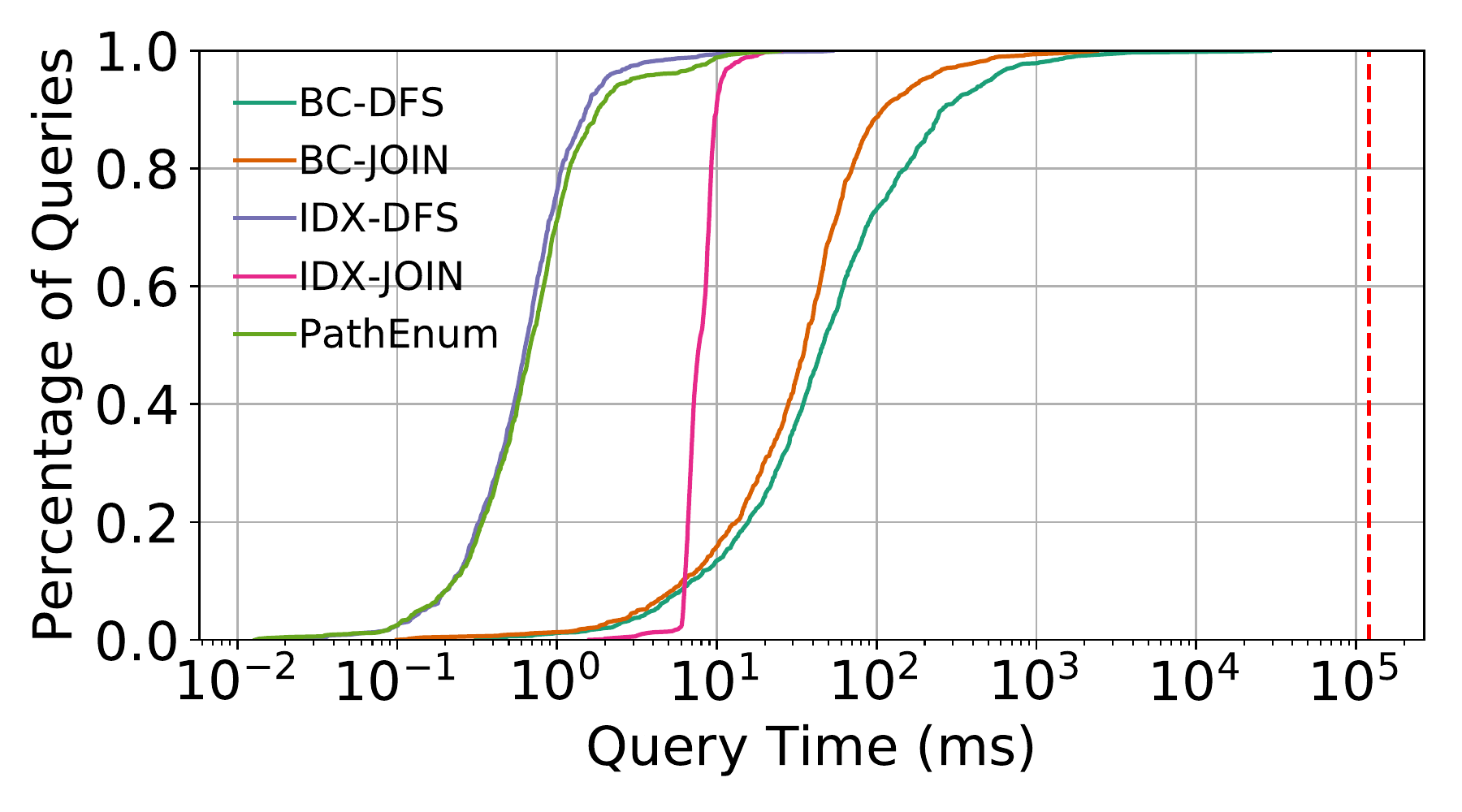}
        \caption{\emph{gg}.}
        \label{fig:cdf_webgoogle}
    \end{subfigure}
    \caption{Cumulative distribution of the query time.}
    \label{fig:overall_cdf}
\end{figure}

\textbf{Time Efficiency of Individual Technique.} Figure \ref{fig:line_breakdown} illustrates the execution time of each individual
technique with $k$ varied. "\emph{Index construction}" denotes the time spent on building the index (Algorithm
\ref{algo:build_index}). "\emph{Optimization}" represents the time spent on generating join orders
(Algorithm \ref{algo:generate_join_order}). "\emph{DFS}" and "\emph{JOIN}" denotes the enumeration time of
Algorithms \ref{algo:dfs_on_index} and \ref{algo:join_on_index}, respectively. Therefore,
the query time of IDX-DFS is the sum of \emph{ Index construction} and \emph{DFS}, while
the query time of IDX-JOIN is the sum of \emph{ Index construction}, \emph{Optimization} and \emph{JOIN}. Additionally,
we report the time of computing the distance of each vertex to $s,t$, which is denoted by \emph{BFS}. \emph{BFS} is included
in \emph{Index construction}. We omit the execution time of the preliminary cardinality estimator because its value is negligible (less than 0.01 ms).

As shown in the figure, \emph{BFS} generally dominates the execution time of Algorithm \ref{algo:build_index}. The cost
of optimizing join orders can be greater than the enumeration for the short running queries. The DFS on the index
runs faster than the join on the index when $k$ is small, but slower when $k$ is large. Nevertheless, the absolute value
of \emph{Index construction} and \emph{Optimization} is very small, which demonstrates the efficiency of our index construction
and query optimization. More importantly, the benefits of those operations is far higher than the overhead, showing significant performance speedup even for short queries over existing approaches (BC-DFS/BC-Join).

\begin{figure}[ht]\small
    \setlength{\abovecaptionskip}{0pt}
    \setlength{\belowcaptionskip}{0pt}
    \captionsetup[subfigure]{aboveskip=0pt,belowskip=0pt}
    \centering
    \begin{subfigure}[t]{0.23\textwidth}
        \centering
        \includegraphics[scale=0.23]{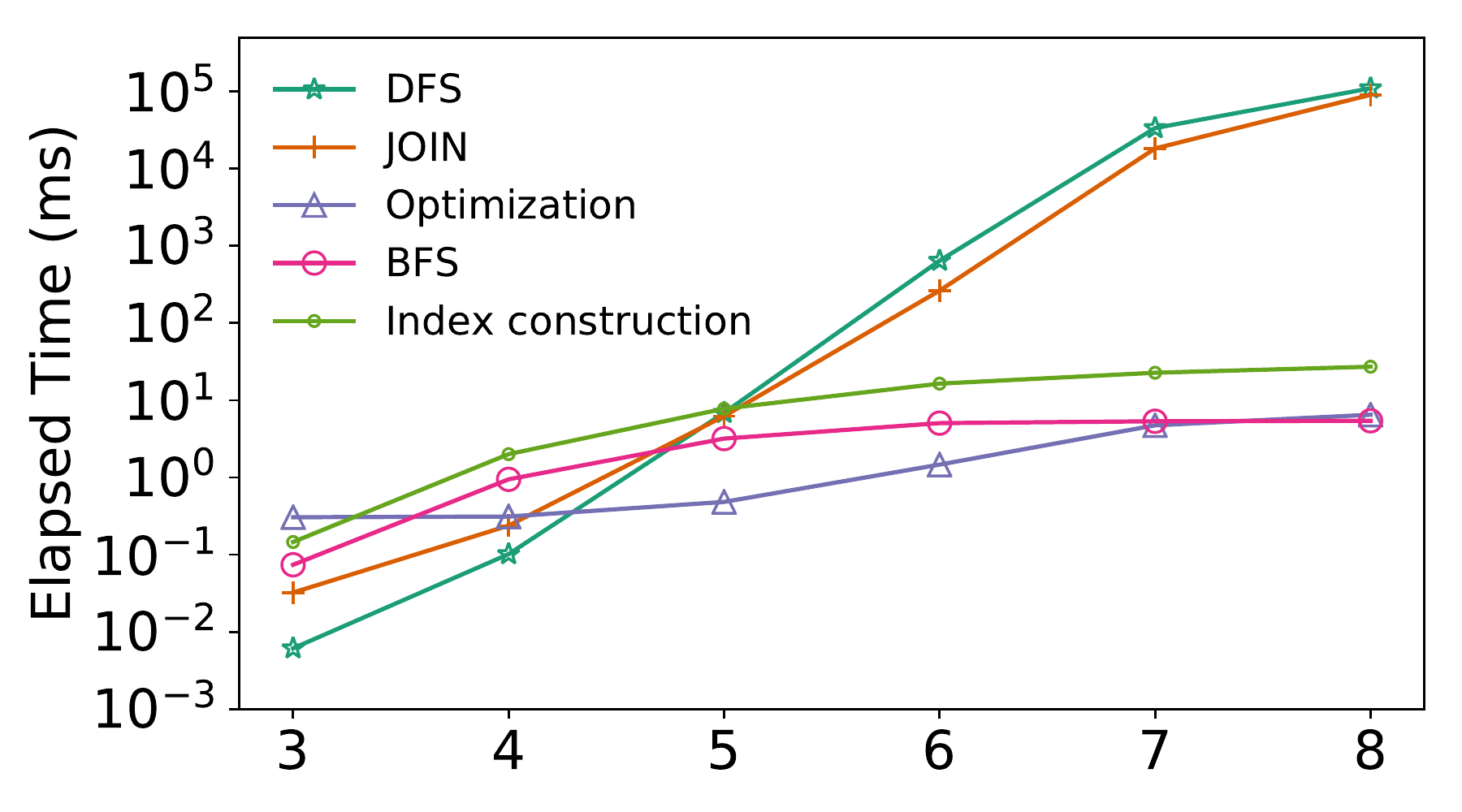}
        \caption{\emph{ep}.}
        \label{fig:line_breakdown_ep}
    \end{subfigure}
    \begin{subfigure}[t]{0.23\textwidth}
        \centering
        \includegraphics[scale=0.23]{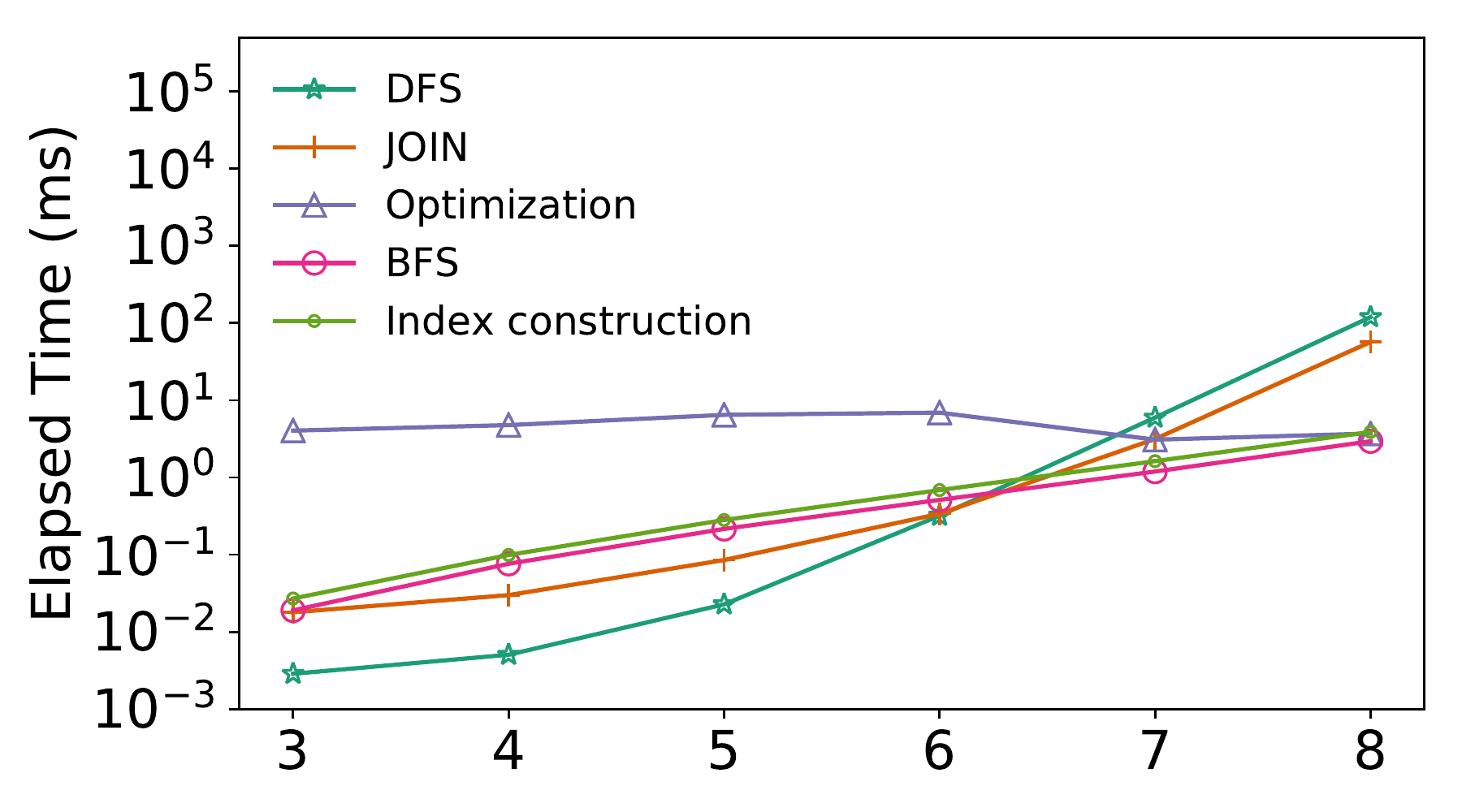}
        \caption{\emph{gg}.}
        \label{fig:line_breakdown_gg}
    \end{subfigure}
    \caption{The execution time of each individual technique with $k$ varied.}
    \label{fig:line_breakdown}
\end{figure}

\textbf{Cardinality Estimation.} We evaluate the preliminary and full-fledged cardinality estimators by comparing the number
of results estimated to the actual value.  
Figure \ref{fig:estimation_vary_k} illustrates the results with $k$ varied.
The gap between our estimation and the actual value widens with $k$ varied from 3 to 8 because it is more challenging to
estimate queries with long paths and the estimation error propagates with the increase of $k$. We omit the result on
\emph{ep} when $k = 8$, because more than 50\% of queries run out of time, and we cannot get their actual number of results of
these queries.

\begin{figure}[ht]\small
    \setlength{\abovecaptionskip}{0pt}
    \setlength{\belowcaptionskip}{0pt}
    \captionsetup[subfigure]{aboveskip=0pt,belowskip=0pt}
    \centering
    \begin{subfigure}[t]{0.23\textwidth}
        \centering
        \includegraphics[scale=0.23]{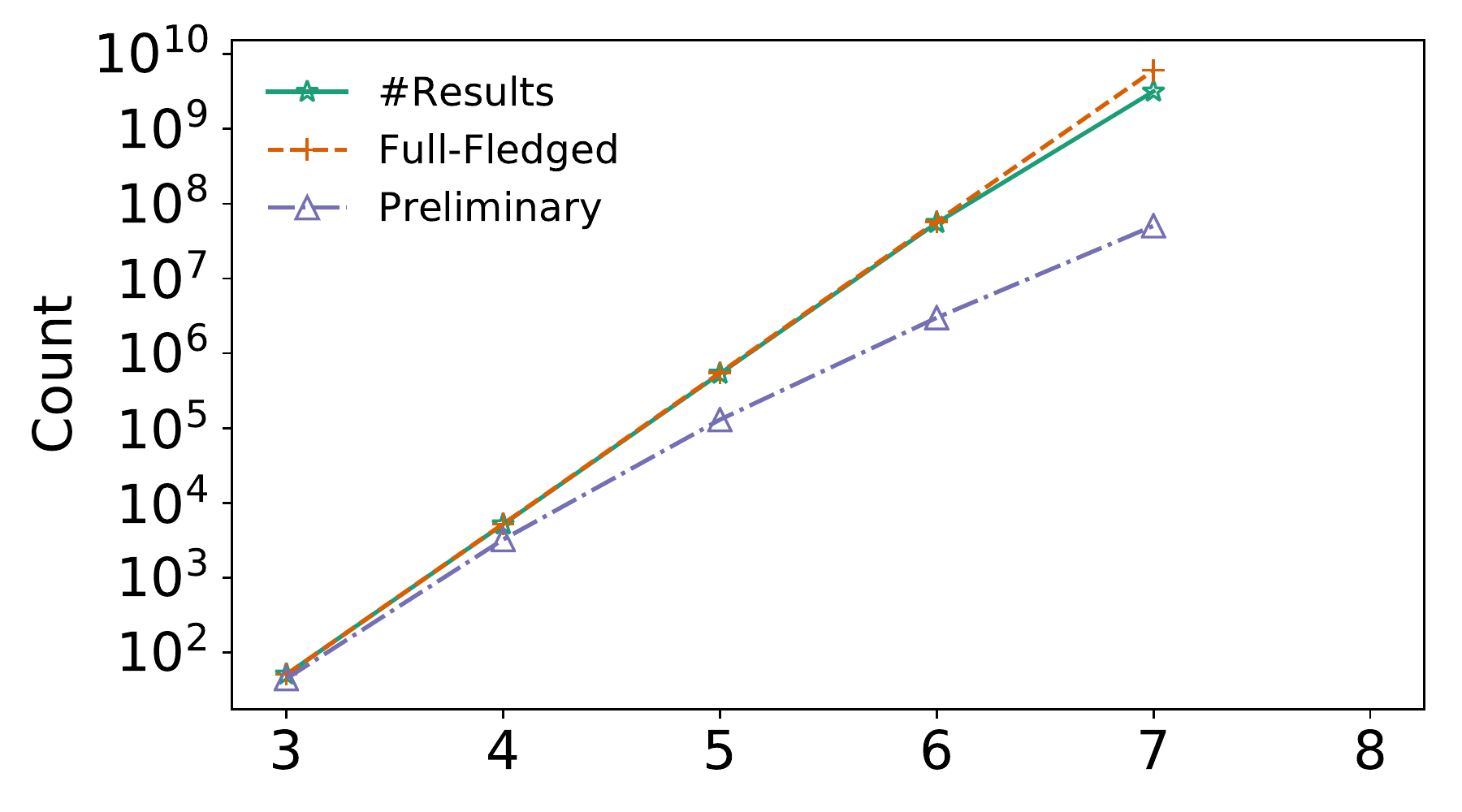}
        \caption{\emph{ep}.}
        \label{fig:estimation_vary_k_ep}
    \end{subfigure}
    \begin{subfigure}[t]{0.23\textwidth}
        \centering
        \includegraphics[scale=0.23]{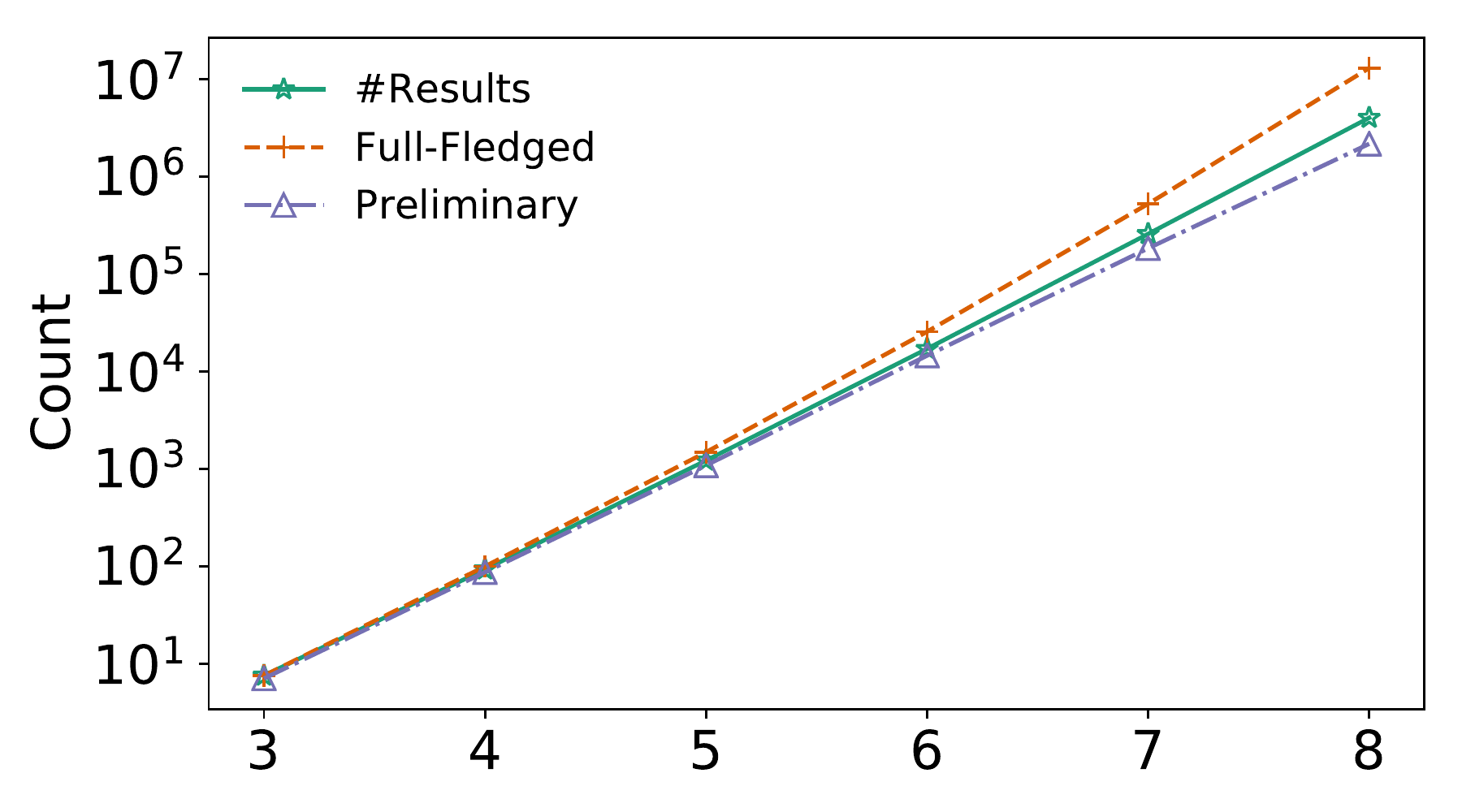}
        \caption{\emph{gg}.}
        \label{fig:estimation_vary_k_gg}
    \end{subfigure}
    \caption{Cardinality estimation with $k$ varied.}
    \label{fig:estimation_vary_k}
\end{figure}



\SUN{\textbf{Summary}. First, the light-weight index can significantly accelerate the enumeration
by reducing the number of edges accessed and eliminating the distance check at each step. Moreover, as the index is query dependent,
which prunes many invalid edges, it is expected to provide more accurate statistics for the query optimizer than that of the original graph.
Second, a cost-based query optimizer is essential for reducing the cost of evaluating HcPE queries. Our query optimizer that
closely works with the light-weight index is effective. Furthermore, the optimizer with cardinality estimation methods
having different time complexities is necessary because the query time of different queries varies greatly even on the same graph.
Third, our solution has a high enumeration speed (i.e., throughput) even on the graphs with billions of edges. As the index
construction is efficient, our solution can directly handle dynamic graphs in general. However, when the graphs
are very large, our method can have a long response time because it builds the index from the scratch.
For example, the index construction takes tens of seconds on \emph{tm} with billions of edges in Figure \ref{fig:scalability_execution_time}.
Fourth, the number of results has a significant impact on HcPE query efficiency. Some queries have a long running time because
they have a huge number of results.}